\documentclass[a4paper,twocolumn,10pt,accepted=2020-10-21]{quantumarticle}
\pdfoutput=1
\usepackage{amsmath, amssymb, amsthm}
\usepackage{mathtools}
\usepackage{color}
\usepackage{braket}
\usepackage{comment}
\usepackage{bbm}
\usepackage[breaklinks=true]{hyperref}
\usepackage{url}
\usepackage[utf8]{inputenc}
\usepackage[english]{babel}
\usepackage[T1]{fontenc}
\usepackage[numbers,sort&compress]{natbib}

\DeclareMathOperator{\tr}{Tr}

\DeclareMathOperator*{\argmin}{argmin}
\DeclareMathOperator{\id}{id}

\allowdisplaybreaks[4]

\theoremstyle{plain}
\newtheorem{theorem}{Theorem}
\newtheorem{lemma}[theorem]{Lemma}
\newtheorem{corollary}[theorem]{Corollary}
\newtheorem{proposition}[theorem]{Proposition}
\theoremstyle{definition}
\newtheorem{definition}[theorem]{Definition}
\newtheorem{example}{Example}

\theoremstyle{remark}
\newtheorem{conjecture}{Conjecture}
\newtheorem{remark}[conjecture]{Remark}

\begin{document}

\title{General Quantum Resource Theories: Distillation, Formation and Consistent Resource Measures}

\author{Kohdai Kuroiwa}
\email{kkuroiwa@uwaterloo.ca}
\affiliation{Institute for Quantum Computing and Department of Physics and Astronomy, University of Waterloo, 200 University Avenue West, Waterloo, Ontario, N2L 3G1, Canada}
\affiliation{Department of Applied Physics, Graduate School of Engineering, The University of Tokyo, 7--3--1 Hongo, Bunkyo-ku, Tokyo 113--8656, Japan}
\author{Hayata Yamasaki}
\email{hayata.yamasaki@gmail.com}
\affiliation{Photon Science Center, Graduate School of Engineering, The University of Tokyo, 7--3--1 Hongo, Bunkyo-ku, Tokyo 113--8656, Japan}
\affiliation{Institute for Quantum Optics and Quantum Information --- IQOQI Vienna, Austrian Academy of Sciences, Boltzmanngasse 3, 1090 Vienna, Austria}
\maketitle

\begin{abstract}
Quantum resource theories (QRTs) provide a unified theoretical framework for understanding inherent quantum-mechanical properties that serve as resources in quantum information processing, but resources motivated by physics may possess structure whose analysis is mathematically intractable, such as non-uniqueness of maximally resourceful states, lack of convexity, and infinite dimension. 
We investigate state conversion and resource measures in general QRTs under minimal assumptions to figure out universal properties of physically motivated quantum resources that may have such mathematical structure whose analysis is intractable. 
In the general setting, we prove the existence of maximally resourceful states in one-shot state conversion. 
Also analyzing asymptotic state conversion, we discover \textit{catalytic replication} of quantum resources, where a resource state is infinitely replicable by free operations. 
In QRTs without assuming the uniqueness of maximally resourceful states, we formulate the tasks of distillation and formation of quantum resources, and introduce distillable resource and resource cost based on the distillation and the formation, respectively. 
Furthermore, we introduce \textit{consistent resource measures} that quantify the amount of quantum resources without contradicting the rate of state conversion even in QRTs with non-unique maximally resourceful states. 
Progressing beyond the previous work showing a uniqueness theorem for additive resource measures, we prove the corresponding uniqueness inequality for the consistent resource measures; that is, consistent resource measures of a quantum state take values between the distillable resource and the resource cost of the state. 
These formulations and results establish a foundation of QRTs applicable in a unified way to physically motivated quantum resources whose analysis can be mathematically intractable.
\end{abstract}

\section{Introduction}\label{Section1}
Advantages in quantum information processing compared to conventional classical information processing arise from various inherent properties of quantum states.
A framework for systematically investigating quantum-mechanical properties is essential for better understandings of quantum mechanics and quantum information processing.
Quantum resource theories (QRTs)~\cite{Chitambar2018} give such a framework, in which the quantum properties are regarded as resources for overcoming restrictions on operations on quantum systems; 
especially, manipulation and quantification of resources are integral parts of QRTs.  
QRTs have covered numerous aspects of quantum properties such as 
entanglement~\cite{Horodecki2009, Rain2001, Brandao2011},  
coherence~\cite{Streltsov2017, Marvian2016, Aberg2006, Baumgratz2014, Yadin2016, Chitambar2016a, Chitambar2016}, 
athermality~\cite{Janzing2000, Horodecki2013, Horodecki2013a,Halpern2017,Halpern2020}, 
magic states~\cite{Veitch2014, Howard2017},
asymmetry~\cite{Gour2008, Gour2009, Marvian2013,Marvian2014,Marvian2014a},
purity~\cite{Gour2015}, 
non-Gaussianity~\cite{Takagi2018,Lami2018,Albarelli2018,Yamasaki2019},
and non-Markovianity~\cite{Wakakuwa2017,Wakakuwa2019}. 
Recently, QRTs for a general resource have been studied to figure out common structures shared among known QRTs and to understand the quantum properties systematically~\cite{Horodecki2013b, Brandao2015, Korzekwa2019,Liu_CL2019,Liu_ZW2019a,Takagi2019b,Takagi2019a,Fang2020,Regula2020,Vijayan2019}. 

However, general QRTs are not necessarily mathematically tractable to analyze, and simply extending the formulation of a known QRT such as bipartite entanglement is insufficient.  
For example, maximal resources in the QRT of magic states~\cite{Veitch2014} and the QRT of coherence with physically incoherent operations (PIO)~\cite{Chitambar2016a} are not unique. 
Gaussian states~\cite{Eisert2005,Holevo2007}, quantum discord~\cite{Bera2017}, and quantum Markov chain~\cite{Gudder2008} are quantum-mechanical properties emerging in a non-convex quantum state space. Gaussian operations~\cite{Eisert2005,Holevo2007} are conventionally defined on a non-convex state space while existing QRTs of non-Gaussianity~\cite{Takagi2018,Lami2018,Albarelli2018,Yamasaki2019} are formulated on a convex state space.
Furthermore, the state spaces of QRTs of non-Gaussianity are infinite-dimensional, and analysis of QRTs on finite-dimensional quantum systems is not necessarily applicable to infinite-dimensional systems.
The analyses of these physically motivated quantum properties are mathematically intractable.

To analyze general quantum properties including those mentioned in the previous paragraph, we investigate manipulation and quantification of quantum resources in general QRTs that are physically motivated but hard to analyze.
We do not make mathematical assumptions such as the existence of a unique maximal resource, a convex state space, and a finite-dimensional state space.

In this paper, we take a position that free operations determine free states. A free operation is an element in a subset of quantum operations. The set of free operations describes what is possible for free when we operate on a quantum system. A quantum state that may not be obtained by free operations is regarded as a resource state, while a quantum state freely obtained by free operations is called a free state.
Convertibility of quantum states under free operations introduces a mathematical concept of order, \textit{preorder}, of the states in terms of resourcefulness.
A maximally resourceful state is a special resource state at the top of this ordering, regarded as a unit of the resource. The existence of maximally resourceful states is essential for quantifying quantum resources. 
Due to the generality of our formulation, the existence of a maximally resourceful state is not obvious, but we prove that a maximally resourceful state always exists by introducing compactness in our framework. 
Furthermore, we analyze the set of free states, and clarify a condition where a maximally resourceful state is not free. 

To investigate manipulation of a quantum resource in general QRTs, we analyze one-shot and asymptotic state conversion in the general framework of QRTs, rather than specific resources.
We discover a type of quantum resource with a counter-intuitive property, 
which is a resource state that is not free to generate but can be replicated infinitely by free operations using a given copy catalytically.
While catalytic conversion of quantum resources is originally found in entanglement theory~\cite{Jonathan1999}, our discovery provides another form of catalytic property of quantum resources.  
We call this resource state a \textit{catalytically replicable state}.
In addition, we formulate resource conversion tasks in general QRTs, namely, \textit{distillation} and \textit{formation} of a resource~\cite{Bennett1996}, and introduce general definitions of the \textit{distillable resource} and the \textit{resource cost} through these tasks, which generalize those defined for bipartite entanglement~\cite{Bennett1996,Hayden2001}, coherence~\cite{Winter2016}, and athermality~\cite{Horodecki2013}.
Formulation of the distillation and the formation of a resource is not straightforward when the QRT has non-unique maximally resourceful states. To overcome this issue, we formulate the distillable resource as how many resources can be extracted from the state in the worst-case scenario, and the resource cost as how many resources are needed to generate the state in the best-case scenario. Under this formulation, we identify a condition of the distillable resource being smaller than the resource cost. 

A resource measure is a tool for quantifying resources. In the QRT of bipartite entanglement, it is known that a resource measure satisfying certain properties given in Ref.~\cite{Donald2002} is lower-bounded by the distillable resource and upper-bounded by the resource cost, which we call the uniqueness inequality.
In this paper, we show that the uniqueness inequality holds for a general QRT under the same properties even in infinite-dimensional cases, but at the same time show that these properties applicable to the QRT of bipartite entanglement are too strong 
to be satisfied in known QRTs such as magic states~\cite{Veitch2014}.
Motivated by this issue, we introduce a concept of \textit{consistent resource measures}, which provide quantification of quantum resources without contradicting the rate of asymptotic state conversion.
We prove that the uniqueness inequality also holds for the consistent resource measure and observe that this uniqueness inequality is more widely applicable than the uniqueness inequality previously proved through the axiomatic approach.
Moreover, we show that the regularized relative entropy of resource serves as a consistent resource measure, generalizing the existing results in reversible QRTs~\cite{Brandao2015}.

These formulations and results establish a framework of general QRTs that are applicable even to physically motivated restrictions on quantum operations whose analysis is mathematically intractable.
We here point out that whether a given QRT is physically motivated or not is indeed a subjective issue, depending on what operation we assume to be free and what quantum property we regard as a resource.
Remarkably, one can use our general results on manipulation and quantification of quantum resources regardless of whether the resources are physical or not from the subjective viewpoint; in particular, our general results on QRTs are applicable to whatever QRTs that one may consider to be ``physical''.
The significance of our contribution is the full generality of the results, which opens a way to stop the controversy over what QRTs are physically meaningful, so that anyone can suitably enjoy the benefit of the general framework of QRTs for quantitatively analyzing quantum properties of interest.
These results lead to a theoretical foundation for further understandings of quantum-mechanical phenomena through a systematic approach based on QRTs.

The rest of this paper is organized as follows.
In Sec.~\ref{Preliminary}, we recall descriptions of infinite-dimensional quantum mechanics and provide a framework of general QRTs. In Sec.~\ref{Section2}, we investigate maximally resourceful states and free states in general QRTs. In Sec.~\ref{Section3}, we analyze manipulation of quantum states in general QRTs, especially, asymptotic state conversion. In Sec.~\ref{Section4}, we focus on the distillation of a resource from a quantum state and the formation of a quantum state from a resource, and prove the uniqueness inequality. In Sec.~\ref{Section5}, we investigate the quantification of a resource, introducing and analyzing a consistent resource measure.
Our conclusion is given in Sec.~\ref{Section6}.

\section{Preliminaries}\label{Preliminary}
We provide preliminaries to quantum resource theories (QRTs) that we analyze in this paper.
In Sec.~\ref{Prelim_a}, we present notations on describing quantum mechanics on infinite-dimensional quantum systems that QRTs in this paper cover.
In Sec.~\ref{Prelim_b}, we recall a formulation of QRTs.
In Sec.~\ref{prelim_c}, we recall the definition of a preorder, and introduce a preorder introduced by free operations in QRTs.
The readers who are interested in QRTs on finite-dimensional quantum systems and are familiar with finite-dimensional quantum mechanics can skip Sec.~\ref{Prelim_a}, while we may use notions summarized in Sec.~\ref{Prelim_a} to show our result in Sec.~\ref{Subsection2_c}.

\subsection{Quantum Mechanics on Infinite-Dimensional Quantum Systems}\label{Prelim_a}

We provide mathematical notations of quantum mechanics that cover infinite-dimensional quantum systems.
Notice that some inherent properties of quantum mechanics, such as non-Gaussianity~\cite{Takagi2018,Lami2018,Albarelli2018,Yamasaki2019}, are easier to formulate on an infinite-dimensional quantum system than its approximation by a finite-dimensional quantum system. 
As for proofs of mathematical facts that we use in the following, see, \textit{e.g.}, Refs.~\cite{ReedandSimon1981,Takesaki2002}.

To represent a (finite- and infinite-dimensional) quantum system, we use a complex Hilbert space $\mathcal{H}$, \textit{i.e.}, a complex inner product space that is also a complete metric space with respect to the distance function induced by the inner product.
We represent a multipartite system as a tensor product of the Hilbert spaces representing its subsystems.
We may write an orthonormal basis (\textit{i.e.}, a complete orthonormal system) of $\mathcal{H}$ as
\begin{equation}
\label{eq:basis}
B_\mathcal{H}\coloneqq{\left\{\Ket{k}\right\}}_k.
\end{equation}
In cases where $\mathcal{H}$ represents a $D$-dimensional system, $B_\mathcal{H}$ is a finite set with cardinality $D$, while $B_\mathcal{H}$ can be an uncountable set in this paper.

We use a subset of operators on $\mathcal{H}$, in particular, the von Neumann algebra, to describe quantum mechanics on the system represented by $\mathcal{H}$, since $\mathcal{H}$ can be infinite-dimensional.
Let $\mathcal{L}\left(\mathcal{H}\right)$ denote the set of linear operators on $\mathcal{H}$.
Let $\mathcal{B}\left(\mathcal{H}\right)\subset\mathcal{L}\left(\mathcal{H}\right)$ denote the set of bounded operators, that is, for any $A\in\mathcal{B}\left(\mathcal{H}\right)$, the operator norm $\left\|A\right\|_\infty\coloneqq\sup\left\{{\left\|A\Ket{v}\right\|}_\mathcal{H}:\Ket{v}\in\mathcal{H},\,{\left\|\Ket{v}\right\|}_\mathcal{H}\leqq 1\right\}$ is bounded, where ${\left\|\cdot\right\|}_\mathcal{H}$ denotes the norm induced by the inner product of $\mathcal{H}$.

To define trace-class operators,
we use the trace defined as $\tr T\coloneqq\sum_{\Ket{k}\in B_\mathcal{H}}\Braket{k|T|k}$, where $B_\mathcal{H}$ is an orthonormal basis of $\mathcal{H}$ defined in~\eqref{eq:basis}, each term on the right-hand side denotes the inner product of $T\Ket{k}$ and $\Ket{k}$ on $\mathcal{H}$, and if $B_\mathcal{H}$ is an infinite set, the summation on the right-hand side means the limit of a net, \textit{i.e.}, a generalization of sequence.
While a sequence $(a_n:n\in\mathbb{N})$ is indexed by a natural number, that is, a countably infinite and totally ordered set, a net $a_n$ is indexed by $n$ in a directed set, a generalization of totally ordered sets, while this directed set can be uncountable.
In the definition of the trace, the net is indexed by any finite subset $B_\mathcal{H}^\prime\subset B_\mathcal{H}$ and defined as $\sum_{\Ket{k}\in B_\mathcal{H}^\prime}\Braket{k|T|k}$, which approaches to the limit $\tr T$ as $B_\mathcal{H}^\prime$ gets larger.
Note that $\tr T$ is independent of the choice of $B_\mathcal{H}$.
Let $\mathcal{T}\left(\mathcal{H}\right)\subset\mathcal{B}\left(\mathcal{H}\right)$ denote the set of trace-class operators; that is, for any $T\in\mathcal{T}\left(\mathcal{H}\right)$, we have finite and hence well-defined $\tr\left|T\right|$, where $\left|T\right|\coloneqq\sqrt{T^\dag T}$, and $T^\dag$ denotes the conjugation of $T$.
For any $T\in\mathcal{T}\left(\mathcal{H}\right)$, the trace norm of $T$ is defined as
\begin{equation}
  \label{eq:trace_norm}
  \left\|T\right\|_1\coloneqq\tr\left|T\right|.
\end{equation}

To define the von Neumann algebra, we need to discuss convergence of bounded operators mathematically. To discuss convergence of bounded operators, we need a topology defined for $\mathcal{B}\left(\mathcal{H}\right)$, and we use the ultraweak operator topology.
The ultraweak operator topology of $\mathcal{B}\left(\mathcal{H}\right)$ is a topology where any sequence $A_1,A_2,\ldots\in \mathcal{B}\left(\mathcal{H}\right)$, or more generally, any net $A_i$, converges to $A$ if and only if $\tr\left[TA_i\right]$ converges to $\tr\left[TA\right]$ for any $T\in\mathcal{T}\left(\mathcal{H}\right)$.
A von Neumann algebra $\mathcal{M}$ on $\mathcal{B}\left(\mathcal{H}\right)$ is a subset of $\mathcal{B}\left(\mathcal{H}\right)$ (or $\mathcal{B}\left(\mathcal{H}\right)$ itself) that contains the identity operator $\mathbbm{1}$ on $\mathcal{H}$, is closed under linear combination, product, and conjugation, and is also closed in terms of the ultraweak operator topology.

A noncommutative von Neumann algebra can be used for describing a quantum system, while a commutative von Neumann algebra for a classical system.
To describe quantum mechanics on $\mathcal{H}$, we use a set of operators represented as a von Neumann algebra $\mathcal{M}$ on $\mathcal{B}\left(\mathcal{H}\right)$.
For example, for any finite-dimensional Hilbert space $\mathcal{H}$, the algebra of all the linear operators $\mathcal{L}\left(\mathcal{H}\right)=\mathcal{B}\left(\mathcal{H}\right)$ is a von Neumann algebra, which suffices to describe the finite-dimensional quantum mechanics.
More generally, for any $\mathcal{H}$ that can be infinite-dimensional, the algebra of all the bounded operators $\mathcal{B}\left(\mathcal{H}\right)$ is a von Neumann algebra.
In this paper, a system $\mathcal{H}$ is always accompanied with a set of operators $\mathcal{M}$ therein, where we may implicitly consider $\mathcal{M}=\mathcal{B}\left(\mathcal{H}\right)$ unless stated otherwise.

Given that a quantum state associates a measurement of an observable with a probability of a measurement outcome,
we introduce a quantum state using a linear functional from an operator to a scalar.
In particular, for a system $\mathcal{H}$ with $\mathcal{M}$, a state is defined as a linear functional $f_\psi:\mathcal{M}\to\mathbb{C}$ that is positive semidefinite $f_\psi\left(M^\dag M\right)\geqq 0,\,\forall M\in\mathcal{M}$,
satisfies the normalization condition $f_\psi\left(\mathbbm{1}\right)=1$,
and is also normal, \textit{i.e.}, continuous in terms of the ultraweak operator topology.
We require this continuity in order for $f_\psi$ to be in the predual $\mathcal{M}_\ast$ of $\mathcal{M}$, \textit{i.e.}, the space whose dual ${\left(\mathcal{M}_\ast\right)}^\ast$ equals (can be identified with) $\mathcal{M}$.
Given this duality and under the condition of $\mathcal{M}=\mathcal{B}\left(\mathcal{H}\right)$, we identify $f_\psi$ with the operator $\psi\in\mathcal{T}\left(\mathcal{H}\right)$ that satisfies for any $M\in\mathcal{M}$
\begin{equation}
  \label{eq:state_identification}
  \tr\left[\psi M\right]=f_\psi\left(M\right).
\end{equation}
Note that we have one-to-one correspondence between $f_\psi$ and $\psi$ if $\mathcal{M}=\mathcal{B}\left(\mathcal{H}\right)$.
This operator $\psi$ is the density operator representing the state $f_\psi$,
and let $\mathcal{D}\left(\mathcal{H}\right)\coloneqq\left\{\psi\in\mathcal{T}\left(\mathcal{H}\right):\psi\geqq 0, \tr\psi=1\right\}$ denote the set of density operators on $\mathcal{H}$ with $\mathcal{M}=\mathcal{B}\left(\mathcal{H}\right)$.
For simplicity, we may call $\psi$ a quantum state, rather than $f_\psi$.

We introduce a quantum channel on a system $\mathcal{H}$ with $\mathcal{M}=\mathcal{B}\left(\mathcal{H}\right)$ in the Heisenberg picture as a completely positive and unital linear map on $\mathcal{M}$, which correspondingly yields the definition of the channel in the Schr\"{o}dinger picture as a completely positive and trace-preserving linear map of density operators on $\mathcal{H}$.
Given two systems $\mathcal{H}^{\left(\mathrm{in}\right)}$ with $\mathcal{M}^{\left(\mathrm{in}\right)}$ and $\mathcal{H}^{\left(\mathrm{out}\right)}$ with $\mathcal{M}^{\left(\mathrm{out}\right)}$ representing the spaces of the input and the output respectively, a channel $\tilde{\mathcal{E}}:\mathcal{M}^{\left(\mathrm{out}\right)}\to\mathcal{M}^{\left(\mathrm{in}\right)}$ in the Heisenberg picture is defined as a linear map that is completely positive
\begin{align}
  &\sum_{j,k=0}^{n-1}M_j^{\left(\mathrm{in}\right)}{}^\dag\tilde{\mathcal{E}}\left(M_j^{\left(\mathrm{out}\right)}{}^\dag M_k^{\left(\mathrm{out}\right)}\right)M_k^{\left(\mathrm{in}\right)}\geqq 0,\nonumber\\
  &\forall M_0^{\left(\mathrm{in}\right)},\ldots,M_{n-1}^{\left(\mathrm{in}\right)}\in\mathcal{M}^{\left(\mathrm{in}\right)},\nonumber\\
  &\forall M_0^{\left(\mathrm{out}\right)},\ldots,M_{n-1}^{\left(\mathrm{out}\right)}\in\mathcal{M}^{\left(\mathrm{out}\right)},\nonumber\\
  &\forall n\in\mathbb{N},
\end{align}
unital $\tilde{\mathcal{E}}\left(\mathbbm{1}^{\left(\mathrm{out}\right)}\right)=\mathbbm{1}^{\left(\mathrm{in}\right)}$,
and normal, \textit{i.e.}, continuous in terms of the ultraweak operator topologies of $\mathcal{M}^{\left(\mathrm{out}\right)}$ and $\mathcal{M}^{\left(\mathrm{in}\right)}$,
where $\mathbbm{1}^{\left(\mathrm{out}\right)}$ and $\mathbbm{1}^{\left(\mathrm{in}\right)}$ are the identity operators on $\mathcal{H}^{\left(\mathrm{out}\right)}$ and $\mathcal{H}^{\left(\mathrm{in}\right)}$, respectively.
In the same way as the identification~\eqref{eq:state_identification} of a functional $f_\psi$ of a state with the density operator $\psi$ of the state, under the conditions of $\mathcal{M}^{\left(\mathrm{in}\right)}=\mathcal{B}\left(\mathcal{H}^{\left(\mathrm{in}\right)}\right)$ and $\mathcal{M}^{\left(\mathrm{out}\right)}=\mathcal{B}\left(\mathcal{H}^{\left(\mathrm{out}\right)}\right)$, we identify a channel $\tilde{\mathcal{E}}:\mathcal{M}^{\left(\mathrm{out}\right)}\to\mathcal{M}^{\left(\mathrm{in}\right)}$ in the Heisenberg picture with the channel $\mathcal{E}:\mathcal{T}\left(\mathcal{H}^{\left(\mathrm{in}\right)}\right)\to\mathcal{T}\left(\mathcal{H}^{\left(\mathrm{out}\right)}\right)$ in the Schr\"{o}dinger picture that satisfies for any $M^{\left(\mathrm{out}\right)}\in\mathcal{M}^{\left(\mathrm{out}\right)}$
\begin{equation}
  \label{eq:heisenberg_schrodinger}
  \tr\left[\mathcal{E}\left(\psi\right)M^{\left(\mathrm{out}\right)}\right]=\left(f_\psi\circ\tilde{\mathcal{E}}\right)\left(M^{\left(\mathrm{out}\right)}\right),
\end{equation}
where $\psi$ and $f_\psi$ are related as~\eqref{eq:state_identification}.
Note that $\mathcal{E}$ is a completely positive and trace-preserving (CPTP) linear map by definition, and if $\mathcal{M}^{\left(\mathrm{in}\right)}=\mathcal{B}\left(\mathcal{H}^{\left(\mathrm{in}\right)}\right)$ and $\mathcal{M}^{\left(\mathrm{out}\right)}=\mathcal{B}\left(\mathcal{H}^{\left(\mathrm{out}\right)}\right)$, we have one-to-one correspondence between $\tilde{\mathcal{E}}$ and $\mathcal{E}$.

The set of channels from an input system $\mathcal{H}^{\left(\mathrm{in}\right)}$ with $\mathcal{M}^{\left(\mathrm{in}\right)}$ to an output system $\mathcal{H}^{\left(\mathrm{out}\right)}$ with $\mathcal{M}^{\left(\mathrm{out}\right)}$ is denoted by $\mathcal{C}\left(\mathcal{H}^{\left(\mathrm{in}\right)}\to\mathcal{H}^{\left(\mathrm{out}\right)}\right)$, which we may write $\mathcal{C}\left(\mathcal{H}\right)$ if $\mathcal{H}=\mathcal{H}^{\left(\mathrm{in}\right)}=\mathcal{H}^{\left(\mathrm{out}\right)}$.
In this paper, we use the Schr\"{o}dinger picture with $\mathcal{M}^{\left(\mathrm{in}\right)}=\mathcal{B}\left(\mathcal{H}^{\left(\mathrm{in}\right)}\right)$ and $\mathcal{M}^{\left(\mathrm{out}\right)}=\mathcal{B}\left(\mathcal{H}^{\left(\mathrm{out}\right)}\right)$; that is, $\mathcal{C}\left(\mathcal{H}^{\left(\mathrm{in}\right)}\to\mathcal{H}^{\left(\mathrm{out}\right)}\right)$ is the set of the CPTP linear maps, while it would be possible to use the Heisenberg picture otherwise. We represent quantum operations as channels, while it is possible to include measurements in our formulation as channels from a quantum input system to a classical output system.

To discuss compactness of a set of states,
we need further definitions of topologies for the set of states.
A compact set in terms of some topology is a set where for any net in the set, there exists a subnet that converges in terms of the topology, where a subnet generalizes a subsequence in the same way as the net generalizing the sequence.
A closed set in terms of some topology is a set where for any net that converges in terms of this topology, its limit point is in the set.
Note that a compact set in a Hausdorff space is a closed set, which holds in our case.
Several different topologies can be defined for the set of states.
The weak operator topology of $\mathcal{T}\left(\mathcal{H}\right)$ is a topology where any net $T_i$ in $\mathcal{T}(\mathcal{H})$ converges to $T$ if and only if $\tr\left[T_i A\right]$ converges to $\tr\left[TA\right]$ for any $A\in\mathcal{B}\left(\mathcal{H}\right)$.
The trace norm topology of $\mathcal{T}\left(\mathcal{H}\right)$ is a topology where any net $T_i$ in $\mathcal{T}(\mathcal{H})$ converges to $T$ if and only if ${\left\|T_i-T\right\|}_1$ converges to $0$.
The trace norm topology is stronger than the ultraweak operator topology, and the ultraweak operator topology is stronger than the weak operator topology, while these topologies are the same in the finite-dimensional case.
In general, whether a set is compact or not may depend on the choice of the topology, but we show in Appendix~\ref{Appendix_A} that the compactness of a set of states in terms of these topologies are equivalent.
Thus, it suffices to consider the trace norm topology when we discuss compactness of a set of states.
Then, in the trace norm topology, compactness is equivalent to sequential compactness, and hence we may use a sequence rather than a net to discuss compactness of a set of states.
Note that if $\mathcal{H}$ is finite-dimensional, $\mathcal{D}\left(\mathcal{H}\right)$ is compact, while $\mathcal{D}\left(\mathcal{H}\right)$ for an infinite-dimensional system is not compact in terms of the trace norm topology.

To discuss convergence and compactness of channels, we need a topology defined for $\mathcal{C}\left(\mathcal{H}^{\left(\mathrm{in}\right)}\to\mathcal{H}^{\left(\mathrm{out}\right)}\right)$, and we use the bounded weak (BW) topology.
The bounded weak topology of $\mathcal{C}\left(\mathcal{H}^{\left(\mathrm{in}\right)}\to\mathcal{H}^{\left(\mathrm{out}\right)}\right)$ is the weakest topology such that for any $f_\psi\in\mathcal{M}^{\left(\mathrm{in}\right)}_\ast$ and $M^{\left(\mathrm{out}\right)}\in\mathcal{M}^{\left(\mathrm{out}\right)}$, a map $\mathcal{S}_{\psi,M^{\left(\mathrm{out}\right)}}:\mathcal{C}\left(\mathcal{H}^{\left(\mathrm{in}\right)}\to\mathcal{H}^{\left(\mathrm{out}\right)}\right)\to\mathbb{C}$ given by
$\mathcal{S}_{\psi,M^{\left(\mathrm{out}\right)}}\left(\mathcal{E}\right)=\left(f_\psi\circ\tilde{\mathcal{E}}\right)\left(M^{\left(\mathrm{out}\right)}\right)$
is continuous, where $\mathcal{E}$ and $\tilde{\mathcal{E}}$ are related as~\eqref{eq:heisenberg_schrodinger}.
Note that if $\mathcal{H}^{\left(\mathrm{in}\right)}$ and $\mathcal{H}^{\left(\mathrm{out}\right)}$ are finite-dimensional, $\mathcal{C}\left(\mathcal{H}^{\left(\mathrm{in}\right)}\to\mathcal{H}^{\left(\mathrm{out}\right)}\right)$ is compact, while $\mathcal{C}\left(\mathcal{H}^{\left(\mathrm{in}\right)}\to\mathcal{H}^{\left(\mathrm{out}\right)}\right)$ for infinite-dimensional systems is not compact in terms of the BW topology.

\subsection{Framework of Quantum Resource Theories}\label{Prelim_b}

In this section, we provide a formulation of quantum resource theories (QRTs) starting from free operations with minimal assumptions. In the definition, we consider a compact set of the free operations.
We also present justification of the compactness by examples from the perspective of indistinguishability.
To represent the state set of interest in a QRT, \textit{e.g.}, the set of pure states on finite-dimensional $\mathcal{H}$,
we consider a compact set of quantum states chosen as desired
\begin{equation}
  \mathcal{S}\left(\mathcal{H}\right) \subseteq \mathcal{D}\left(\mathcal{H}\right).
\end{equation}
Note that the quantum system $\mathcal{H}$ can be infinite-dimensional as we have introduced in Sec.~\ref{Prelim_a}.

Free operations in our formulation are introduced as follows~\cite{Chitambar2018}. Let $\mathcal{O}$ be a mapping that takes two quantum systems $\mathcal{H}^{\left(\mathrm{in}\right)}$ and $\mathcal{H}^{\left(\mathrm{out}\right)}$ and outputs a compact set of completely positive and trace-preserving (CPTP) maps from $\mathcal{S}\left(\mathcal{H}^{\left(\mathrm{in}\right)}\right)$ to $\mathcal{S}\left(\mathcal{H}^{\left(\mathrm{out}\right)}\right)$.
This set is denoted by
\begin{equation}
\mathcal{O}\left(\mathcal{H}^{\left(\mathrm{in}\right)} \to \mathcal{H}^{\left(\mathrm{out}\right)}\right)\subseteq\mathcal{C}\left(\mathcal{H}^{\left(\mathrm{in}\right)} \to \mathcal{H}^{\left(\mathrm{out}\right)}\right).
\end{equation}
A map contained in $\mathcal{O}\left(\mathcal{H}^{\left(\mathrm{in}\right)} \to \mathcal{H}^{\left(\mathrm{out}\right)}\right)$ is called a \textit{free operation} from $\mathcal{H}^{\left(\mathrm{in}\right)}$ to $\mathcal{H}^{\left(\mathrm{out}\right)}$. If the input space and output space are the same quantum system $\mathcal{H}$, we write the set of free operations from $\mathcal{H}$ to $\mathcal{H}$ as $\mathcal{O}\left(\mathcal{H}\right)\subseteq\mathcal{C}\left(\mathcal{H}\right)$.
We consider a compact set because two arbitrarily close CPTP maps are indistinguishable by any protocol in a task of channel discrimination~\cite{Kitaev1997}, as we will discuss below by examples.
We assume that $\mathcal{O}$ satisfies the following axioms of QRTs:
\begin{enumerate}
  \item Let $\mathcal{H}_1$, $\mathcal{H}_2$ and $\mathcal{H}_3$ be arbitrary quantum systems. For any $\mathcal{M}\in\mathcal{O}\left(\mathcal{H}_{1} \to \mathcal{H}_2\right)$ and $\mathcal{N}\in\mathcal{O}\left(\mathcal{H}_{2} \to \mathcal{H}_3\right)$, it holds that $\mathcal{N}\circ\mathcal{M}\in\mathcal{O}\left(\mathcal{H}_{1} \to \mathcal{H}_3\right)$, where $\circ$ represents the composition.
  \item Let $\mathcal{H}^{\left(\mathrm{in}\right)}_1$, $\mathcal{H}^{\left(\mathrm{out}\right)}_1$, $\mathcal{H}^{\left(\mathrm{in}\right)}_2$ and $\mathcal{H}^{\left(\mathrm{out}\right)}_2$ be arbitrary quantum systems. For any $\mathcal{M}\in\mathcal{O}(\mathcal{H}^{\left(\mathrm{in}\right)}_1 \to \mathcal{H}^{\left(\mathrm{out}\right)}_1)$ and $\mathcal{N}\in\mathcal{O}(\mathcal{H}^{\left(\mathrm{in}\right)}_2 \to \mathcal{H}^{\left(\mathrm{out}\right)}_2)$, it holds that $\mathcal{M}\otimes\mathcal{N}\in\mathcal{O}(\mathcal{H}^{\left(\mathrm{in}\right)}_1\otimes\mathcal{H}^{(\textrm{in})}_{2} \to \mathcal{H}^{\left(\mathrm{out}\right)}_1\otimes\mathcal{H}^{(\textrm{out})}_{2})$.
  \item Let $\mathcal{H}$ be an arbitrary quantum system. Then, it holds that $\id\in\mathcal{O}\left(\mathcal{H}\right)$, where $\id$ is an identity map.
  \item Let $\mathcal{H}$ be an arbitrary quantum system. Then, it holds that $\tr\in\mathcal{O}\left(\mathcal{H}\to\mathbb{C}\right)$, where $\tr$ is the trace. Note that due to the above conditions, it is necessary that the partial trace is also free.
\end{enumerate}
The meanings of these axioms are as follows:
\begin{enumerate}
  \item We always have access to free operations and can use free operations as many times as necessary.
  \item We can arbitrarily apply free operations to a quantum system regardless of what free operations are applied to another quantum system.
  \item Doing nothing is free.
  \item Ignorance is free.
\end{enumerate}
\begin{remark}[Operations Not Satisfying the Axioms]
There can be classes of operations that do not satisfy the axioms stated above. 
For example, Refs.~\cite{Brandao2008} and \cite{Brandao2015} consider $\epsilon$-resource non-generating operations. 
However, the composition of two $\epsilon$-resource non-generating operations is not necessarily an $\epsilon$-resource non-generating operation, which implies the set of $\epsilon$-resource non-generating operations does not satisfy the first axiom. 
Hence, we do not employ this class of operations as free operations since they are not free to use multiple times.
In addition, Ref.~\cite{Contreras-Tejada2019} considers separability preserving (SEPP) operations. However, the set of SEPP operations is not closed under tensor product, and hence does not satisfy the second axiom.
We do not use these operations as free operations since they are not free to apply to multiple quantum systems simultaneously.
\end{remark}

In the definition above, we use a compact set as the set of free operations.
Some classes of operations that are conventionally used as free operations do not satisfy this compactness,
such as local quantum operations and classical communication (LOCC) in the QRT of bipartite entanglement~\cite{Chitambar2014}.
However, in this case, we take a position that the closure of LOCC, \textit{i.e.}, a compact superset of LOCC, can be considered to be free in the sense that any channel in the closure of LOCC is indistinguishable from a channel implementable in the setting of LOCC, as discussed in Example~\ref{ex:locc}.
In the same way, Example~\ref{ex:circuit} shows that we conventionally consider any unitary transformation to be implementable by the Clifford+$T$ gate set in the sense that any unitary can be approximated with arbitrary precision by this gate set.
Note that the compactness of the set of free operations is essential for guaranteeing the existence of maximally resourceful states as we will see in Sec.~\ref{Subsection2_c}.

\begin{example}[LOCC and Closure of LOCC]\label{ex:locc}
  In the case of the QRT of entanglement, LOCC is conventionally considered to be physically implementable operations,
  but our formulation of QRTs may use the closure of LOCC in this case as a compact set of free operations instead of LOCC\@.
  In particular, let
  $\mathcal{O}_\mathrm{LOCC}\left(\mathcal{H}^{\left(\mathrm{in}\right)}\to\mathcal{H}^{\textrm{out}}\right)$
  be the set of LOCC from
  $\mathcal{H}^{\left(\mathrm{in}\right)}$
  to
  $\mathcal{H}^{\left(\mathrm{out}\right)}$.
  It is known that
  $\mathcal{O}_\mathrm{LOCC}\left(\mathbb{C}^{4}\to\mathbb{C}^{4}\right)$
  is not closed; that is,
  $\overline{\mathcal{O}_\mathrm{LOCC}\left(\mathbb{C}^{4}\to\mathbb{C}^{4}\right)}
  \neq \mathcal{O}_\mathrm{LOCC}\left(\mathbb{C}^{4}\to\mathbb{C}^{4}\right)$~\cite{Chitambar2014}. In this case, we use $\mathcal{O}\left(\mathbb{C}^{4}\to\mathbb{C}^{4}\right)=\overline{\mathcal{O}_\mathrm{LOCC}\left(\mathbb{C}^{4}\to\mathbb{C}^{4}\right)}$ as the set of free operation because for any CPTP map $\mathcal{N}\in\overline{\mathcal{O}_\mathrm{LOCC}\left(\mathbb{C}^{4}\to\mathbb{C}^{4}\right)}\setminus\mathcal{O}_\mathrm{LOCC}\left(\mathbb{C}^{4}\to\mathbb{C}^{4}\right)$ and any $\epsilon>0$, we can construct a CPTP map $\tilde{\mathcal{N}}\in\mathcal{O}_\mathrm{LOCC}\left(\mathbb{C}^{4}\to\mathbb{C}^{4}\right)$ that is indistinguishable from $\mathcal{N}$ up to an $\epsilon$ probability by any protocol in a task of channel discrimination~\cite{Kitaev1997}.
\end{example}

In the next example, 
we consider a situation of universal quantum computation where any finite-depth quantum circuit composed of a universal gate set is implementable as the free operations.

\begin{example}\label{ex:circuit}
For any finite-dimensional quantum systems $\mathcal{H}^{\left(\mathrm{in}\right)}$ and $\mathcal{H}^{\left(\mathrm{out}\right)}$, 
we define $\mathcal{O}'\left(\mathcal{H}^{\left(\mathrm{in}\right)}\to\mathcal{H}^{\left(\mathrm{out}\right)}\right)$ as the set of CPTP maps that can be realized by a finite-depth circuit composed of the identity gate, the partial trace, the Hadamard gate $H$, the controlled-NOT gate, the $\pi/8$ phase gate $T$, appending $\ket{0}$, and the measurement in the computational basis.
Here, a finite-depth circuit refers to a $d$-depth circuit for some integer $d\geqq0$, 
and operations conditioned on measurement outcomes are allowed.
Conventionally, combination of these operations can be considered to be universal because  $\mathcal{O}'\left(\mathcal{H}^{\left(\mathrm{in}\right)}\to\mathcal{H}^{\left(\mathrm{out}\right)}\right)$ is dense in the set of all the channels $\mathcal{C}\left(\mathcal{H}^{\left(\mathrm{in}\right)}\to\mathcal{H}^{\left(\mathrm{out}\right)}\right)$, but $\mathcal{O}'\left(\mathcal{H}^{\left(\mathrm{in}\right)}\to\mathcal{H}^{\left(\mathrm{out}\right)}\right)$ is different from $\mathcal{C}\left(\mathcal{H}^{\left(\mathrm{in}\right)}\to\mathcal{H}^{\left(\mathrm{out}\right)}\right)$
since a finite number of these gates cannot exactly implement qubit rotation at an arbitrary angle~\cite{Kitaev1997, nielsen_chuang_2010}.
In this case, our framework may use the closure of the set of operations implemented by all the $d$-depth circuits for any $d$ as the set of the free operations; that is, we take the set of free operations in this case as the set of all the channels $\mathcal{O}\left(\mathcal{H}^{\left(\mathrm{in}\right)}\to\mathcal{H}^{\left(\mathrm{out}\right)}\right)=\mathcal{C}\left(\mathcal{H}^{\left(\mathrm{in}\right)}\to\mathcal{H}^{\left(\mathrm{out}\right)}\right)$.
\end{example}

We do not assume the convexity of the set of free operations in our framework. Convex QRTs are a class of QRTs where the set of free operations is convex. 
For instance, the QRT of bipartite entanglement~\cite{Chitambar2014}, the QRT of coherence~\cite{Winter2016} and the QRT of magic states~\cite{Veitch2014} are known as convex QRTs. 
We can achieve a convex combination of operations using classical randomness. 
In general, randomness is regarded as a resource~\cite{Groisman2005,Anshu2018}, and randomness generation~\cite{Pir2019} is indeed a promising application of noisy intermediate-scale quantum (NISQ) devices~\cite{Preskill2018}; therefore, we also consider non-convex QRTs in our framework, such as the following example.

\begin{example}[Non-Convex QRT]
The QRT of non-Markovianity~\cite{Wakakuwa2017,Wakakuwa2019} is known as a non-convex QRT, where the set of free operations is not convex. 
\end{example}

\subsection{Preorder in Quantum Resource Theories}~\label{prelim_c}
In this section, we recall the definition of a \textit{preorder} in order theory, 
and provide the definition of the preorder of resourcefulness of quantum states introduced by free operations.

Suppose that $S$ is a set and $\succeq$ is a binary relation on $S$.
Then, the relation $\succeq$ is called a \textit{preorder} if it holds that
\begin{align}
	a &\succeq a, \\
	a &\succeq b , b\succeq c \Rightarrow a \succeq c
\end{align}
for all $a,b,c \in S$.
This preorder introduces \textit{maximal elements} and \textit{minimal elements} 
in the set $S$. 
If $a \in S$ satisfies
\begin{equation}
	b \succeq a \Rightarrow a \succeq b
\end{equation}
for all $b \in S$, $a$ is called a maximal element. 
Similarly, if $a \in S$ satisfies
\begin{equation}
	a \succeq b \Rightarrow b \succeq a
\end{equation}
for all $b \in S$, $a$ is called a minimal element.
Intuitively, an element $a \in S$ is maximal/minimal if $a$ is the largest/smallest 
among all the elements that can be compared with $a$.

For any quantum system $\mathcal{H}$, free operations introduce a preorder $\succeq$ on $\mathcal{S}\left(\mathcal{H}\right)$.
This preorder is introduced in terms of the exact one-shot state conversion under free operations. 
Given two states $\phi,\psi \in \mathcal{S}\left(\mathcal{H}\right)$, the exact one-shot state conversion from $\phi$ to $\psi$ is a task of transforming a single $\phi$ exactly into a single $\psi$ by a free operation $\mathcal{N}\in\mathcal{O}\left(\mathcal{H}\right)$.  
Formally, we write
\begin{equation}\label{def:preorder}
  \phi\succeq\psi
\end{equation}
if there exists a free operation $\mathcal{N}\in\mathcal{O}\left(\mathcal{H}\right)$ such that
\begin{equation}\label{one-shot_conversion}
  \mathcal{N}\left(\phi\right)=\psi.
\end{equation}
This relation $\succeq$ is indeed a preorder because it holds that
\begin{align}
\phi &\succeq \phi, \\
\phi &\succeq \psi , \psi\succeq \sigma \Rightarrow \phi \succeq \sigma
\end{align}
for any states $\phi,\psi,\sigma\in\mathcal{S}\left(\mathcal{H}\right)$.

With respect to this preorder, two states $\phi,\psi \in \mathcal{S}\left(\mathcal{H}\right)$ are said to be equivalent if both $\phi\succeq\psi$ and $\phi\preceq\psi$ hold. 
If $\phi$ and $\psi$ are equivalent, we write
\begin{equation}
  \phi\sim\psi. 
\end{equation}

\section{Maximally Resourceful States and Free States}\label{Section2}
In this section, we analyze properties of maximally resourceful states and free states in general quantum resource theories (QRTs). 
In Sec.~\ref{Subsection2_c}, we provide the definition of maximally resourceful states based on the preorder mentioned in Sec.~\ref{prelim_c}, and prove the existence of maximally resourceful states in QRTs under our formulation in Sec.~\ref{Prelim_b}.
In Sec.~\ref{Subsection2_b}, we provide the definition of free states, and give a condition under which a maximally resourceful state is not free.

\subsection{Maximally Resourceful States}\label{Subsection2_c}
In this section, we analyze maximally resourceful states defined by the preorder of resourcefulness of states mentioned in Sec.~\ref{prelim_c}. 
The existence of maximally resourceful states is not trivial in general, although it is desired for quantification of the resource. 
We prove the existence of maximally resourceful states in any QRT that satisfies the four axioms and the compactness given in Sec.~\ref{Prelim_b}.

The preorder of states introduces maximal elements in the set of states in terms of order theory. 
Given a quantum system $\mathcal{H}$,
we let $\mathcal{G}\left(\mathcal{H}\right)$ denote the set of the maximal states of $\mathcal{S}\left(\mathcal{H}\right)$ in terms of the preorder defined as~\eqref{def:preorder}, that is,
\begin{equation}
  \mathcal{G}\left(\mathcal{H}\right)\coloneqq\left\{\phi\in\mathcal{S}\left(\mathcal{H}\right):\forall\psi\in\mathcal{S}\left(\mathcal{H}\right),\psi\succeq\phi\Rightarrow\phi\succeq\psi\right\}.
\end{equation}
The elements of $\mathcal{G}\left(\mathcal{H}\right)$ are called \textit{maximally resourceful states}. Note that there may be several non-equivalent maximally resourceful states, which are not comparable with each other. 
Here, we recall two QRTs that have two or more non-equivalent maximally resourceful states. 
\begin{example}[QRTs with Non-Equivalent Maximally Resourceful States]\label{example_QRT_noneq_max}
The first example is the QRT of magic for qutrits~\cite{Veitch2014}.
In this QRT, there exist two non-equivalent maximally resourceful states, which are called the Norrell state and the Strange state.
The second example is the QRT of coherence with physically incoherent operations (PIO)~\cite{Chitambar2016a}.
In the QRT of coherence with PIO, a free operation cannot change diagonal elements of a quantum state represented in the standard basis. 
Therefore, there exist infinitely many non-equivalent maximally resourceful states, which have different diagonal elements from each other. 
\end{example}

We here prove that a maximally resourceful state always exists in any QRT satisfying axioms and the compactness of the set of states discussed in Sec.~\ref{Prelim_b}. Maximally resourceful states are regarded as a unit of resource~\cite{Bennett1996,Baumgratz2014,Bravyi2012}. For example, in the QRT of bipartite entanglement, the amount of entanglement of the Bell state is defined as one ebit. Therefore, it is crucial for QRTs to have a maximally resourceful state. In general, whether a maximally resourceful state exists is not obvious. For example, a maximally entangled state does not necessarily exist in a QRT of bipartite entanglement for an infinite-dimensional system with a non-compact set of free operations such as LOCC while a unique maximally entangled state exists for a finite-dimensional system. Theorem~\ref{prop5} shows that for any given state, there exists a maximally resourceful state that is more resourceful than the state, which ensures the existence of maximally resourceful states in our framework.

\begin{theorem}[Existence of a Maximally Resourceful State]\label{prop5}
Let $\mathcal{H}$ is a quantum system. For any state $\psi \in \mathcal{S}\left(\mathcal{H}\right)$, there exists a state $\phi \in \mathcal{G}\left(\mathcal{H}\right)$ that upper-bounds $\psi$; that is, $\psi \preceq \phi$.
\end{theorem}

\begin{proof}
It is known that a compact space $\mathcal{X}$ with a preorder $\succeq$ has a maximal element if the upper closure
$\mathcal{U}_x \coloneqq \left\{y \in \mathcal{X} | x \preceq y\right\}$ is closed for any $x \in \mathcal{X}$~\cite{Berger2003} (Proposition VI-1.6.{(i)}).
Thus, it suffices to show $\mathcal{U}_\psi \coloneqq \left\{\phi \in \mathcal{S}\left(\mathcal{H}\right) | \psi \preceq \phi\right\}$ is closed in terms of the weak operator topology, or equivalently closed in terms of the trace norm topology due to Lemma~\ref{l1} in Appendix~\ref{Appendix_A}.
We take a sequence ${(\phi_n)}_{n \in \mathbb{N}}$
in $\mathcal{U}_\psi$ convergent to $\phi \in \mathcal{S}\left(\mathcal{H}\right)$ in terms of the trace norm topology
and prove $\phi \in \mathcal{U}_\psi$.
By the definition of the preorder $\preceq$,
for each $n \in \mathbb{N}$, there exists a free operation $\mathcal{N}_n \in \mathcal{O}\left(\mathcal{H}\right)$
such that $\psi =  \mathcal{N}_n (\phi_n )$.
Since $\mathcal{O}\left(\mathcal{H}\right)$ is compact in terms of the bounded weak (BW) topology, there exists a subnet
${\left(\mathcal{N}_{n(i)}\right)}_{i \in I}$ convergent to some $\mathcal{N} \in \mathcal{O}\left(\mathcal{H}\right)$ with respect to the BW topology.
In the following, we show $\psi = \mathcal{N} (\phi)$ to prove the theorem.

Take an arbitrary $\epsilon >0 $ and an arbitrary $A\in \mathcal{B}\left(\mathcal{H}\right)\setminus\{0\}$, which satisfies ${\left\|A\right\|}_\infty>0$.
Since $ \mathcal{S}\left(\mathcal{H}\right)$ is compact in terms of the trace norm topology,
there exists a finite subset $\left\{\chi_k:k\in\left\{1,\ldots,N_\epsilon\right\}\right\}$ of $ \mathcal{S}\left(\mathcal{H}\right)$ such that for any $\chi\in\mathcal{S}(\mathcal{H})$
\begin{equation}
  \label{eq:compact_subset}
  \min_{k\in\left\{1,\ldots,N_\epsilon\right\}}{\left\| \chi - \chi_k \right\|}_1 < \frac{\epsilon}{{\|A\|}_\infty }.
\end{equation}
By definition of the convergence $\mathcal{N}_{n(i)} \xrightarrow{\mathrm{BW}} \mathcal{N}$ in terms of the BW topology with respect to $i$,
there exists $i_{\epsilon, A} \in I$ such that for any $i \geqq i_{\epsilon, A}$
\begin{equation}
\label{eq:bw_convergence_N}
  \max_{k\in\left\{1,\ldots,N_\epsilon\right\}}\left|  \tr \left( \mathcal{N}_{n(i)} \left(\chi_k\right) A \right) - \tr \left(\mathcal{N} \left(\chi_k \right) A\right) \right|
  < \epsilon.
\end{equation}
Thus, for any $i \geqq i_{\epsilon , A} $ and any $\chi\in\mathcal{S}(\mathcal{H})$, we have
\begin{align}
  \label{eq:tr_N_chi_bound}
&\left| \tr\left( \mathcal{N}_{n(i)} \left(\chi\right) A\right) - \tr \left( \mathcal{N}\left(\chi\right)A\right) \right| \nonumber\\
&=\left| \tr\left( \mathcal{N}_{n(i)} \left(\chi-\chi_k\right) A\right) +  \tr\left( \mathcal{N}_{n(i)} \left(\chi_k\right) A\right) \right.\nonumber\\
&\quad\quad\left. - \tr \left( \mathcal{N}\left(\chi-\chi_k\right)A\right) - \tr \left( \mathcal{N}\left(\chi_k\right)A\right) \right| \nonumber\\
&\leqq\left| \tr\left( \mathcal{N}_{n(i)} \left(\chi-\chi_k\right) A\right)\right| + \left|\tr \left( \mathcal{N}\left(\chi-\chi_k\right)A\right)\right|\nonumber\\
&\quad\quad  +  \left|\tr\left( \mathcal{N}_{n(i)} \left(\chi_k\right) A\right)-\tr \left( \mathcal{N}\left(\chi_k\right)A\right) \right|\nonumber\\
&<\left| \tr\left( \mathcal{N}_{n(i)} \left(\chi-\chi_k\right) A\right)\right| + \left|\tr \left( \mathcal{N}\left(\chi-\chi_k\right)A\right)\right|+\epsilon,
\end{align}
where $\chi_k$ is an element in the finite subset $\left\{\chi_k\right\}$ of $ \mathcal{S}\left(\mathcal{H}\right)$ in~\eqref{eq:compact_subset} satisfying
\begin{equation}
  {\left\| \chi - \chi_k \right\|}_1 < \frac{\epsilon}{{\| A\|}_\infty},
\end{equation}
and we use~\eqref{eq:bw_convergence_N} in the last line.
With
\begin{equation}
 {\left\|\mathcal{N}_{n(i)}\right\|}_\infty\coloneqq\sup\left\{{\left\|\mathcal{N}_{n(i)}(T)\right\|}_1:T\in\mathcal{T}(\mathcal{H}),{\|T\|_1}\leqq 1\right\}
\end{equation}
denoting the operator norm of the linear map $\mathcal{N}_{n(i)}:\mathcal{T}(\mathcal{H})\to\mathcal{T}(\mathcal{H})$,
we have
\begin{align}
  \label{eq:N_n_i_bound}
  &\left| \tr\left( \mathcal{N}_{n(i)} \left(\chi-\chi_k\right) A\right)\right|\nonumber\\
  &\leqq{\left\| \mathcal{N}_{n(i)} \left(\chi-\chi_k\right)\right\|}_1\cdot {\left\|A\right\|}_\infty\nonumber\\
  &\leqq{\left\|\mathcal{N}_{n(i)}\right\|}_\infty\cdot {\left\|\chi-\chi_k\right\|}_1\cdot {\left\|A\right\|}_\infty\nonumber\\
  &< 1\cdot \frac{\epsilon}{{\left\|A\right\|}_\infty}\cdot{\left\|A\right\|}_\infty=\epsilon,
\end{align}
where the last inequality follows from the fact that any CPTP map $\mathcal{N}_{n(i)}$ satisfies
\begin{equation}
{\left\|\mathcal{N}_{n(i)}\right\|}_\infty=\sup{\left\|\mathcal{N}_{n(i)}(T)\right\|}_1\leqq \sup{\|T\|}_1\leqq 1.
\end{equation}
In the same way as~\eqref{eq:N_n_i_bound} by substituting $\mathcal{N}_{n(i)}$ with $\mathcal{N}$, it holds that
\begin{equation}
  \label{eq:N_bound}
  \left| \tr\left( \mathcal{N} \left(\chi-\chi_k\right) A\right)\right|<\epsilon.
\end{equation}
Therefore, applying~\eqref{eq:N_n_i_bound} and~\eqref{eq:N_bound} to~\eqref{eq:tr_N_chi_bound}, for any $i \geqq i_{\epsilon , A} $ and any $\chi\in\mathcal{S}(\mathcal{H})$, we have
\begin{equation}
\label{eq1}
\left| \tr\left( \mathcal{N}_{n(i)} \left(\chi\right) A\right) - \tr \left( \mathcal{N}\left(\chi\right)A\right) \right| < \epsilon+\epsilon+\epsilon
= 3 \epsilon.
\end{equation}

Consequently, for any $ i \geqq i_{\epsilon , A}$,
we obtain
\begin{align}
       \label{eq2}
       &\left| \tr \left(\left(\psi-\mathcal{N}\left(\phi\right)\right) A\right)\right| \nonumber \\
       &=\left| \tr \left(\psi A\right) - \tr \left( \mathcal{N}\left(\phi\right)A\right)\right| \nonumber \\
       &=
       \left| \tr \left( \mathcal{N}_{n(i)} \left(\phi_{n(i)} \right)A\right) - \tr \left( \mathcal{N} \left(\phi\right)A\right) \right|
       \nonumber \\
       & \leqq 
       \left| \tr \left(\mathcal{N}_{n(i)} \left(\phi_{n(i)}  \right)A\right) -\tr \left( \mathcal{N}\left( \phi_{n(i)} \right)A\right) \right| \nonumber \\
       &\quad\quad+
       \left|  \tr \left( \mathcal{N}\left( \phi_{n(i)} \right)A\right)-  \tr \left( \mathcal{N}\left( \phi\right)A\right) \right|
       \nonumber \\
       &< 3 \epsilon +{\left\|\mathcal{N}\right\|}_\infty\cdot{\left\| \phi_{n(i)} - \phi \right\|}_1\cdot {\left\| A \right\|}_\infty\to 3 \epsilon,
\end{align}
where the last inequality follows from~\eqref{eq1} by substituting $\chi$ with $\phi_{n(i)}$ and from the inequality shown in the same way as~\eqref{eq:N_n_i_bound}
\begin{align}
  &\left|  \tr \left( \mathcal{N}\left( \phi_{n(i)} \right)A\right)-  \tr \left( \mathcal{N}\left( \phi\right)A\right) \right|\nonumber\\
  &\leqq{\left\|\mathcal{N}\right\|}_\infty\cdot{\left\| \phi_{n(i)} - \phi \right\|}_1\cdot {\left\| A \right\|}_\infty,
\end{align}
and the limit in the last line in terms of $i$ yields $\left\| \phi_{n(i)} -\phi \right\|_1 \to 0$.
Since $\epsilon >0$ and $A \in \mathcal{B}\left(\mathcal{H}\right)\setminus\{0\}$ are arbitrary,
this shows $\psi = \mathcal{N} (\phi)$.
\end{proof}

\begin{remark}\label{remark_minimal}
In a similar manner, we can prove that the set of minimal elements
	\begin{equation}\label{minimal}
		\left\{\phi \in \mathcal{S}\left(\mathcal{H}\right): \forall \psi \in \mathcal{S}\left(\mathcal{H}\right), \phi \succeq \psi \Rightarrow \psi \succeq \phi \right\}
	\end{equation}
is not empty as well. 
The set of minimal elements is considered as the set of the least resourceful states.
If the set of free states, which is defined in the following section, is not empty, the set of minimal set is identical to the set of free states.
However, the set of minimal elements and that of free states may be different because the set of free states can be empty for some QRTs as we will show in Example~\ref{ex:GKPcode}. 
\end{remark}

\subsection{Free States}\label{Subsection2_b}
In this section, we analyze properties of \textit{free states}.
A free state is defined as a state that can be generated from any other state by a free operation. Let $\mathcal{F}\left(\mathcal{H}\right)$ denote the set of free states; that is,
\begin{equation}\label{eq:free_state}
	\begin{aligned}
  \mathcal{F}\left(\mathcal{H}\right) \coloneqq 
	&\Big\{\psi\in\mathcal{S}\left(\mathcal{H}\right):\forall \mathcal{H}', \forall\phi\in\mathcal{S}\left(\mathcal{H}'\right), \\
	&\exists \mathcal{N}\in\mathcal{O}\left(\mathcal{H}'\to\mathcal{H}\right) \,\mathrm{s.t.}\, \psi = \mathcal{N}\left(\phi\right)\Big\}.
	\end{aligned}
\end{equation}
A state $\psi \in \mathcal{S}(\mathcal{H})\setminus \mathcal{F}(\mathcal{H})$ that is not free is called a resourceful state or a resource state. Since $\tr$ is a free operation, the set of free states is equal to the set of states that can be generated from the scalar $1 \in \mathcal{S}\left(\mathbb{C}\right)$ as shown in the following proposition.
\begin{proposition}\label{proposition_free_states}
	Let $\mathcal{H}$ be a quantum system. Then, it holds that
	\begin{equation}
		\begin{aligned}
		&\mathcal{F}\left(\mathcal{H}\right) = \\
		&\Big\{\psi\in\mathcal{S}\left(\mathcal{H}\right): \exists \mathcal{N}\in\mathcal{O}\left(\mathbb{C}\to\mathcal{H}\right) \,\mathrm{s.t.}\, \psi = \mathcal{N}\left(1\right)\Big\}.
		\end{aligned}
	\end{equation}
\end{proposition}
\begin{proof}
By the definition~\eqref{eq:free_state} of $\mathcal{F}\left(\mathcal{H}\right)$, it trivially holds that 
	\begin{equation}
		\mathcal{F}\left(\mathcal{H}\right) \subseteq \Big\{\psi\in\mathcal{S}\left(\mathcal{H}\right): \exists \mathcal{N}\in\mathcal{O}\left(\mathbb{C}\to\mathcal{H}\right) \,\mathrm{s.t.}\, \psi = \mathcal{N}\left(1\right)\Big\}.
	\end{equation}
To show the converse inclusion, assume that
	\begin{equation}
		\psi \in \Big\{\psi\in\mathcal{S}\left(\mathcal{H}\right): \exists \mathcal{N}\in\mathcal{O}\left(\mathbb{C}\to\mathcal{H}\right) \,\mathrm{s.t.}\, \psi = \mathcal{N}\left(1\right)\Big\}. 
	\end{equation}
Let $\mathcal{N} \in \mathcal{O}\left(\mathbb{C}\to\mathcal{H}\right)$ be a free operation such that $\psi = \mathcal{N}\left(1\right)$. Consider an arbitrary quantum system $\mathcal{H}'$ and an arbitrary state $\phi \in \mathcal{S}\left(\mathcal{H}'\right)$. Since $\tr \in \mathcal{O}\left(\mathcal{H}'\to\mathbb{C}\right)$, it holds that
	\begin{equation}
		\psi = \mathcal{N}\circ\tr\left(\phi\right). 
	\end{equation}
Therefore, $\psi \in \mathcal{F}\left(\mathcal{H}\right)$, which yields the conclusion. 
\end{proof}

The set of free states $\mathcal{F}\left(\mathcal{H}\right)$ may be empty for some $\mathcal{H}$ while the set of minimal elements 
defined in~\eqref{minimal} is not empty as seen in Remark~\ref{remark_minimal}. For example, if the set of free operations $\mathcal{O}\left(\mathbb{C} \to \mathcal{H}\right)$ does not contain any operation for a quantum system $\mathcal{H}$, then $\mathcal{F}\left(\mathcal{H}\right) = \emptyset$. 
The following example gives a more concrete scenario, where we take the logical 2-dimensional space of the Gottesman-Kitaev-Preskill (GKP) code~\cite{Gottesman2001} as $\mathcal{S}\left(\mathbb{C}^2\right)$. In this paper, to investigate constraints and properties of QRTs
in as general a setup as possible, we do not make any assumption on whether  $\mathcal{F}\left(\mathcal{H}\right)$ is empty or not.

\begin{example}[QRT of Non-Gaussianity on GKP Code]~\label{ex:GKPcode}
  The QRT of non-Gaussianity has applications to analyzing continuous-variable quantum computation using the Gottesman-Kitaev-Preskill (GKP) code as shown in Ref.~\cite{Yamasaki2019}. 
The GKP code encodes a qubit into an infinite-dimensional oscillator of an optical mode, and the logical 2-dimensional space can be defined by dividing the Hilbert space of the bosonic mode into a logical qubit and a gauge mode~\cite{Pantaleoni2019}.
  Gaussian operations~\cite{Weedbrook2012} at a physical level suffice to implement logical Clifford gates for the GKP code~\cite{Gottesman2001}.
  Suppose that $\mathcal{S}\left(\mathbb{C}^2\right)$ is the set of logical states in the logical 2-dimensional space of the GKP code. Take the quantum operations on $\mathcal{S}\left(\mathbb{C}^2\right)$ implementable by the Gaussian operations as the free operations. Any physical state of the GKP code is non-Gaussian, and hence in this case, $\mathcal{F}$ only has the trivial element $1$; that is, $\mathcal{F}(\mathcal{H}) = \left\{1\right\}$ if $\dim\mathcal{H}=1$, and $\mathcal{F}(\mathcal{H}) = \emptyset$ otherwise.
\end{example}

The following proposition guarantees that a maximally resourceful state cannot be a free state if a resource state exists. 
\begin{proposition}\label{proposition_free_max}
Let $\mathcal{H}$ be a quantum system. Suppose that the set of resource state is not empty; that is $\mathcal{S}\left(\mathcal{H}\right)\setminus\mathcal{F}\left(\mathcal{H}\right) \neq \emptyset$. Then, it holds that 
	\begin{equation}\label{eq_prop_free_max}
		\mathcal{G}\left(\mathcal{H}\right) \cap \mathcal{F}\left(\mathcal{H}\right) = \emptyset. 
	\end{equation}
\end{proposition}

\begin{proof}
	The proof is by contradiction. To prove~\eqref{eq_prop_free_max}, assume that $\phi \in \mathcal{G}\left(\mathcal{H}\right) \cap \mathcal{F}\left(\mathcal{H}\right)$. 
Take a resource state 
\begin{equation}\label{prop_free_max_assume}
\psi \in \mathcal{S}\left(\mathcal{H}\right)\setminus\mathcal{F}\left(\mathcal{H}\right).
\end{equation}
Since $\phi \in \mathcal{F}\left(\mathcal{H}\right)$, it holds that $\psi \succeq \phi$. Then, since $\phi \in \mathcal{G}\left(\mathcal{H}\right)$, it holds that $\phi \succeq \psi$. Therefore, $\psi$ is also a free state; that is, $\psi \in \mathcal{F}\left(\mathcal{H}\right)$, which contradicts~\eqref{prop_free_max_assume}. 
\end{proof}

We can observe that some properties of the set of free states $\mathcal{F}\left(\mathcal{H}\right)$ are inherent in the set of the free operations $\mathcal{O}\left(\mathcal{H}\right)$. The compactness of $\mathcal{O}\left(\mathcal{H}\right)$ leads to the compact set of free states $\mathcal{F}\left(\mathcal{H}\right)$. If $\mathcal{O}\left(\mathcal{H}\right)$ is convex, $\mathcal{F}\left(\mathcal{H}\right)$ is also convex.

\section{Asymptotic State Conversion}\label{Section3}
In this section, we characterize the asymptotic state conversion in general quantum resource theories (QRTs). Asymptotic state conversion gives a fundamental limit of large-scale quantum information processing exploiting quantum resources, and it has been widely discussed for known QRTs~\cite{Chitambar2018}. We provide a general definition of a state conversion rate in Sec.~\ref{Subsection3_a}. In terms of the conversion rate, we find a class of resources that cannot be generated from any free state with any free operation but can be replicated infinitely by free operations. We call this state a \textit{catalytically replicable} state. We give the definition and an example of catalytically replicable states in Sec.~\ref{Subsection3_b}. In Sec.~\ref{Subsection3_c}, we formulate relations between asymptotic state conversion and one-shot state conversion that hold in general QRTs, which may have catalytically replicable states. In the following, the ceiling function is denoted by $\lceil{}\cdots{}\rceil$, and the floor function is denoted by $\lfloor{}\cdots{}\rfloor$.

\subsection{Formulation of State Conversion Rate}\label{Subsection3_a}
We recall the concept of asymptotic state conversion and provide possible two definitions of asymptotic state conversion rates. We show the equivalence of these two definitions. 

For two quantum systems $\mathcal{H}_1$ and $\mathcal{H}_2$, and two quantum states $\phi\in\mathcal{S}\left(\mathcal{H}_1\right)$ and $\psi\in\mathcal{S}\left(\mathcal{H}_2\right)$,
asymptotic state conversion from $\phi$ to $\psi$ is a task of transforming infinitely many copies of $\phi$ into as many copies of $\psi$ as possible by a sequence of free operations $\mathcal{N}_1,\mathcal{N}_2,\ldots$ within a vanishing error.
There are two possible ways to define state conversion rates from $\phi$ to $\psi$: how many $\psi$'s can be generated from a single $\phi$, and how many $\phi$'s are necessary to generate a single $\psi$. 
We write the first conversion rate as $r_\textup{conv}\left(\phi\to\psi\right)$, and the second conversion rate as $r'_\textup{conv}\left(\phi\to\psi\right)$. 
As will be shown in Theorem~\ref{thm:conversion_rates}, these two conversion rates are related to each other in such a way that $r'_\textup{conv}\left(\phi\to\psi\right)$ is the inverse of $r_\textup{conv}\left(\phi\to\psi\right)$. Therefore, we consider $r_\textup{conv}\left(\phi\to\psi\right)$ as the asymptotic state conversion rate in this paper. 

More formally, $r_\textup{conv}\left(\phi\to\psi\right)$ is defined as follows. A set of asymptotic achievable rates is defined as 
\begin{equation}
	\label{eq:conversion_rate_set}
  \begin{aligned}
  &\mathcal{R}\left(\phi\to\psi\right)\coloneqq\\
  &\Big\{r\geqq 0:\exists\left(\mathcal{N}_n\in\mathcal{O}\left(\mathcal{H}_1^{\otimes n}\to\mathcal{H}_2^{\otimes \left\lceil rn\right\rceil}\right):n\in\mathbb{N}\right),\\
  &\quad\liminf_{n\to\infty}\left\|\mathcal{N}_n\left(\phi^{\otimes n}\right)-\psi^{\otimes \left\lceil rn\right\rceil}\right\|_1 = 0\Big\},
  \end{aligned}
\end{equation}
where $\phi^{\otimes 0}\coloneqq 1$. 
Intuitively, achievable rate $r \geqq 0$ is a positive number 
for which we can generate $rn$ copies of $\psi$ from $n$ copies of $\phi$.
However, $rn$ is not necessarily an integer in general; 
therefore, in \eqref{eq:conversion_rate_set}, we regard $r$ as an achievable rate 
when we can generate $\lceil rn \rceil$ copies of $\psi$, 
which guarantees that we can obtain $rn$ or more $\psi$'s.
An asymptotic state conversion rate $r_\textup{conv}\left(\phi\to\psi\right)$ is defined as
\begin{equation}
  \label{eq:conversion_rate}
  r_\textup{conv}\left(\phi\to\psi\right)\coloneqq\sup\mathcal{R}\left(\phi\to\psi\right).
\end{equation}
Similarly, we can consider the other definition of a state conversion rate $r'_\textup{conv}\left(\phi\to\psi\right)$. 
Here, we define another set of asymptotic achievable rates
\begin{equation}
	\label{eq:another_conversion_rate_set}
  \begin{aligned}
  &\mathcal{R}'\left(\phi\to\psi\right)\coloneqq\\
  &\Big\{r\geqq 0:\exists\left(\mathcal{N}'_n\in\mathcal{O}\left(\mathcal{H}_1^{\otimes \left\lfloor rn\right\rfloor}\to\mathcal{H}_2^{\otimes n}\right):n\in\mathbb{N}\right),\\
  &\quad\liminf_{n\to\infty}\left\|\mathcal{N}'_n\left(\phi^{\otimes\left\lfloor rn\right\rfloor}\right)-\psi^{\otimes n}\right\|_1 = 0\Big\}. 
  \end{aligned}
\end{equation}
Here, achievable rate $r' \geqq 0$ is a positive number 
for which we can generate $n$ copies of $\psi$ from $r'n$ copies of $\phi$.
However, $r'n$ is not necessarily an integer in general; 
therefore, in \eqref{eq:another_conversion_rate_set}, 
we regard $r'$ as an achievable rate 
when we can generate $n$ copies of $\psi$ from $\lfloor r'n \rfloor$ copies of $\phi$, 
which guarantees that $r'n$ or more copies of $\phi$ suffice to generate $n$ copies of $\psi$.
With respect to this definition of achievable rates, an asymptotic conversion rate $r'_\textup{conv}\left(\phi\to\psi\right)$ is defined as
\begin{equation}
  \label{eq:conversion_rate_cf}
  r'_\textup{conv}\left(\phi\to\psi\right)\coloneqq\inf\mathcal{R}'\left(\phi\to\psi\right),
\end{equation}
where $r'_\textup{conv}\left(\phi\to\psi\right)$ is infinity if the set on the right-hand side is empty. 

These two conversion rates $r_\textup{conv}\left(\phi\to\psi\right)$ and $r'_\textup{conv}\left(\phi\to\psi\right)$ are related to each other as shown in the following theorem. 
Hereafter, we will use $r_\textup{conv}\left(\phi\to\psi\right)$ as the asymptotic states conversion rate rather than $r'_\textup{conv}\left(\phi\to\psi\right)$.
\begin{theorem}[Relation Between Two Conversion Rates]\label{thm:conversion_rates}
Let $\mathcal{H}$ and $\mathcal{H}'$ be quantum systems. For any states $\phi \in \mathcal{S}\left(\mathcal{H}\right)$ and $\psi \in \mathcal{S}\left(\mathcal{H}'\right)$, it holds that,
\begin{equation}
r_\textup{conv}\left(\phi\to\psi\right) = \frac{1}{r'_\textup{conv}\left(\phi\to\psi\right)},
\end{equation}
where we regard $1/0 = \infty$. 
\end{theorem}

\begin{proof}
It suffices to show that 
\begin{equation}\label{eq:R_condition1}
	r \in \mathcal{R}\left(\phi\to\psi\right) \Rightarrow \frac{1}{r} \in \mathcal{R}'\left(\phi\to\psi\right),
\end{equation}
and that
\begin{equation}\label{eq:R_condition2}
	r \in \mathcal{R}'\left(\phi\to\psi\right) \Rightarrow \frac{1}{r} \in \mathcal{R}\left(\phi\to\psi\right).
\end{equation}
First, assume that $r \in \mathcal{R}\left(\phi\to\psi\right)$ to show~\eqref{eq:R_condition1}. Choose a fixed positive real number $\epsilon >0$. Let $n$ be an arbitrary positive integer such that 
\begin{equation}\label{eq:prop1_eq1}
	\left\|\mathcal{N}_n\left(\phi^{\otimes n}\right)-\psi^{\otimes \left\lceil rn\right\rceil}\right\|_1 < \epsilon. 
\end{equation}
Let $n' =  \left\lceil rn\right\rceil$. 
Because $n \leqq \lfloor n'/r\rfloor$, we can define a free operation $\mathcal{M}_{n'}$ as the partial trace over $\lfloor n'/r\rfloor - n$ systems so that 
\begin{equation}\label{eq:prop1_eq2} 
	 \mathcal{M}_{n'}\left(\phi^{\otimes\left\lfloor n'/r\right\rfloor}\right) = \phi^{\otimes n}.
\end{equation}
From~\eqref{eq:prop1_eq1} and~\eqref{eq:prop1_eq2}, it holds that
\begin{equation}
	\left\|\mathcal{N}_n\circ\mathcal{M}_{n'}\left(\phi^{\otimes\left\lfloor n'/r\right\rfloor}\right)-\psi^{\otimes n'}\right\|_1 < \epsilon,
\end{equation}
and $1/r \in \mathcal{R}'\left(\phi\to\psi\right)$ follows. 

On the other hand, assume that $r \in \mathcal{R'}\left(\phi\to\psi\right)$ to show~\eqref{eq:R_condition2}. Choose a fixed positive real number $\epsilon >0$. Let $n$ be an arbitrary positive integer such that 
\begin{equation}\label{eq:prop1_eq3}
	\left\|\mathcal{N}'_n\left(\phi^{\otimes\left\lfloor rn\right\rfloor}\right)-\psi^{\otimes n}\right\|_1 < \epsilon
\end{equation}
Let $n' =  \left\lfloor rn\right\rfloor$. 
Because $n \geqq \lceil n'/r\rceil$, we can define a free operation $\mathcal{M}'_{n'}$ as the partial trace over $n - \lceil n'/r\rceil$ systems so that 
\begin{equation}\label{eq:prop1_eq4}
	 \mathcal{M}'_{n'}\left(\psi^{\otimes n}\right) = \psi^{\otimes \left\lceil n'/r\right\rceil}.
\end{equation}
From~\eqref{eq:prop1_eq3} and~\eqref{eq:prop1_eq4}, it holds that
\begin{equation}
	\begin{aligned}
		&\left\|\mathcal{M}'_{n'}\circ\mathcal{N}'_n\left(\phi^{\otimes n'}\right)-\psi^{\otimes \left\lceil n'/r \right\rceil}\right\|_1 \\
		&= \left\|\mathcal{M}'_{n'}\circ\mathcal{N}'_n\left(\phi^{\otimes n'}\right) - \mathcal{M}'_{n'}\left(\psi^{\otimes n}\right)\right\|_1 \\
		&\leqq \left\|\mathcal{N}'_n\left(\phi^{\otimes n'}\right) - \psi^{\otimes n}\right\|_1\\
		&< \epsilon, 
	\end{aligned}
\end{equation}
and $1/r \in \mathcal{R}\left(\phi\to\psi\right)$ follows. 
\end{proof}

\begin{remark}
	Conventionally, a conversion rate may also be defined as~\cite{Chitambar2018,Streltsov2017}
	\begin{equation}\label{eq:conversion_rate_convention}
		\tilde{r}_\textup{conv}\left(\phi\to\psi\right)\coloneqq\sup\tilde{\mathcal{R}}\left(\phi\to\psi\right), 
	\end{equation}
	where 
	\begin{equation}
		\label{eq:conversion_rate_set_convention}
	  \begin{aligned}
	  &\tilde{\mathcal{R}}\left(\phi\to\psi\right)\coloneqq\\
	  &\Big\{r\geqq 0:\exists\left(\mathcal{N}_n\in\mathcal{O}\left(\mathcal{H}_1^{\otimes n}\to\mathcal{H}_2^{\otimes \left\lfloor rn\right\rfloor}\right):n\in\mathbb{N}\right),\\
	  &\quad\liminf_{n\to\infty}\left\|\mathcal{N}_n\left(\phi^{\otimes n}\right)-\psi^{\otimes \left\lfloor rn\right\rfloor}\right\|_1 = 0\Big\},
	  \end{aligned}
	\end{equation}
	instead of a conversion rate $r_\textup{conv}\left(\phi\to\psi\right)$ defined in \eqref{eq:conversion_rate_set} and \eqref{eq:conversion_rate}.
	In \eqref{eq:conversion_rate_set_convention}, the floor function is used instead of the ceiling function used in \eqref{eq:conversion_rate_set}.

	As we show in the following, these two conversion rates $r_\textup{conv}\left(\phi\to\psi\right)$ and $\tilde{r}_\textup{conv}\left(\phi\to\psi\right)$ are identical; 
	however, considering the meaning of $r_\textup{conv}\left(\phi\to\psi\right)$ discussed below \eqref{eq:conversion_rate}, 
	we adopt $r_\textup{conv}\left(\phi\to\psi\right)$ in this paper. 
	To see that the two conversion rates are identical, 
	first suppose that $r \in \mathcal{R}\left(\phi\to\psi\right)$ for $\phi, \psi \in \mathcal{S}(\mathcal{H})$.
	Since $\lceil rn \rceil \geqq \lfloor rn \rfloor$ and since the partial trace is free, 
	it holds that $r \in \tilde{\mathcal{R}}\left(\phi\to\psi\right)$. 
	Conversely, suppose that $r \in \tilde{\mathcal{R}}\left(\phi\to\psi\right)$. 
	Then, the fact that 
	\begin{equation}
		\lfloor rn \rfloor \geqq \lceil rn \rceil - 1 = \left\lceil \left(r - \frac{1}{n}\right)n \right\rceil
	\end{equation}
	implies that $r -1/n $ can be regarded as a achievable rate with respect to $\mathcal{R}\left(\phi\to\psi\right)$. 
	Since $r - 1/n$ approaches to $r$ as $n$ becomes large, considering the definition \eqref{eq:conversion_rate_set} of $\mathcal{R}\left(\phi\to\psi\right)$, 
	it can be concluded that $r \in \mathcal{R}\left(\phi\to\psi\right)$. 
	Therefore, $r_\textup{conv}\left(\phi\to\psi\right)$ and $\tilde{r}_\textup{conv}\left(\phi\to\psi\right)$ are identical. 

	In the same way, 
	we may replace the floor function in \eqref{eq:another_conversion_rate_set} 
	with the ceiling function to obtain the identical conversion rate, 
	but we adopt the definition \eqref{eq:another_conversion_rate_set} and \eqref{eq:conversion_rate_cf}
	due to the meaning of this conversion rate discussed below \eqref{eq:another_conversion_rate_set}.
\end{remark}

Finally, we recall a useful relation of state conversion rates given in Ref.~\cite{Horodecki2003}. Let $\mathcal{H}_1$, $\mathcal{H}_2$ and $\mathcal{H}_3$ be quantum systems. Let $\rho \in \mathcal{S}\left(\mathcal{H}_1\right)$, $\sigma \in \mathcal{S}\left(\mathcal{H}_2\right)$ and $\omega \in \mathcal{S}\left(\mathcal{H}_3\right)$ be quantum states. Suppose that we first asymptotically generate $\sigma$ from $\rho$, then we generate $\omega$ from $\sigma$ to achieve conversion from $\rho$ to $\omega$. While this protocol generate $\omega$ from $\rho$, the protocol is not necessarily optimal.
In fact, it is known that 
	\begin{equation}\label{conversion_rate_ineq}
	r_\textup{conv}\left(\rho\to\omega\right) \geqq r_\textup{conv}\left(\rho\to\sigma\right)r_\textup{conv}\left(\sigma\to\omega\right).
	\end{equation}
Note that the equality of \eqref{conversion_rate_ineq} does not necessarily hold. 
For example, a bound entangled state cannot generate maximally entangled states, but needs them to be formed~\cite{Horodecki1998}.
On the other hand, it is known that the equality of \eqref{conversion_rate_ineq} holds in special cases such as the conversion between pure states in the QRT of bipartite entanglement~\cite{Bernstein1996,Wilde2017}. 

\subsection{Catalytic Replication of Resource}\label{Subsection3_b}
In this section, we analyze the replication of a resource. One of the fundamental principles of quantum mechanics is the no-cloning theorem~\cite{Wootters1982}, which shows that we cannot clone a quantum state if we do not know the description of the state. The no-cloning theorem gives a fundamental limitation of quantum mechanics, and contributes to understanding what is achievable in quantum mechanics. Similarly, to figure out what is capable in our framework of QRTs, we consider replication of a quantum resource. In the task of the replication, we generate tensor products of a resource state, where the description of the resource state is known but the operation is restricted to the set of free operations. In terms of the asymptotic state conversion, the replication of a resource is regarded as a catalytic state conversion between the same state similarly to catalytic transformation of entanglement~\cite{Jonathan1999}.

We prove that the replication of a resource has only two scenarios: we cannot replicate the resource, or we can replicate the resource infinitely. Furthermore, we find a counter-intuitive example where a resource state is replicable infinitely. Note that the infinite replication of a resource does not necessarily mean that the amount of the resource increases under free operations because the quantification of a resource depends on a resource measure, which will be discussed in detail in Sec.~\ref{Section5}.  

\begin{theorem}[Replication of State]\label{thm_duplicability}
Let $\mathcal{H}$ be a quantum system. For any state $\psi \in \mathcal{S}\left(\mathcal{H}\right)$, $ r_\textup{conv}\left(\psi\to\psi\right)$ is equal to either 1 or $+\infty$. 
\end{theorem}

\begin{proof}
It trivially holds that
\begin{equation}
	r_\textup{conv}\left(\psi\to\psi\right) \geqq 1
\end{equation}
because $\id \in \mathcal{O}\left(\mathcal{H}\right)$.

Assume that $r_\textup{conv}\left(\psi\to\psi\right) > 1$, that is, there exists $r>1$ such that $r \in \mathcal{R}\left(\psi\to\psi\right)$. To prove $r_\textup{conv}\left(\psi\to\psi\right) = \infty$, it suffices to show that 
\begin{equation}\label{belong:thm3_sufficient}
	2r-1 \in \mathcal{R}\left(\psi\to\psi\right)
\end{equation}
because if~\eqref{belong:thm3_sufficient} holds, an arbitrarily large rate can be achieved by exploiting~\eqref{belong:thm3_sufficient} repeatedly. 

Choose a fixed positive real number $\epsilon$. There exists a sequence of free operations $\left(\mathcal{N}_{n}\in \mathcal{O}\left(\mathcal{H}^{\otimes n} \to \mathcal{H}^{\otimes \left\lceil rn \right\rceil}\right) : n\in\mathbb{N}\right)$ such that 
\begin{equation}\label{eq:them2_assumption1}
	\left\|\mathcal{N}_{n}\left(\psi^{\otimes n}\right) - \psi^{\otimes \left\lceil rn \right\rceil}\right\|_1 < \frac{\epsilon}{2}
\end{equation}
holds for an infinitely large subset of $\mathbb{N}$.

Define 
\begin{equation}
	r_0 \coloneqq \inf_{n} r_n,
\end{equation}
where 
\begin{equation}
	r_n \coloneqq \max \left\{r' \geqq 0 : \left\lceil r'n \right\rceil = 2\left\lceil rn \right\rceil - n \right\},
\end{equation}
and the infimum is taken over $n$ satisfying~\eqref{eq:them2_assumption1}. 

For $n$ satisfying~\eqref{eq:them2_assumption1}, it holds that
\begin{equation}
	r_n = 2r -1 + \frac{2\alpha_n}{n},
\end{equation}
where $\alpha_n$ is a real number satisfying $0 \leqq \alpha_n < 1$ and $\left\lceil rn \right\rceil = rn + \alpha_n$. The term $2\alpha_n/n$ approaches to zero as $n$ approaches to infinity. Then, it holds that
	\begin{equation}
		\begin{aligned}
	\inf_{n} r_n 
	&= \inf_{n} \left\{2r-1+ \frac{2\alpha_n}{n}\right\}\\
	&= 2r-1. 
		\end{aligned}
	\end{equation}
Therefore, it holds that
	\begin{equation}\label{eq:them2_3}
		r_0 = 2r - 1. 
	\end{equation}
Then, due to~\eqref{belong:thm3_sufficient} and~\eqref{eq:them2_3}, it suffices to show that 
	\begin{equation}
		r_0 \in \mathcal{R}\left(\psi\to\psi\right).
	\end{equation} 

Now, observe that for any $n$, 
\begin{equation}\label{eq:them2_assumption2}
	\left\lceil rn \right\rceil > n
\end{equation}
always holds. This implies that for any $n$, 
\begin{equation}
	2\left\lceil rn \right\rceil - n > \left\lceil rn \right\rceil
\end{equation}
holds. 

For $n$ satisfying~\eqref{eq:them2_assumption1} and~\eqref{eq:them2_assumption2}, we define a free operation $\mathcal{M}_{n}$ as the partial trace over $\left(2\left\lceil rn \right\rceil - n\right) -   \left\lceil r_0n \right\rceil$ systems so that we can obtain
\begin{equation}
		\mathcal{M}_{n}\left(\psi^{\otimes 2\left\lceil rn \right\rceil - n}\right) = \psi^{\otimes \left\lceil r_0n \right\rceil}.
\end{equation}
Therefore, from the triangle inequality, it follows that
\begin{equation}\label{ineq:thm3}
	\begin{aligned}
		&\left\|\mathcal{M}_{n}\circ\left(\mathcal{N}_{n}\otimes \id \right)\circ\mathcal{N}_{n}\left(\psi^{\otimes n}\right) - \psi^{\otimes \left\lceil r_0n \right\rceil} \right\|_1\\
		&\leqq \left\|\mathcal{M}_{n}\circ\left(\mathcal{N}_{n}\otimes \id \right)\circ\mathcal{N}_{n}\left(\psi^{\otimes n}\right)\right.\\
		&\quad\quad\left.-\mathcal{M}_{n}\circ\left(\mathcal{N}_{n}\otimes \id \right)\left(\psi^{\otimes \left\lceil rn \right\rceil}\right)\right\|_1 \\
		&\quad + \left\| \mathcal{M}_{n}\circ\left(\mathcal{N}_{n}\otimes \id \right)\left(\psi^{\otimes \left\lceil rn \right\rceil}\right) - \psi^{\otimes \left\lceil r_0n \right\rceil} \right\|_1, 
	\end{aligned}
\end{equation}
where $\id$ is the identity map on $\left\lceil rn \right\rceil - n$ systems. 
Since the trace distance is non-increasing for quantum operations, it holds that 
\begin{equation}\label{ineq_2:thm3}
	\begin{aligned}
		&\left\|\mathcal{M}_{n}\circ\left(\mathcal{N}_{n}\otimes \id \right)\circ\mathcal{N}_{n}\left(\psi^{\otimes n}\right)\right.\\
		&\quad\left.-\mathcal{M}_{n}\circ\left(\mathcal{N}_{n}\otimes \id \right)\left(\psi^{\otimes \left\lceil rn \right\rceil}\right)\right\|_1 \\
		&\leqq \left\|\mathcal{N}_{n}\left(\psi^{\otimes n}\right)-\psi^{\otimes \left\lceil rn \right\rceil}\right\|_1, 
	\end{aligned}
\end{equation}
and that
	\begin{equation}\label{ineq_3:thm3}
		\begin{aligned}
		&\left\| \mathcal{M}_{n}\circ\left(\mathcal{N}_{n}\otimes \id \right)\left(\psi^{\otimes \left\lceil rn \right\rceil}\right) - \psi^{\otimes \left\lceil r_0n \right\rceil} \right\|_1 \\
		&\leqq \left\|\left(\mathcal{N}_{n}\otimes \id \right)\left(\psi^{\otimes \left(n + \left(\left\lceil rn \right\rceil - n\right)\right)}\right) - \psi^{\otimes \left(2\left\lceil rn \right\rceil - n\right)} \right\|_1\\
		&= \left\|\left(\mathcal{N}_{n}\left(\psi^{\otimes n}\right) - \psi^{\otimes \left\lceil rn \right\rceil}\right)\otimes\psi^{\otimes \left(\left\lceil rn \right\rceil - n\right)} \right\|_1\\
		&= \left\|\mathcal{N}_{n}\left(\psi^{\otimes n}\right)-\psi^{\otimes \left\lceil rn \right\rceil}\right\|_1. 
		\end{aligned}
	\end{equation}
Therefore, by~\eqref{ineq:thm3},~\eqref{ineq_2:thm3}, and~\eqref{ineq_3:thm3}, we obtain
	\begin{equation}
		\begin{aligned}
		&\left\|\mathcal{M}_{n}\circ\left(\mathcal{N}_{n}\otimes \id \right)\circ\mathcal{N}_{n}\left(\psi^{\otimes n}\right) - \psi^{\otimes \left\lceil r_0n \right\rceil} \right\|_1\\
		&\leqq 2\left\|\mathcal{N}_{n}\left(\psi^{\otimes n}\right)-\psi^{\otimes \left\lceil rn \right\rceil}\right\|_1 \\
		&< \epsilon,
	\end{aligned}
\end{equation}
which implies that $r_0 \in \mathcal{R}\left(\psi\to\psi\right)$. 
\end{proof}
Remarkably, we here give an example where $r_\textup{conv}\left(\psi\to\psi\right) = \infty$, but $\psi$ is not a free state, that is, $\psi \not\in \mathcal{F}\left(\mathcal{H}\right)$. In this paper, 
we call a state $\psi$ that satisfies $r_\textup{conv}(\psi \to \psi) = \infty$ and $\psi \notin \mathcal{F}\left(\mathcal{H}\right)$ a \textit{catalytically replicable state}.
A catalytically replicable state is regarded as a form of catalytic property of quantum resources, which are similar to catalytic state conversion in the entanglement theory~\cite{Jonathan1999}. Any free state $\psi$ is a trivial example of $r_\textup{conv}\left(\psi\to\psi\right) = \infty$, but the following example shows that this is not the whole story; that is, $r_\textup{conv}(\psi \to \psi) = \infty$ implies that $\psi$ is free or catalytically replicable. 
\begin{example}[Catalytically Replicable Resource]\label{ex4}
Suppose that $\mathcal{S}\left(\mathbb{C}^{2}\right) = \left\{\ket{0}\bra{0}, \ket{1}\bra{1}\right\}$. Further suppose that the set of free operations $\mathcal{O}$ consists of operations that are realized by circuits composed of the identity gate, the partial trace, the controlled-NOT gate, the preparation of an auxiliary qubit in $\ket{0}$ state. For any integer $n\geqq 0$, the set of free states is
\begin{equation}
  \mathcal{F}\left({\left(\mathbb{C}^2\right)}^{\otimes n}\right) = \left\{\ket{0}\bra{0}^{\otimes n}  \right\}. 
\end{equation}
In this case, whereas $\ket{1}\bra{1} \not\in \mathcal{F}\left(\mathbb{C}^2\right)$, $r_\textup{conv}\left(\ket{1}\bra{1}\to\ket{1}\bra{1}\right) = +\infty$ because we can convert $\ket{1}\bra{1}$ into $\ket{1}\bra{1}^{\otimes n}$ for any $n$ by appending an auxiliary system prepared in $\ket{0}$ and applying controlled-NOT repeatedly.  
\end{example}

\subsection{Relations Between One-Shot State Conversion and Asymptotic State Conversion}\label{Subsection3_c}
In this section, we analyze relations between the asymptotic state conversion and the exact one-shot state conversion. 
The asymptotic state conversion from $\phi$ to $\psi$ is a task transforming infinitely many copies of $\phi$ into many copies of $\psi$ with a vanishing error, while the exact one-shot conversion $\phi$ to $\psi$ is a task transforming a single $\phi$ into a single $\psi$ exactly. 

We prove two propositions both of which give relations between the asymptotic state conversion and the exact one-shot state conversion. The first proposition provides the relation that holds for inequivalent states. On the other hand, the second proposition characterizes the asymptotic conversion rate between two equivalent states.
Firstly, the following proposition shows that the more resourceful a state is, the harder it is to distill the state and the easier it is to form another state from the state. 

\begin{proposition}\label{conversion_order}
Let $\mathcal{H}_1$,$\mathcal{H}_2$ be quantum systems. Let $\phi,\psi \in \mathcal{S}\left(\mathcal{H}_1\right)$ and $\rho \in \mathcal{S}\left(\mathcal{H}_2\right)$ be quantum states. If $\phi \succeq \psi$, then it holds that
\begin{align}
	\label{ineq:lem5_assert1}
	r_\textup{conv}\left(\rho\to\phi\right) &\leqq r_\textup{conv}\left(\rho\to\psi\right)\\
	\label{ineq:lem5_assert2}
	r_\textup{conv}\left(\psi\to\rho\right) &\leqq r_\textup{conv}\left(\phi\to\rho\right).
\end{align}
\end{proposition}

\begin{proof}
To prove~\eqref{ineq:lem5_assert1}, it suffices to show that $\mathcal{R}\left(\rho\to\phi\right) \subseteq \mathcal{R}\left(\rho\to\psi\right)$. Suppose that $r \in \mathcal{R}\left(\rho\to\phi\right)$. Then, there exists a sequence of free operations $\left(\mathcal{N}_n \in \mathcal{O}\left(\mathcal{H}_2^{\otimes n} \to \mathcal{H}_1^{\otimes \left\lceil rn \right\rceil}\right) : n \in \mathbb{N}\right)$ such that for arbitrary $\epsilon > 0$,  
\begin{equation}\label{ineq:lem5_1}
	\left\|\mathcal{N}_n\left(\rho^{\otimes n}\right) - \phi^{\otimes \left\lceil rn \right\rceil} \right\|_1 < \epsilon 
\end{equation}
holds for an infinitely large subset of $\mathbb{N}$.
As $\phi \succeq \psi$, there exists a free operation $\mathcal{N}$ such that 
\begin{equation}
\mathcal{N}\left(\phi\right) = \psi. 
\end{equation}
Define a sequence of free operations $\left(\mathcal{M}_n \in \mathcal{O}\left(\mathcal{H}_2^{\otimes n} \to \mathcal{H}_1^{\otimes \left\lceil rn \right\rceil}\right) : n \in \mathbb{N}\right)$ as
\begin{equation}
	\mathcal{M}_n \coloneqq \mathcal{N}^{\otimes \left\lceil rn \right\rceil} \circ \mathcal{N}_n. 
\end{equation}
Then, for any $n$ satisfying~\eqref{ineq:lem5_1}, 
\begin{equation}
	\begin{aligned}
		&\left\|\mathcal{M}_n\left(\rho^{\otimes n}\right) - \psi^{\otimes \left\lceil rn \right\rceil} \right\|_1 \\
		&= \left\|\mathcal{N}^{\otimes \left\lceil rn \right\rceil} \circ \mathcal{N}_n\left(\rho^{\otimes n}\right) -\mathcal{N}^{\otimes \left\lceil rn \right\rceil}\left( \phi^{\otimes \left\lceil rn \right\rceil}\right) \right\|_1 \\
		&\leqq \left\|\mathcal{N}_n\left(\rho^{\otimes n}\right) -\phi^{\otimes \left\lceil rn \right\rceil} \right\|_1 \\
		&<\epsilon
	\end{aligned}
\end{equation}
holds, and this implies that $r \in \mathcal{R}\left(\rho\to\psi\right)$. 

Then, we prove~\eqref{ineq:lem5_assert2}.
Note that~\eqref{ineq:lem5_assert2} is equivalent to $r'_\textup{conv}\left(\phi\to\rho\right) \leqq r'_\textup{conv}\left(\psi\to\rho\right)$ because of Theorem~\ref{prop5}. It suffices to show that $\mathcal{R}'\left(\psi\to\rho\right) \subseteq \mathcal{R}'\left(\phi\to\rho\right)$. Suppose that $r \in \mathcal{R}'\left(\psi\to\rho\right)$. 
Then, there exists a sequence of free operations $\left(\mathcal{N}'_n \in \mathcal{O}\left(\mathcal{H}_1^{\otimes \left \lfloor rn \right \rfloor \to \mathcal{H}_2^{\otimes n }}\right) : n \in \mathbb{N}\right)$ such that for arbitrary $\epsilon > 0$, 
\begin{equation}\label{ineq:lem5_2}
	\left\|\mathcal{N}'_n\left(\psi^{\otimes \left\lfloor rn \right\rfloor}\right) - \rho^{\otimes n} \right\|_1 < \epsilon
\end{equation}
holds for an infinitely large subset of $\mathbb{N}$. As $\phi \succeq \psi$, there exists a free operation $\mathcal{N}'$ such that 
\begin{equation}
\mathcal{N}'\left(\phi\right) = \psi. 
\end{equation}
Define a sequence of free operations $\left(\mathcal{M}'_n \in \mathcal{O}\left(\mathcal{H}_1^{\otimes \left \lfloor rn \right \rfloor} \to \mathcal{H}_2^{\otimes n }\right) : n \in \mathbb{N}\right)$ as
\begin{equation}
	\mathcal{M}'_n \coloneqq \mathcal{N}'_n \circ \mathcal{N}'^{\otimes \left\lfloor rn \right\rfloor}. 
\end{equation}
Then, for any $n$ satisfying~\eqref{ineq:lem5_2}, 
\begin{equation}
	\begin{aligned}
		&\left\|\mathcal{M}'_n\left(\phi^{\otimes \left\lfloor rn \right\rfloor}\right) - \rho^{\otimes n}  \right\|_1 \\
		&= \left\|\mathcal{N}'_n \circ \mathcal{N}'^{\otimes \left\lfloor rn \right\rfloor}\left(\phi^{\otimes \left\lfloor rn \right\rfloor}\right)  - \rho^{\otimes n} \right\|_1 \\
		&\leqq \left\|\mathcal{N}'_n\left(\psi^{\otimes \left\lfloor rn \right\rfloor}\right) -\rho^{\otimes n} \right\|_1 \\
		&<\epsilon
	\end{aligned}
\end{equation}
holds and this implies that $r \in \mathcal{R}'\left(\phi\to\rho\right)$. 
\end{proof}

Next, we investigate the other relation between the asymptotic conversion and the exact one-shot conversion, \textit{i.e.}, the asymptotic state conversion between two equivalent states. Asymptotically, we may achieve conversion between states that are not convertible to each other in the one-shot state conversion. One may wonder whether we can achieve a better asymptotic conversion rate between states that are equivalent under the one-shot conversion. 
The following proposition shows that the asymptotic conversion rate for two equivalent states is equal to $1$ in a QRT without catalytically replicable states.
\begin{proposition}\label{prop15}
	Let $\mathcal{H}$ be a quantum system. Let $\psi,\phi \in \mathcal{S}\left(\mathcal{H}\right)$ be quantum states such that $\psi \sim \phi$. Suppose that $r_\textup{conv}\left(\psi\to\psi\right)=1$. Then, it holds that
	\begin{equation}
		r_\textup{conv}\left(\psi\to\phi\right) = r_\textup{conv}\left(\phi\to\psi\right) = 1.
	\end{equation}
\end{proposition}
\begin{proof}
	Since $r_\textup{conv}\left(\psi\to\psi\right) = 1$, it follows that
	\begin{equation}
		r_\textup{conv}\left(\psi\to\phi\right)r_\textup{conv}\left(\phi\to\psi\right) \leqq r_\textup{conv}\left(\psi\to\psi\right) = 1
	\end{equation}
from~\eqref{conversion_rate_ineq}.
On the other hand, since $\psi \sim \phi$, there exist free operations $\mathcal{M} \in \mathcal{O}\left(\mathcal{H}\right)$ and $\mathcal{N} \in \mathcal{O}\left(\mathcal{H}\right)$ such that
	\begin{align}
		\mathcal{M}\left(\psi\right) &= \phi\\
		\mathcal{N}\left(\phi\right) &= \psi,  
	\end{align}
which implies that 
	\begin{align}
		r_\textup{conv}\left(\psi\to\phi\right) &\geqq 1\\
		r_\textup{conv}\left(\phi\to\psi\right) &\geqq 1. 
	\end{align}
Therefore, both $r_\textup{conv}\left(\psi\to\phi\right)$ and $r_\textup{conv}\left(\phi\to\psi\right)$ must be equal to $1$. 
\end{proof}

\section{Distillable Resource and Resource Cost}\label{Section4}
In this section, we analyze properties of the distillable resource $R_\textup{D}$ and the resource cost $R_\textup{C}$, which represent how many resources can be extracted from a state and how many resources are needed to generate a state respectively. As noted in Sec.~\ref{Subsection2_c}, maximally resourceful states are not necessarily unique in general quantum resource theories (QRTs). In Sec.~\ref{Subsection4_a}, we define the distillable resource $R_\textup{D}$ as how many resourceful states can be generated from a state in the worst-case scenario, and we define the resource cost $R_\textup{C}$ as how many resourceful states are needed to generate a state in the best-case scenario. Our definition formulates distillation and formation of a resource even in cases where maximally resourceful states are not unique. 
In Sec.~\ref{Subsection4_b}, we analyze distillation and formation of catalytically replicable states. 
In Sec.~\ref{Subsection4_c}, we prove weak subadditivity of the distillable resource and the resource cost. 
In Sec.~\ref{Subsection4_d}, we further investigate the resource cost, and prove that an upper bound of the resource cost is achievable by a maximally resourceful state if the number of non-equivalent maximally resourceful states is finite.
In Sec.~\ref{Subsection4_e}, generalizing the fact that the distillable entanglement is always smaller than the entanglement cost~\cite{Donald2002}, we prove that the distillable resource is smaller than the resource cost in general QRTs without catalytically replicable states.

\subsection{Definitions of Resource Cost and Distillable Resource}\label{Subsection4_a}
In this section, we provide a formulation of distillation and formation of a resource, and give the definitions of the distillable resource and the resource cost, generalizing those in known QRTs such as bipartite entanglement~\cite{Bennett1996,Hayden2001}, coherence~\cite{Winter2016}, and athermality~\cite{Horodecki2013}.
In contrast with the definition in these previous works, our definitions are applicable to general QRTs, where maximally resourceful states are not necessarily unique. Our definition of the distillable resource represents how many resources can be generated in the worst case, and the definition of the resource cost represents how many resources are needed to form a state in the best case. 

Our formulation of distillation and formation of a resource is as follows. 
For a quantum system $\mathcal{H}$ and a state $\psi\in\mathcal{S}\left(\mathcal{H}\right)$, distillation from the state $\psi$ is a task of extracting many copies of a state $\phi\in\mathcal{S}\left(\mathcal{H}\right)$ from many copies of $\psi$, where $\phi$ is a state that is the most difficult to generate from $\psi$. More formally, distillation is regarded as state conversion from $\psi$ to a state $\phi$ for which $r_\textup{conv}\left(\psi\to\phi\right)$ takes a minimum value. 
Similarly, formation of $\psi$ is a task of generating many copies of $\psi$ from many copies of a state $\phi\in\mathcal{S}\left(\mathcal{H}\right)$ where $\phi$ is a state that can the most easily generate $\psi$. Formation is regarded as state conversion from $\phi$ to a state $\psi$, where $r_\textup{conv}\left(\phi\to\psi\right)$ takes a maximum value for $\phi$. 

The distillable resource $R_\textup{D}$ represents the amount of resource obtained by distillation; the resource cost $R_\textup{C}$ represents the amount of resource needed for formation of a state. 
Formally, the distillable resource of any state $\psi\in\mathcal{S}\left(\mathcal{H}\right)$ is defined as
\begin{equation}
  \label{eq:distillable_resource}
  R_\textup{D}\left(\psi\right)\coloneqq\inf_{\phi\in\mathcal{S}\left(\mathcal{H}\right)}\left\{r_\textup{conv}\left(\psi\to\phi\right) R_{\max}^{(\mathcal{H})}\right\},
\end{equation}
where $R_{\max}^{(\mathcal{H})} \geqq 0$ is a normalization constant. If the dimension of $\mathcal{H}$ is finite, we typically take $R_{\max}^{(\mathcal{H})}$ as the required number of qubits for representing the system $\mathcal{H}$; that is,
	\begin{equation}\label{typical_normalized_constant}
		R_{\max}^{(\mathcal{H})} = \log_2\left(\dim\mathcal{H}\right), 
	\end{equation}
where $\dim\mathcal{H}$ denotes the dimension of $\mathcal{H}$, which implies that $R_{\max}^{(\mathcal{H})}$ represents the maximum amount of a resource in $\mathcal{S}\left(\mathcal{H}\right)$.
Similarly, the resource cost of any state $\psi\in\mathcal{S}\left(\mathcal{H}\right)$ is defined as
\begin{equation}
  \label{eq:resource_cost}
  R_\textup{C}\left(\psi\right)\coloneqq\inf_{\phi\in\mathcal{S}\left(\mathcal{H}\right)}\left\{\frac{R_{\max}^{(\mathcal{H})}}{r_\textup{conv}\left(\phi\to\psi\right)}\right\}.
\end{equation}
Note that if we set the normalization constant $R_{\max}^{(\mathcal{H})} = \log_2\left(\dim\mathcal{H}\right)$ as mentioned in \eqref{typical_normalized_constant} in the QRT of bipartite entanglement, $R_\textup{D}$ and $R_\textup{C}$ reduce to the distillable entanglement and the entanglement cost~\cite{Bennett1996,Hayden2001}, respectively.

We obtain the following proposition for QRTs without catalytically replicable states, while QRTs with catalytically replicable states will be discussed in the next subsection.
This proposition provides general bounds of the distillable resource and the resource cost,
and we will also analyze achievability of the bound of the resource cost in Sec.~\ref{Subsection4_d}.
\begin{proposition}\label{proposition4_a_1}
	Let $\mathcal{H}$ be a quantum system, and let $\psi\in \mathcal{S}\left(\mathcal{H}\right)$ be a state. Suppose that $\mathcal{S}\left(\mathcal{H}\right)\setminus\mathcal{F}\left(\mathcal{H}\right) \neq \emptyset$. If $\psi$ is not a catalytically replicable state, it holds that 
\begin{align}
	\label{ineq:prop8_dist}
  &0\leqq R_\textup{D}\left(\psi\right)\leqq R_{\max}^{(\mathcal{H})},\\
	\label{ineq:prop8_cost}
  &0\leqq R_\textup{C}\left(\psi\right)\leqq R_{\max}^{(\mathcal{H})}.	
\end{align}
Especially, if $\psi$ is a free state, it holds that
	\begin{align}
	\label{ineq:prop8_dist_free}
	R_\textup{D}\left(\psi\right) &= 0,\\
	\label{ineq:prop8_cost_free}
   R_\textup{C}\left(\psi\right) &= 0.
	\end{align}
\end{proposition}
\begin{proof} 
	Note that by the definitions \eqref{eq:distillable_resource} of $R_\textup{D}$ and \eqref{eq:resource_cost} of $R_\textup{C}$, $0 \leqq R_\textup{D}\left(\psi\right)$ and $0 \leqq R_\textup{C}\left(\psi\right)$ trivially hold.
	First, we prove the statement for a free state. Let $\psi \in \mathcal{F}\left(\mathcal{H}\right)$ be a free state. Since the set of free states is closed, for any resource state $\phi \in \mathcal{S}\left(\mathcal{H}\right)\setminus\mathcal{F}\left(\mathcal{H}\right)$, it holds that 
\begin{equation}
	r_\textup{conv}\left(\psi \to \phi\right) = 0. 
\end{equation}
Therefore, it holds that
\begin{equation}
		0\leqq R_\textup{D}\left(\psi\right) \leqq R_{\max}^{(\mathcal{H})}r_\textup{conv}\left(\psi \to \phi\right) = 0, 
	\end{equation}
which shows $R_\textup{D}\left(\psi\right) = 0$. 
On the other hand, by Proposition~\ref{proposition_free_states}, there exists a free operation $\mathcal{N} \in \mathcal{O}\left(\mathbb{C}\to\mathcal{H}\right)$ such that $\mathcal{N}\left(1\right) = \psi$. Therefore, it holds that
	\begin{equation}
		\mathcal{N}^{\otimes n}\left(1\right) = \psi^{\otimes n}
	\end{equation}
for any positive integer $n$, which implies that
	\begin{equation}
		r_\textup{conv}\left(1 \to \psi\right) = \infty. 
	\end{equation}
Therefore, it holds that
	\begin{equation}
		0\leqq R_\textup{C}\left(\psi\right) \leqq \frac{R_{\max}^{(\mathcal{H})}}{r_\textup{conv}\left(1 \to \psi\right)} = 0, 
	\end{equation}
which shows $R_\textup{C}\left(\psi\right) = 0$. 
	
Next, we prove~\eqref{ineq:prop8_dist} for a resource state $\psi \in \mathcal{S}\left(\mathcal{H}\right)\setminus\mathcal{F}\left(\mathcal{H}\right)$. By Theorem~\ref{thm_duplicability}, we have $r_\textup{conv}\left(\psi\to\psi\right) = 1$ because $\psi$ is not a catalytically replicable state. 
Then, from~\eqref{eq:distillable_resource}, we obtain
\begin{equation}
	\begin{aligned}
		R_\textup{D}\left(\psi\right) 
		&\leqq R_{\max}^{(\mathcal{H})}r_\textup{conv}\left(\psi\to\psi\right) \\
		&= R_{\max}^{(\mathcal{H})}. 
	\end{aligned}	
\end{equation}
We can show~\eqref{ineq:prop8_cost} by replacing $R_\textup{D}$ with $R_\textup{C}$ and $r_\textup{conv}\left(\psi\to\psi\right)$ with $1/r_\textup{conv}\left(\psi\to\psi\right)$ respectively in the proof of~\eqref{ineq:prop8_dist}. 
\end{proof}

In fact, from the relation between the preorder introduced by the free operations and the asymptotic conversion rate shown in Proposition~\ref{conversion_order}, we obtain the following theorem, which shows that it is sufficient to take the infimum over the maximally resourceful states in the definitions of $R_\textup{D}$ and $R_\textup{C}$, rather than the infimum over the whole set of states.
\begin{theorem}[Maximally Resourceful States are Sufficient for Distillable Resource and Resource Cost]\label{theorem11}
Let $\psi \in \mathcal{S}\left(\mathcal{H}\right)$ be an arbitrary state. It is sufficient to consider $\mathcal{G}\left(\mathcal{H}\right)$ instead of $\mathcal{S}\left(\mathcal{H}\right)$ when we take the infimum in the definitions~\eqref{eq:distillable_resource} and~\eqref{eq:resource_cost} of $R_\textup{D}$ and $R_\textup{C}$; that is, it holds that
\begin{align}
	\label{eq:thm6_assert2}
	R_\textup{D}\left(\psi\right) &= \inf_{\phi\in\mathcal{G}\left(\mathcal{H}\right)}\left\{r_\textup{conv}\left(\psi\to\phi\right) R_{\max}^{(\mathcal{H})}\right\},\\
	\label{eq:thm6_assert1}
	R_\textup{C}\left(\psi\right) &= \inf_{\phi \in \mathcal{G}\left(\mathcal{H}\right)} \left\{\frac{R_{\max}^{(\mathcal{H})}}{r_\textup{conv}\left(\phi\to\psi\right)}\right\}.
\end{align}
\end{theorem}
\begin{proof}
To show~\eqref{eq:thm6_assert1}, it suffices to show that 
\begin{equation}\label{eq:thm6_assert1-2}
R_\textup{D}\left(\psi\right) = \inf_{\phi \in \mathcal{G}\left(\mathcal{H}\right)} \left\{r'_\textup{conv}\left(\phi\to\psi\right)R_{\max}^{(\mathcal{H})}\right\}. 
\end{equation}
By Proposition~\ref{conversion_order} and Theorem~\ref{prop5},
for any state $\rho \in \mathcal{S}\left(\mathcal{H}\right)$, there always exists a maximally resourceful state $\phi \in \mathcal{G}\left(\mathcal{H}\right)$ such that 
\begin{equation}
	r'_\textup{conv}\left(\phi\to\psi\right)R_{\max}^{(\mathcal{H})} \leqq r'_\textup{conv}\left(\rho\to\psi\right)R_{\max}^{(\mathcal{H})}. 
\end{equation}
Therefore,~\eqref{eq:thm6_assert1-2} holds. Equation~\eqref{eq:thm6_assert2} can be shown by replacing $R_\textup{C}$ with $R_\textup{D}$ and $r'_\textup{conv}\left(\rho\to\psi\right)$ with $r_\textup{conv}\left(\rho\to\psi\right)$ in~\eqref{eq:thm6_assert1-2}.
\end{proof}

\begin{remark}
Due to Theorem~\ref{theorem11}, the infimum in the definitions of the distillable resource and the resource cost is achieved in the following cases. 
Let 
	\begin{equation}\label{def:equivalence_class}
		\mathcal{G}\left(\mathcal{H}\right)/\sim \coloneqq \left\{C_\phi : \phi \in \mathcal{G}\left(\mathcal{H}\right)\right\}
	\end{equation}
be the set of equivalence classes of the maximally resourceful states, where 
	\begin{equation}
		C_\phi \coloneqq \left\{\psi \in \mathcal{G}\left(\mathcal{H}\right) : \psi \sim \phi\right\}
	\end{equation}
is the equivalence class of $\phi$. Suppose that the number of non-equivalent maximally resourceful states is finite; that is, $\left|\mathcal{G}\left(\mathcal{H}\right)/\sim\right| < \infty$. For example, in the QRT of bipartite entanglement, $\left|\mathcal{G}\left(\mathcal{H}\right)/\sim\right| = 1$; in the QRT of magic states for qutrits, $\left|\mathcal{G}\left(\mathcal{H}\right)/\sim\right| = 2$~\cite{Veitch2014}.
In these cases, the infimum is achievable by a maximally resourceful state because of Proposition~\ref{prop15}.
Thus, for these existing QRTs, we can actually replace the infimum in the definitions of the distillable resource~\eqref{eq:distillable_resource} and the resource cost~\eqref{eq:resource_cost} with the minimum, while further research is needed to clarify whether or not we can replace the infimum with the minimum for QRTs with infinitely many non-equivalent maximally resourceful states, \textit{i.e.}, $\left|\mathcal{G}\left(\mathcal{H}\right)/\sim\right| = \infty$.
\end{remark}

By using Theorem~\ref{theorem11}, we here prove that the more resourceful a state is, the larger the distillable resource and the resource cost of the state are. 

\begin{proposition}\label{prop_thm9}
	For a quantum system $\mathcal{H}$, let $\psi, \phi \in \mathcal{S}\left(\mathcal{H}\right)$ be quantum states such that $\phi\succeq\psi$. Then, it holds that 
	\begin{align}
		R_\textup{D}\left(\psi\right) &\leqq R_\textup{D}\left(\phi\right),\\
		R_\textup{C}\left(\psi\right) &\leqq R_\textup{C}\left(\phi\right).
	\end{align}
\end{proposition}

\begin{proof}
Let $\epsilon$ be an arbitrary positive number. Due to Theorem~\ref{theorem11}, we can take a maximally resourceful state $\rho \in \mathcal{G}\left(\mathcal{H}\right)$ such that
	\begin{equation}
		R_\textup{D}\left(\phi\right) + \epsilon \geqq R_{\max}^{(\mathcal{H})}r_\textup{conv}\left(\phi\to\rho\right). 
	\end{equation}
Then, by Proposition~\ref{conversion_order}, 
	\begin{align}
		R_\textup{D}\left(\psi\right) 
		&= \inf_{\sigma \in \mathcal{G}\left(\mathcal{H}\right)} \left\{R_{\max}^{(\mathcal{H})}r_\textup{conv}\left(\psi\to\sigma\right)\right\}\\
		&\leqq R_{\max}^{(\mathcal{H})}r_\textup{conv}\left(\psi\to\rho\right)\\
		&\leqq R_{\max}^{(\mathcal{H})}r_\textup{conv}\left(\phi\to\rho\right)\\
		&\leqq R_\textup{D}\left(\phi\right) + \epsilon
	\end{align}
holds. As we can take an arbitrarily small $\epsilon$, $R_\textup{D}\left(\psi\right) \leqq R_\textup{D}\left(\phi\right)$ holds.

Similarly, due to Theorem~\ref{theorem11}, we can take a maximally resourceful state $\rho \in \mathcal{G}\left(\mathcal{H}\right)$ such that
	\begin{equation}
		R_\textup{C}\left(\phi\right) + \epsilon \geqq \frac{R_{\max}^{(\mathcal{H})}}{r_\textup{conv}\left(\rho\to\phi\right)}. 
	\end{equation}
Then, by Proposition~\ref{conversion_order}, 
	\begin{align}
		R_\textup{C}\left(\psi\right) 
		&= \inf_{\sigma \in \mathcal{G}\left(\mathcal{H}\right)} \left\{\frac{R_{\max}^{(\mathcal{H})}}{r_\textup{conv}\left(\sigma\to\psi\right)}\right\}\\
		&\leqq \frac{R_{\max}^{(\mathcal{H})}}{r_\textup{conv}\left(\rho\to\psi\right)}\\
		&\leqq \frac{R_{\max}^{(\mathcal{H})}}{r_\textup{conv}\left(\rho\to\phi\right)}\\
		&\leqq R_\textup{C}\left(\phi\right) + \epsilon
	\end{align}
holds. As we can take an arbitrarily small $\epsilon$, $R_\textup{C}\left(\psi\right) \leqq R_\textup{C}\left(\phi\right)$ holds.
\end{proof}

\subsection{Distillable Resource and Resource Cost of Catalytically Replicable States}\label{Subsection4_b}
In this section, we analyze the distillable resource and the resource cost of a catalytically replicable state. As the conversion rate between a catalytically replicable state is infinite, we obtain a counter-intuitive result, which shows that an infinitely large number of a resource can be distilled from a catalytically replicable state and that a catalytically replicable state can be generated without any cost.

The following proposition shows that the resource cost needed to form a catalytically replicable state is equal to zero. Moreover, if the distillable resource of a catalytically replicable state is nonzero, an infinite amount of a resource can be distilled from the state. 
\begin{proposition}\label{proposition4_b_1}
Let $\psi \in \mathcal{S}\left(\mathcal{H}\right)$ be a state satisfying $r_\textup{conv}\left(\psi \to \psi \right) = \infty$. Then,
\begin{equation}
	R_\textup{C}\left(\psi\right) = 0.
\end{equation}
holds. Moreover, if $R_\textup{D}\left(\psi\right) > 0$, 
\begin{equation}
R_\textup{D}\left(\psi\right) = \infty
\end{equation}
holds.
\end{proposition}

\begin{proof}
	Note that $0 \leqq R_\textup{C}\left(\psi\right)$ and $0 \leqq R_\textup{D}\left(\psi\right)$ hold by the definitions. Since $r_\textup{conv}\left(\psi\to\psi\right) = \infty$, 
\begin{equation}
	\begin{aligned}
		R_\textup{C}\left(\psi\right) 
		&\leqq \frac{R_{\max}^{(\mathcal{H})}}{r_\textup{conv}\left(\psi\to\psi\right)} \\
		&= 0
	\end{aligned}
\end{equation}
holds. Therefore, it holds that $R_\textup{C}\left(\psi\right) = 0$.

Recall that for quantum states $\rho$, $\sigma$ and $\omega$, it holds that $r_\textup{conv}\left(\rho\to\omega\right) \geqq r_\textup{conv}\left(\rho\to\sigma\right)r_\textup{conv}\left(\sigma\to\omega\right)$ as shown in~\eqref{conversion_rate_ineq}. 
Take an arbitrary positive number $\epsilon$. Let $\phi \in \mathcal{S}\left(\mathcal{H}\right)$ be a state such that
\begin{equation}
R_\textup{D}\left(\psi\right) + \epsilon  \geqq r_\textup{conv}\left(\psi\to\phi\right)R_{\max}^{(\mathcal{H})} {\geqq R_\textup{D}\left(\psi\right)}.
\end{equation}
Then, it holds that
	\begin{align}
		R_\textup{D}\left(\psi\right) + \epsilon
		& \geqq r_\textup{conv}\left(\psi\to\phi\right)R_{\max}^{(\mathcal{H})} \\
		& \geqq  r_\textup{conv}\left(\psi\to\psi\right)r_\textup{conv}\left(\psi\to\phi\right)R_{\max}^{(\mathcal{H})}\\
		& {\geqq}  r_\textup{conv}\left(\psi\to\psi\right)R_\textup{D}\left(\psi\right)\\
		& = \infty. 
	\end{align} 
As we can take an arbitrarily small $\epsilon$, it holds that $R_\textup{D}\left(\psi\right) = \infty$. 
\end{proof}

In a QRT whose maximally resourceful states are catalytically replicable, there may be a state of which the distillable resource is infinite, and the resource cost is zero, as shown in the Example~\ref{ex5}. 
\begin{example}[Zero Resource Cost and Infinite Distillable Resource]\label{ex5}
As shown in Proposition~\ref{proposition4_a_1} and Proposition~\ref{proposition4_b_1}, the distillable resource of a catalytically replicable state may be infinity while that of a free state is zero. Consider the same setup as Example~\ref{ex4}. In this case,
\begin{align}
	R_\textup{D}\left(\ket{1}\bra{1}\right) &= \infty,\\
	R_\textup{C}\left(\ket{1}\bra{1}\right) &= 0
\end{align}
follows from $\mathcal{G}\left(\mathbb{C}^{2}\right) = \left\{\ket{1}\bra{1}\right\}$ and $r_\textup{conv}\left(\ket{1}\bra{1} \to \ket{1}\bra{1} \right) = \infty$. 
In contrast, for any free state $\psi$, 
Proposition~\ref{proposition4_a_1} shows that
\begin{align}
	R_\textup{D}\left(\psi\right) &= 0,\\
	R_\textup{C}\left(\psi\right) &= 0.
\end{align}
\end{example}

\subsection{Weak Subadditivity of Distillable Resource and Resource Cost}\label{Subsection4_c}
In this section, we prove that the distillable resource and the resource cost are \textit{weakly subadditive} if the Hilbert space $\mathcal{H}$ is finite-dimensional and the {normalization constant} is set as $R_{\max}^{(\mathcal{H})} = \log_2\left(\dim\mathcal{H}\right)$. The definitions of additivity and subadditivity are as follows. 
\begin{definition}[Additivity and Subadditivity]
	Let $f_\mathcal{H}$ be a family of functions from $\mathcal{S}\left(\mathcal{H}\right)$ to $\mathbb{R}$, where $\mathcal{H}$ is a quantum system. We may omit the subscript of $f_\mathcal{H}$ to write $f$ for brevity. 
	Then, $f$ is said to be {additive} if it holds that
	\begin{equation}
		f\left(\psi\otimes\phi\right) = f\left(\psi\right) + f\left(\phi\right)
	\end{equation}
for any states $\psi \in \mathcal{S}\left(\mathcal{H}\right)$ and $\phi \in \mathcal{S}\left(\mathcal{H}'\right)$. 
On the other hand, $f$ is said to be weakly additive if it holds that
	\begin{equation}\label{def:weakly_additive}
		f\left(\psi^{\otimes n}\right) = nf\left(\psi\right)
	\end{equation}
for any state $\psi \in \mathcal{S}\left(\mathcal{H}\right)$ and for any positive integer $n$. 

Similarly, $f$ is said to be {subadditive} if it holds that
	\begin{equation}
		f\left(\psi\otimes\phi\right) \leqq f\left(\psi\right) + f\left(\phi\right)
	\end{equation}
for any states $\psi \in \mathcal{S}\left(\mathcal{H}\right)$ and $\phi \in \mathcal{S}\left(\mathcal{H}'\right)$. 
On the other hand, $f$ is said to be weakly subadditive if it holds that
	\begin{equation}
		f\left(\psi^{\otimes n}\right) \leqq nf\left(\psi\right)
	\end{equation}
for any state $\psi \in \mathcal{S}\left(\mathcal{H}\right)$ and for any positive integer $n$. 
\end{definition}

Note that in some cases, \textit{e.g.}, in Ref.~\cite{Donald2002}, 
the property mentioned in \eqref{def:weakly_additive}, 
which we call \textit{weak additivity} here, 
is reffered to as \textit{additivity}. 
However, in this {paper}, we follow the convention of Ref.~\cite{Chitambar2018}.

The proof of weak subadditivity exploits the following proposition. 
\begin{proposition}\label{lem9}
	Let $\mathcal{H}$ and $\mathcal{H}'$ be quantum systems. Let $\psi \in \mathcal{S}\left(\mathcal{H}\right)$ and $\phi \in \mathcal{S}\left(\mathcal{H}'\right)$ be quantum states. Then for any $n\in\mathbb{N}$, it holds that
\begin{equation}
	r_\textup{conv}\left(\psi\to\phi\right) = r_\textup{conv}\left(\psi^{\otimes n} \to \phi^{\otimes n}\right).
\end{equation}
\end{proposition}
\begin{proof}
	It suffices to show that $\mathcal{R}\left(\psi\to\phi\right) = \mathcal{R}\left(\psi^{\otimes n} \to \phi^{\otimes n}\right)$. 
First, assume $r \in \mathcal{R}\left(\psi^{\otimes n}\to\phi^{\otimes n}\right)$ to show $\mathcal{R}\left(\psi^{\otimes n} \to \phi^{\otimes n}\right) \subseteq\mathcal{R}\left(\psi\to\phi\right)$. 
Choose a positive number $\epsilon>0$. Then, there exists a sequence of free operations $\left(\mathcal{M}^{(n)}_m \in \mathcal{O}\left(\mathcal{H}^{\otimes nm} \to \mathcal{H}'^{\otimes n\left\lceil rm \right\rceil}\right) : m \in \mathbb{N}\right)$ such that
	\begin{equation}
	  \left\|\mathcal{M}^{(n)}_m\left({\left(\psi^{\otimes n}\right)}^{\otimes m} \right) - {\left(\phi^{\otimes n}\right)}^{\otimes \left\lceil rm \right\rceil}\right\|_1 < \epsilon
	\end{equation}
holds for an infinitely large subset of $\mathbb{N}$. Since $n\left\lceil rm \right\rceil \geqq \left\lceil rnm \right\rceil$ holds, we can define $\mathcal{N}^{(n)}_{m} \in \mathcal{O}\left(\mathcal{H}'^{\otimes n\left\lceil rm \right\rceil} \to \mathcal{H}'^{\otimes \left\lceil rnm \right\rceil}\right)$ as the partial trace {over} $n\left\lceil rm \right\rceil - \left\lceil rnm \right\rceil$ systems. Then, we have $\mathcal{N}^{(n)}_{m}\left(\phi^{\otimes n\left\lceil rm \right\rceil}\right) = \phi^{\otimes \left\lceil rnm \right\rceil}$. Therefore, it holds that
	\begin{equation}
		\begin{aligned}
		&\left\|\mathcal{N}^{(n)}_m\circ\mathcal{M}^{(n)}_m\left(\psi^{\otimes nm}\right) - \phi^{\otimes \left\lceil rnm \right\rceil}\right\|_1 \\
		&= \left\|\mathcal{N}^{(n)}_m\circ\mathcal{M}^{(n)}_m\left({\left(\psi^{\otimes n}\right)}^{\otimes m}\right) - \mathcal{N}^{(n)}_m\left(\phi^{\otimes n\left\lceil rm \right\rceil}\right)\right\|_1\\
		&\leqq\left\|\mathcal{M}^{(n)}_m\left({\left(\psi^{\otimes n}\right)}^{\otimes m}\right) - \phi^{\otimes n\left\lceil rm \right\rceil}\right\|_1\\
		&<\epsilon. 
		\end{aligned}
	\end{equation}
Therefore, for an integer $k = nm$ with a sufficiently large $m$, there exists a free operation $\mathcal{L}_k \coloneqq \mathcal{N}^{(n)}_m\circ\mathcal{M}^{(n)}_m$ such that
	\begin{equation}
		\begin{aligned}
		\left\|\mathcal{L}_k\left(\psi^{\otimes k}\right) - \phi^{\otimes \left\lceil rk \right\rceil}\right\|_1 <\epsilon.
		\end{aligned}
	\end{equation} 
Therefore, $r \in \mathcal{R}\left(\psi\to\phi\right)$, which implies  $\mathcal{R}\left(\psi^{\otimes n} \to \phi^{\otimes n}\right) \subseteq\mathcal{R}\left(\psi\to\phi\right)$. 

On the other hand, to show $\mathcal{R}\left(\psi\to\phi\right) \subseteq \mathcal{R}\left(\psi^{\otimes n} \to \phi^{\otimes n}\right)$, assume $r \in \mathcal{R}\left(\psi\to\phi\right)$. Choose a positive number $\epsilon$. Then, there exists a sequence of free operations $\left(\mathcal{M}_m \in \mathcal{O}\left(\mathcal{H}^{\otimes m} \to \mathcal{H}'^{\otimes \left\lceil rm \right\rceil}\right) : m \in \mathbb{N}\right)$ such that for a fixed positive integer $n$, 
	\begin{equation}\label{ineq:lem9_conv}
		\left\|\mathcal{M}_m\left(\psi^{\otimes m} \right) - \phi^{\otimes \left\lceil rm \right\rceil}\right\|_1 < \frac{\epsilon}{n}
	\end{equation}
holds for an infinitely large subset of $\mathcal{N}$. Therefore, for any $n$ satisfying~\eqref{ineq:lem9_conv}, it holds that 
\begin{equation}
  \left\|\mathcal{M}_m^{\otimes n}\left({\left(\psi^{\otimes m}\right)}^{\otimes n} \right) - {\left(\phi^{\otimes \left\lceil rm \right\rceil}\right)}^{\otimes n}\right\|_1 < \epsilon, 
\end{equation}
which implies that $\mathcal{R}\left(\psi\to\phi\right) \subseteq \mathcal{R}\left(\psi^{\otimes n} \to \phi^{\otimes n}\right)$. 
\end{proof}

Using Proposition~\ref{lem9}, we show Theorem~\ref{prop13}. 
The meaning of this theorem will be discussed after the proof.

\begin{theorem}[Weak Subadditivity of Distillable Resource and Resource Cost]\label{prop13}
  Let $\mathcal{H}$ be an arbitrary finite-dimensional system. Set the normalization constant as $R_{\max}^{(\mathcal{H})} = \log_2\left(\dim\mathcal{H}\right)$ as {stated} in~\eqref{typical_normalized_constant}. For any $n \in \mathbb{N}$ and for any state $\psi \in \mathcal{S}\left(\mathcal{H}\right)$,
  \begin{align}
	\label{ineq:prop9_dist}
    R_\textup{D}\left(\psi^{\otimes n}\right)&\leqq nR_\textup{D}\left(\psi\right),\\
	\label{ineq:prop9_cost}
    R_\textup{C}\left(\psi^{\otimes n}\right)&\leqq nR_\textup{C}\left(\psi\right).
  \end{align}
\end{theorem}
\begin{proof}
	First, we prove~\eqref{ineq:prop9_dist}. Let $n$ be a fixed positive integer, and let $\epsilon$ be an arbitrary positive number. Due to Theorem~\ref{theorem11}, we can take a maximally resourceful state $\phi \in \mathcal{G}\left(\mathcal{H}\right)$ such that
\begin{equation}
R_\textup{D}\left(\psi\right) + \frac{\epsilon}{n}  \geqq R_{\max}^{(\mathcal{H})} r_\textup{conv}\left(\psi\to\phi\right).
\end{equation} 
Since $\phi^{\otimes n} \in \mathcal{S}\left(\mathcal{H}^{\otimes n}\right)$, 
\begin{align}
	nR_\textup{D}\left(\psi\right) + \epsilon
	&= n R_{\max}^{(\mathcal{H})}r_\textup{conv}\left(\psi\to\phi\right) \\
	\label{eq:thm13_1}
	&= R_{\max}^{(\mathcal{H}^{\otimes n})}r_\textup{conv}\left(\psi^{\otimes n}\to\phi^{\otimes n}\right)\\
	&\geqq R_\textup{D}\left(\psi^{\otimes n}\right)
\end{align}
holds where~\eqref{eq:thm13_1} follows from Proposition~\ref{lem9}. As we can take an arbitrarily small $\epsilon$,~\eqref{ineq:prop9_dist} holds.
We can show~\eqref{ineq:prop9_dist} in a similar way by replacing $R_\textup{D}$ with $R_\textup{C}$ and $r_\textup{conv}\left(\psi\to\phi\right)$ with $1/r_\textup{conv}\left(\phi\to\psi\right)$. 
\end{proof}
In the statement of Theorem~\ref{prop13},~\eqref{ineq:prop9_dist} means that for a maximally resourceful state $\phi \in \mathcal{G}\left(\mathcal{H}\right)$, there may be $\phi_n \in \mathcal{G}\left(\mathcal{H}^{\otimes n}\right)$ that is harder to distill than $\phi^{\otimes n}$. On the other hand,~\eqref{ineq:prop9_cost} means that for a maximally resourceful state $\phi \in \mathcal{G}\left(\mathcal{H}\right)$, there may be a more resourceful state $\phi_n \in \mathcal{G}\left(\mathcal{H}^{\otimes n}\right)$ in resource formation than $\phi^{\otimes n}$ for $n\geqq 2$.

\subsection{Maximally Resourceful State Maximizing Resource Cost}\label{Subsection4_d}
In this section, we prove that the upper bound $R_{\max}^{(\mathcal{H})}$ of the resource cost $R_\textup{C}$ shown in Proposition~\ref{proposition4_a_1} is indeed achievable by a maximally resourceful state if the number of equivalence classes of the maximally resourceful states is finite and if there is no catalytically replicable state.
Note that this property holds even in infinite-dimensional cases; that is, $R_{\max}^{(\mathcal{H})}=\log_2\left(\dim\mathcal{H}\right)$ for finite-dimensional $\mathcal{H}$ is not assumed in this section.

Using Proposition~\ref{prop15}, we prove Theorem~\ref{thm16}. Recall the set of equivalence classes of the maximally resourceful states $\mathcal{G}\left(\mathcal{H}\right)/\sim$ defined in~\eqref{def:equivalence_class}.
Consider a QRT where the number of maximally resourceful states is finite up to the equivalence with regard to the preorder, that is, 
	\begin{equation}
		\left|\mathcal{G}\left(\mathcal{H}\right)/\sim\right| < \infty. 
	\end{equation}
In this case, the following theorem shows that the upper bound $R_{\max}^{(\mathcal{H})}$ of the resource cost given in Proposition~\ref{proposition4_a_1} is actually achievable by a maximally resourceful state.
\begin{theorem}[Maximally Resourceful State that Maximizes Resource Cost]\label{thm16}	
Suppose that there is no catalytically replicable state. 
Suppose further that the set of resource states is not empty; that is, $\mathcal{S}\left(\mathcal{H}\right)\setminus\mathcal{F}\left(\mathcal{H}\right) \neq \emptyset$. 
If $\left|\mathcal{G}\left(\mathcal{H}\right)/\sim\right| < \infty$ where $\mathcal{G}\left(\mathcal{H}\right)/\sim$ is defined in~\eqref{def:equivalence_class}, then there exists a {maximally resourceful} state $\phi \in \mathcal{G}\left(\mathcal{H}\right)$ such that 
\begin{equation}
	R_\textup{C}\left(\phi\right) = R_{\max}^{(\mathcal{H})}
\end{equation}
holds.
\end{theorem}

\begin{proof}
	Assume that for all states $\phi \in \mathcal{G}\left(\mathcal{H}\right)$, $R_\textup{C}\left(\phi\right) < R_{\max}^{(\mathcal{H})}$. 
	Because of Theorem~\ref{theorem11}, this assumption implies that for all states $\phi \in \mathcal{G}\left(\mathcal{H}\right)$, there exists a maximally resourceful state $\rho \in \mathcal{G}\left(\mathcal{H}\right)$ such that $r_\textup{conv}\left(\rho\to\phi\right) > 1$. 
	According to Proposition~\ref{prop15}, $\rho$ must be in an equivalence class different from $C_\phi$. 
	Write this relation as $\phi \mapsto \rho$; that is, for $\phi \in \mathcal{G}\left(\mathcal{H}\right)/\sim$ and $\rho \in \mathcal{G}\left(\mathcal{H}\right)/\sim$, we write $\phi \mapsto \rho$ if $r_\textup{conv}\left(\rho\to\phi\right) > 1$. Since $\left|\mathcal{G}\left(\mathcal{H}\right)/\sim\right| < \infty$, there must exist a loop of elements in $\mathcal{G}\left(\mathcal{H}\right)/\sim$
	\begin{equation}
		\rho_0 \mapsto \rho_1 \mapsto \cdots \mapsto \rho_n \mapsto \rho_0. 
	\end{equation}
Therefore, Theorem~\ref{thm_duplicability} shows that
	\begin{equation}\label{eq:thm17_1}
		\begin{aligned}
		&r_\textup{conv}\left(\rho_0\to\rho_n\right)\times r_\textup{conv}\left(\rho_n\to\rho_{n-1}\right)\times\nonumber\\
		&\quad\cdots\times r_\textup{conv}\left(\rho_1\to\rho_0\right) \\
		&> 1.
		\end{aligned}
	\end{equation}
On the other hand, note that for any maximally resourceful state $\rho \in \mathcal{G}\left(\mathcal{H}\right)$, it holds that $r_\textup{conv}\left(\rho \to \rho\right) = 1$ because there is no catalytically replicable state and because $\rho \not\in \mathcal{F}\left(\mathcal{H}\right)$ due to Proposition~\ref{proposition_free_max}.
From~\eqref{conversion_rate_ineq}, it follows that
\begin{equation}
  \begin{aligned}
		&r_\textup{conv}\left(\rho_0\to\rho_n\right)\times \cdots\times r_\textup{conv}\left(\rho_1\to\rho_0\right) \\
		&\leqq  r_\textup{conv}\left(\rho_0\to\rho_0\right) \\
		&= 1,
  \end{aligned}
\end{equation}
which contradicts~\eqref{eq:thm17_1}.
Therefore, there exists a maximally resourceful state $\phi \in \mathcal{G}\left(\mathcal{H}\right)$ such that $R_\textup{C}\left(\phi\right) = R_{\max}^{(\mathcal{H})}$.
\end{proof}

\subsection{Condition for Distillable Resource Upper-Bounded by Resource Cost}\label{Subsection4_e}
In this section, we prove that the distillable resource is smaller than or equal to the resource cost if the QRT does not have any catalytically replicable state. 
The claim that the distillable resource is smaller than or equal to the resource cost was proved in the QRT of bipartite entanglement~\cite{Donald2002}. Progressing beyond this previous work, our proof does not make any assumption on the existence of additive measures, and hence is simpler and has more applicability than the existing technique.
\begin{theorem}[Condition for Distillable Resource Upper-Bounded by Resource Cost]
Let $\mathcal{H}$ be a quantum system. Suppose that $\mathcal{S}\left(\mathcal{H}\right)\setminus\mathcal{F}\left(\mathcal{H}\right) \neq \emptyset$. Further suppose that there is no catalytically replicable state; that is, $r_\textup{conv}\left(\phi\to\phi\right) = 1$ for any resource state $\phi \in \mathcal{S}\left(\mathcal{H}\right)\setminus\mathcal{F}\left(\mathcal{H}\right)$. Then, for any state $\psi \in \mathcal{S}\left(\mathcal{H}\right)$, it holds that
	\begin{equation}\label{ineq_theorem_dist_cost}
		R_\textup{D}\left(\psi\right) \leqq R_\textup{C}\left(\psi\right). 
	\end{equation}
\end{theorem}

\begin{proof}
	As shown in Proposition~\ref{proposition4_a_1}, for a free state $\psi \in \mathcal{F}\left(\mathcal{H}\right)$, it holds that $R_\textup{D}\left(\psi\right)  = R_\textup{C}\left(\psi\right) = 0$. Therefore,~\eqref{ineq_theorem_dist_cost} trivially holds for a free state. Then, let $\psi \in \mathcal{S}\left(\mathcal{H}\right)\setminus\mathcal{F}\left(\mathcal{H}\right)$ be an arbitrary resource state. 
Let $\epsilon > 0$ be an arbitrary positive number. Due to Theorem~\ref{theorem11}, we take a maximally resourceful state $\phi \in \mathcal{G}\left(\mathcal{H}\right)$ such that
	\begin{align}
		R_\textup{C}\left(\psi\right) + \epsilon \geqq \frac{R_{\max}^{(\mathcal{H})}}{r_\textup{conv}\left(\phi\to\psi\right)}.
	\end{align}
By the definition~\eqref{eq:distillable_resource} of the distillable resource, it holds that
	\begin{equation}
		R_\textup{D}\left(\psi\right) \leqq r_\textup{conv}\left(\psi\to\phi\right)R_{\max}^{(\mathcal{H})}.
	\end{equation}
As $r_\textup{conv}\left(\psi\to\phi\right)r_\textup{conv}\left(\phi\to\psi\right) \leqq r_\textup{conv}\left(\psi\to\psi\right) = 1$ shown in~\eqref{conversion_rate_ineq}, it holds that
	\begin{align}
		R_\textup{D}\left(\psi\right)
		&\leqq r_\textup{conv}\left(\psi\to\phi\right)R_{\max}^{(\mathcal{H})}\\
		&\leqq \frac{R_{\max}^{(\mathcal{H})}}{r_\textup{conv}\left(\phi\to\psi\right)}\\
		&\leqq R_\textup{C}\left(\psi\right)+\epsilon. 
	\end{align}
As we can take an arbitrarily small $\epsilon$, it holds that $R_\textup{D}\left(\psi\right) \leqq R_\textup{C}\left(\psi\right)$.
\end{proof}

\section{Resource Measures}\label{Section5}
In this section, we investigate a formulation of resource measures in general QRTs and clarify general properties of the resource measures. The resource measures quantify the amount of quantum resources, which is a central interest in QRTs~\cite{Chitambar2018}. In Sec.~\ref{Subsection5_a} we provide the definition of a resource measure. In Sec.~\ref{Subsection5_b}, progressing beyond the existing result on bipartite entanglement~\cite{Donald2002}, we show in the general setting that a resource measure is upper-bounded by the resource cost and lower-bounded by the distillable resource if it satisfies the same properties as those given in Ref.~\cite{Donald2002}. At the same time, we show that the QRT of magic for qutrits~\cite{Wang2018}, which has several non-equivalent maximally resourceful states, has no resource measure satisfying the {properties}. To overcome this problem in the axiomatic approach based on Ref.~\cite{Donald2002}, we here introduce a concept of \textit{consistency} of a resource measure in Sec.~\ref{Subsection5_c}. In contrast with the previous approach, a consistent resource measure exists in the case where multiple non-equivalent maximally resourceful states exist. Furthermore, we prove a similar uniqueness inequality to the previous approach; that is, the consistent resource measure is bounded by the distillable resource and the resource cost if it is normalized. In Sec.~\ref{Subsection5_d}, we provide the definition of the relative entropy of resource, and show that the regularized relative entropy of resource serves as a consistent resource measure.

\subsection{{Properties of} Resource Measures}\label{Subsection5_a}
In this section, we provide a definition of resource measure as discussed in previous works on studying resource measures in a wide class of QRTs, such as Refs.~\cite{Plenio2005,Brandao2015,Bromley2018,Gour2009,Liu2017,Regula2017}, and recall {conventionally studied properties of} a resource measure some of which are also discussed in {Refs.~\cite{Chitambar2018,Donald2002}}. 
A resource measure $R$ quantifies the amount of the resource of a state. It takes a state as an input and outputs a real number that represents the amount of the resource. 
To quantify the resource without contradicting the fact that free operations cannot generate resources by themselves, a resource measure must satisfy a property called \textit{monotonicity}; \textit{i.e.}, the amount of the resource quantified by a resource measure does not increase through application of free operations.
Formally, monotonicity is defined as follows. For quantum systems $\mathcal{H}^{(\mathrm{in})}$ and $\mathcal{H}^{(\mathrm{out})}$, any state $\psi \in \mathcal{S}\left(\mathcal{H}^{(\mathrm{in})}\right)$, and any free operation $\mathcal{N} \in \mathcal{O}\left(\mathcal{H}^{(\mathrm{in})} \to \mathcal{H}^{(\mathrm{out})}\right)$, it holds that
	\begin{equation}
	R_{\mathcal{H}^{(\mathrm{in})}}(\psi) \geqq R_{\mathcal{H}^{(\mathrm{out})}}(\mathcal{N}(\psi)). 
	\end{equation}
Here, we recall the definition of a resource measure. 
\begin{definition}[Resource Measure]
	A resource measure $R_\mathcal{H}$ is a family of real functions from $\mathcal{S}\left(\mathcal{H}\right)$ for a quantum system $\mathcal{H}$ to $\mathbb{R}$ satisfying monotonicity. We may omit the subscript of $R_\mathcal{H}$ to write $R$ for brevity.
\end{definition}
By monotonicity, a resource measure $R_\mathcal{H}$ for a quantum system $\mathcal{H}$ quantifies the resource consistently with the preorder introduced by free operations. 
For two states satisfying
	\begin{equation}
		\phi \succeq \psi,
	\end{equation}
it holds that
	\begin{equation}
		R_\mathcal{H}\left(\phi\right) \geqq R_\mathcal{H}\left(\psi\right).
	\end{equation}
Note that if two states satisfy
	\begin{equation}
		\phi \sim \psi
	\end{equation}
then we have
	\begin{equation}
		R_\mathcal{H}\left(\phi\right) = R_\mathcal{H}\left(\psi\right).
	\end{equation}
Furthermore, using a resource measure, we can evaluate the resource amounts of two different states that cannot be compared in terms of the preorder introduced by free operations. 

For example, in the QRT of bipartite entanglement, the distillable entanglement and the entanglement cost are known to be entanglement measures~\cite{Bennett1996,Hayden2001}. Generalizing this fact, we find a condition where the distillable resource and the resource cost become resource measures as shown in the following proposition. 
\begin{proposition}
Suppose that for any quantum system $\mathcal{H}$, there exists a maximally resourceful state $\phi \in \mathcal{G}(\mathcal{H})$ such that for any quantum system $\mathcal{H}'$ and for any maximally resourceful state $\phi'\in\mathcal{G}(\mathcal{H}')$, it holds that
	\begin{equation}\label{measure_eq1}
		r_\textup{conv}(\phi \to \phi') \geqq \frac{R_{\max}^{(\mathcal{H})}}{R_{\max}^{(\mathcal{H}')}}.
	\end{equation}
Then, the distillable resource $R_\textup{D}$ and the resource cost $R_\textup{C}$ satisfy monotonicity. 
\end{proposition}

\begin{proof}
Choose any $\delta > 0$. Let $\mathcal{H}_1$ and $\mathcal{H}_2$ be quantum systems. Let $\psi \in \mathcal{S}(\mathcal{H}_1)$ be a quantum state and $\mathcal{N} \in\mathcal{O}(\mathcal{H}_1 \to \mathcal{H}_2)$ be a free operation. Suppose that $\phi_2 \in \mathcal{G}(\mathcal{H}_2)$ is a maximally resourceful state satisfying the assumption given in \eqref{measure_eq1}. 

Now, suppose that $\phi_1 \in \mathcal{G}(\mathcal{H}_1)$ is a maximally resourceful state such that
	\begin{equation}\label{distillation_measure_0}
		R_\textup{D}(\psi) \geqq R_{\max}^{(\mathcal{H}_1)}r_\textup{conv}(\psi \to \phi_1) -\delta. 
	\end{equation}
Then, it holds that
	\begin{equation}\label{distillation_measure_1}
		\begin{aligned}
			R_\textup{D}(\mathcal{N}(\psi)) 
			&\leqq R_{\max}^{(\mathcal{H}_2)}r_\textup{conv}(\mathcal{N}(\psi) \to \phi_2)\\
			&\leqq R_{\max}^{(\mathcal{H}_2)}r_\textup{conv}(\psi \to \phi_2)\\
			&\leqq R_{\max}^{(\mathcal{H}_1)}r_\textup{conv}(\psi \to \phi_2)r_\textup{conv}(\phi_2 \to \phi_1)\\
			&\leqq R_{\max}^{(\mathcal{H}_1)}r_\textup{conv}(\psi \to \phi_1)\\
			&\leqq R_\textup{D}(\psi) + \delta. 
		\end{aligned}
	\end{equation}
The first inequality holds by definition. The second inequality follows from Proposition \ref{conversion_order}. The third inequality follows from \eqref{measure_eq1}. Using \eqref{conversion_rate_ineq}, we have the fourth inequality. Since we can take an arbitrarily small $\delta$, we have $R_\textup{D}(\mathcal{N}(\psi)) \leqq  R_\textup{D}(\psi)$. 

Next, suppose that $\phi_1 \in \mathcal{G}(\mathcal{H}_1)$ is a maximally resourceful state such that
	\begin{equation}\label{cost_measure_0}
		R_\textup{C}(\psi) \geqq \frac{R_{\max}^{(\mathcal{H}_1)}}{r_\textup{conv}(\phi_1 \to \psi)} - \delta. 
	\end{equation}
Then, {it holds that}
	\begin{equation}
		\begin{aligned}
			R_\textup{C}(\mathcal{N}(\psi)) 
			&\leqq \frac{R_{\max}^{(\mathcal{H}_2)}}{r_\textup{conv}(\phi_2 \to\mathcal{N}(\psi))}\\
			&\leqq \frac{R_{\max}^{(\mathcal{H}_2)}}{r_\textup{conv}(\phi_2 \to\psi)}\\
			&\leqq \frac{R_{\max}^{(\mathcal{H}_1)}r_\textup{conv}(\phi_2 \to\phi_1)}{r_\textup{conv}(\phi_2 \to\psi)}\\
			&\leqq \frac{R_{\max}^{(\mathcal{H}_1)}}{r_\textup{conv}(\phi_1 \to\psi)}\\
			&\leqq R_\textup{C}(\psi) + \delta. 
		\end{aligned}
	\end{equation}
Since we can take an arbitrarily small $\delta$, we have $R_\textup{C}(\mathcal{N}(\psi)) \leqq  R_\textup{C}(\psi)$. 
\end{proof}
Now, we recall several {properties of a resource measure that are conventionally studied}. 

\textit{Additivity}:
Strong superadditivity refers to
\begin{equation}
  R\left(\psi^{AB}\right)\geqq R\left(\psi^A\right)+R\left(\psi^B\right),
\end{equation}
and {superadditivity} refers to
\begin{equation}
  R\left(\psi\otimes\phi\right)\geqq R\left(\psi\right)+R\left(\phi\right).
\end{equation}
{Subadditivity} refers to
\begin{equation}
  R\left(\psi\otimes\phi\right)\leqq R\left(\psi\right)+R\left(\phi\right).
\end{equation}
{Additivity} refers to
\begin{equation}
  R\left(\psi\otimes\phi\right)= R\left(\psi\right)+R\left(\phi\right), 
\end{equation}
while {weak additivity} refers to 
\begin{equation}
  R\left(\psi^{\otimes n}\right)= nR\left(\psi\right). 
\end{equation}
Regularization of $R$ provides a measure that is additive for tensor product of the same states
\begin{equation}
  R^{\infty}\left(\psi\right)\coloneqq\lim_{n\to\infty}\frac{R\left(\psi^{\otimes n}\right)}{n},
\end{equation}
as long as the right-hand side exists.
The following proposition shows that {additivity} of a resource measure implies that free states have zero resource, which is a generalization of the statement shown for entanglement in Ref.~\cite{Donald2002} to general QRTs.

\begin{proposition}\label{prop_additive_free}
If a resource measure $R$ is {weakly additive}, $R\left(\phi\right) = 0$ for any free state $\phi$.
\end{proposition}

\begin{proof}
	Suppose that $\psi$ is a free state. Then, there exists a free operation $\mathcal{M}$ such that $\mathcal{M}\left(1\right) = \psi$. Therefore, for any $n \in \mathbb{N}$, $\mathcal{M}^{\otimes n}\left(1\right) = \psi^{\otimes n}$ holds, which implies $ \psi^{\otimes n}$ is also a free state. Then, there exist free operations $\mathcal{N}_1$ and $\mathcal{N}_2$ such that
\begin{equation}
	\begin{aligned}
		&\mathcal{N}_1\left(\psi^{\otimes n}\right) = \psi, \\
		&\mathcal{N}_2\left(\psi \right) = \psi^{\otimes n}
	\end{aligned}
\end{equation}
hold. Therefore, it holds that $R\left(\psi\right) = R\left(\psi^{\otimes n}\right)$. 
Since $R$ is {weakly additive}, $R\left(\psi\right) = n R\left(\psi\right)$ for any $n$, which implies $R\left(\psi\right) = 0$. 
\end{proof}

One conventional way of normalizing resource measures such as that in the entanglement theory is as follows, which we call \textit{conventional normalization}:
\begin{itemize}
\item For any free state $\sigma$, 
	\begin{equation}\label{conventional_normalization_zero}
  R\left(\sigma\right)=0.
\end{equation}
\item For any {maximally resourceful state} $\phi\in\mathcal{G}\left(\mathcal{H}\right)$, 
\begin{equation}\label{conventional_normalization_max}
  R\left(\phi\right)=R_{\max}^{(\mathcal{H})}.
\end{equation}
\end{itemize}
Because of monotonicity of a resource measure, a resource measure takes the least value for free states. From Proposition~\ref{prop_additive_free}, the least value is automatically set to zero for an additive measure. We can assume this normalization also for non-additive measures. Furthermore, we can set the greatest value of a resource measure to $R_{\max}^{(\mathcal{H})}$ in the same way that we normalize the distillable resource~\eqref{eq:distillable_resource} and the resource cost~\eqref{eq:resource_cost} with the {normalization constant} $R_{\max}^{(\mathcal{H})}$.  
In a finite-dimensional case with $R_{\max}^{(\mathcal{H})} = \log_2\left(\dim\mathcal{H}\right)$ shown in~\eqref{typical_normalized_constant},~\eqref{conventional_normalization_max} provides the normalization generalizing that of entanglement measures in the entanglement theory, but our definition is applicable to infinite-dimensional cases. 
We here remark that general QRTs do not necessarily have resource measures satisfying this conventional normalization, as we will prove in the next subsection.

\textit{Asymptotic continuity}:
For any quantum system $\mathcal{H}$, $R_{\mathcal{H}}$ is \textit{asymptotically continuous} if for any sequence of positive integers ${\left(n_i\right)}_{i\in\mathbb{N}}$, and any sequences of states ${\left(\phi_{n_i} \in \mathcal{S}\left(\mathcal{H}^{\otimes n_i}\right)\right)}_i$ and ${\left(\psi_{n_i} \in \mathcal{S}\left(\mathcal{H}^{\otimes n_i}\right)\right)}_i$ satisfying $\lim_{i \to \infty} \left\|\phi_{n_i}-\psi_{n_i}\right\|_1 =0$, it holds that 
\begin{equation}\label{asymptotic_continuity}
  \lim_{i\to\infty}\frac{\left|R_{\mathcal{H}^{\otimes n_i}}\left(\phi_{n_i}\right) - R_{\mathcal{H}^{\otimes n_i}}\left(\psi_{n_i}\right)\right|}{n_i} = 0.
\end{equation}
Our definition of asymptotic continuity is applicable to an infinite-dimensional system. If we take $R_{\max}^{(\mathcal{H})} = \log_2\left(\dim\mathcal{H}\right)$ as {stated} in~\eqref{typical_normalized_constant} for a finite-dimensional system $\mathcal{H}$, our definition~\eqref{asymptotic_continuity} corresponds to Condition (E3) in Ref.~\cite{Donald2002}. 
Note that our definition includes asymptotic continuity discussed in~\cite{Synak2006} (Definition 1) as a tighter bound applicable to a finite-dimensional system.  

\begin{remark}[Continuity and Asymptotic Continuity]
Since a resource measure is a family of functions each of which may be defined for different quantum systems, we employ asymptotic continuity of a family of functions as {a desired property} {of} resource measures rather than {continuity} of a single function $R_\mathcal{H}$ for a fixed quantum system $\mathcal{H}$ defined as follows. A function $R_\mathcal{H}: \mathcal{S}\left(\mathcal{H}\right) \to \mathbb{R}$ is \textit{continuous} if for any sequences of states ${\left(\phi_{n} \in \mathcal{S}\left(\mathcal{H}\right)\right)}_n$ and ${\left(\psi_{n} \in \mathcal{S}\left(\mathcal{H}\right)\right)}_n$ satisfying $\lim_{n \to \infty} \left\|\phi_{n}-\psi_{n}\right\|_1 =0$, it holds that 
	\begin{equation}\label{continuity}
		\lim_{n\to\infty} \left|R_\mathcal{H}\left(\phi_n\right) - R_\mathcal{H}\left(\psi_n\right)\right| = 0. 
	\end{equation}
Our definition of asymptotic continuity implies continuity as a special case.
\end{remark}

\subsection{Generalization of Uniqueness Inequality}\label{Subsection5_b}
In this section, we show that we have the inequality $R_\textup{D} \leqq R \leqq R_\textup{C}$ for a resource measure $R$ if $R$ satisfies conventional normalization, asymptotic continuity, and {weak additivity}. We call this inequality the uniqueness inequality. 
The uniqueness inequality is originally proved in the QRT of bipartite entanglement in finite-dimensional cases~\cite{Donald2002}. We show that the proof of this uniqueness inequality can be generalized to all the QRTs in our framework that covers infinite-dimensional cases.
At the same time, we also show a QRT in which no resource measure satisfies these {properties}; that is, the set of states satisfying the uniqueness inequality becomes empty. 

We prove the following uniqueness inequality for general QRTs in our framework.

\begin{proposition}[Uniqueness Inequality]\label{prop24}
  Let $\mathcal{H}$ be a quantum system. Suppose that there is no catalytically replicable state; that is, $r_\textup{conv}\left(\phi\to\phi\right) = 1$ for any resource state $\phi \in \mathcal{S}\left(\mathcal{H}\right)\setminus\mathcal{F}\left(\mathcal{H}\right)$. If a resource measure $R_\mathcal{H}$ satisfies the conventional normalization, asymptotic continuity, and {weak additivity}, 
then for any state $\psi \in \mathcal{S}\left(\mathcal{H}\right)$, $R_\mathcal{H}$ satisfies
  \begin{equation}\label{eq:prop14_claim}
         R_\textup{D}\left(\psi\right)\leqq R_{\mathcal{H}}\left(\psi\right)\leqq R_\textup{C}\left(\psi\right).
  \end{equation}
\end{proposition}

\begin{proof}
First, we prove $R_\textup{D}\left(\psi\right)\leqq R_{\mathcal{H}}\left(\psi\right)$ for any state $\psi \in \mathcal{S}\left(\mathcal{H}\right)$. 
Let $\delta$ be an arbitrary positive number. 
Due to Theorem~\ref{theorem11}, we take a maximally resourceful state $\phi \in \mathcal{G}\left(\mathcal{H}\right)$ such that 
	\begin{equation}\label{uniqueness_distill_0}
		R_\textup{D}\left(\psi\right) \leqq r_\textup{conv}\left(\psi\to\phi\right)R_{\max}^{(\mathcal{H})}.
	\end{equation}
Let $r\coloneqq r_\textup{conv}\left(\psi\to\phi\right)$. 
For any positive integer $n$, it holds that
	\begin{equation}\label{uniqueness_distill_1}
		rR_{\max}^{(\mathcal{H})} \leqq\frac{\left\lceil rn\right\rceil}{n}R_{\max}^{(\mathcal{H})}. 
	\end{equation}
Then, by {conventional normalization }and {weak additivity}, it holds that
	\begin{equation}\label{uniqueness_distill_2}
		\begin{aligned}
			\frac{\left\lceil rn\right\rceil}{n}R_{\max}^{(\mathcal{H})}
			&= \frac{\left\lceil rn\right\rceil}{n} R_{\mathcal{H}}\left(\phi\right)\\
			&=\frac{R_{\mathcal{H}^{\otimes\left\lceil rn\right\rceil}}\left(\phi^{\otimes\left\lceil rn\right\rceil}\right)}{n}. 
		\end{aligned}
	\end{equation}
By the definition of $r_\textup{conv}\left(\psi\to\phi\right)$ shown in~\eqref{eq:conversion_rate}, there exist free operations $\left(\mathcal{N}_n \in \mathcal{O}\left(\mathcal{H}^{\otimes n} \to \mathcal{H}^{\otimes\left\lceil rn\right\rceil}\right)\right)$ such that for any $\epsilon$, it holds that
	\begin{equation}
		\left\|\mathcal{N}_n\left(\psi^{\otimes n}\right) - \phi^{\otimes\left\lceil rn\right\rceil}\right\|_1 < \epsilon, 
	\end{equation}
for an infinitely large subset of $\mathbb{N}$.
Because of asymptotic continuity, there exists sufficiently large $n$ such that
	\begin{equation}\label{uniqueness_distill_3}
		\frac{R_{\mathcal{H}^{\otimes\left\lceil rn\right\rceil}}\left(\phi^{\otimes\left\lceil rn\right\rceil}\right)}{n}\leqq\frac{R_{\mathcal{H}^{\otimes\left\lceil rn\right\rceil}}\left(\mathcal{N}_n\left(\psi^{\otimes n}\right)\right)}{n} + \delta.
	\end{equation}
By monotonicity and {weak additivity}, it holds that
	\begin{equation}\label{uniqueness_distill_4}
		\begin{aligned}
			\frac{R_{\mathcal{H}^{\otimes\left\lceil rn\right\rceil}}\left(\mathcal{N}_n\left(\psi^{\otimes n}\right)\right)}{n}
			&\leqq \frac{R_{\mathcal{H}^{\otimes n}}\left(\psi^{\otimes n}\right)}{n}\\
          &=\frac{nR_{\mathcal{H}}\left(\psi\right)}{n}\\
          &=R_{\mathcal{H}}\left(\psi\right).
		\end{aligned}
	\end{equation}
Therefore, by~\eqref{uniqueness_distill_0},~\eqref{uniqueness_distill_1},~\eqref{uniqueness_distill_2},~\eqref{uniqueness_distill_3}, and~\eqref{uniqueness_distill_4}, it holds that
	\begin{equation}
		R_\textup{D}\left(\psi\right) \leqq R_{\mathcal{H}}\left(\psi\right) + \delta.
	\end{equation}
Since we can take arbitrarily small $\delta$, it holds that $R_\textup{D}\left(\psi\right)\leqq  R_{\mathcal{H}}\left(\psi\right)$ holds. 

Next, we prove $R_\textup{C}\left(\psi\right)\geqq  R_{\mathcal{H}}\left(\psi\right)$ for any state $\psi \in \mathcal{S}\left(\mathcal{H}\right)$. 
Let $\delta$ be an arbitrary positive number. 
Due to Theorem~\ref{theorem11}, we take a maximally resourceful state $\phi \in \mathcal{G}\left(\mathcal{H}\right)$ such that 
	\begin{equation}\label{uniqueness_cost_0}
		\begin{aligned}
		R_\textup{C}\left(\psi\right) + \delta
		&\geqq \frac{R_{\max}^{(\mathcal{H})}}{r_\textup{conv}\left(\phi\to\psi\right)}\\
		& = r'_\textup{conv}\left(\phi\to\psi\right)R_{\max}^{(\mathcal{H})}.
		\end{aligned}
	\end{equation}
	Let $r\coloneqq r'_\textup{conv}\left(\phi\to\psi\right)$. For any positive integer $n$, it holds that
	\begin{equation}\label{uniqueness_cost_1}
		rR_{\max}^{(\mathcal{H})} \geqq\frac{\left\lfloor rn\right\rfloor}{n}R_{\max}^{(\mathcal{H})}.  
	\end{equation}
By the conventional normalization, {weak additivity}, and monotonicity, it holds that
	\begin{equation}\label{uniqueness_cost_2}
		\begin{aligned}
		\frac{\left\lfloor rn\right\rfloor}{n}R_{\max}^{(\mathcal{H})}
		&= \frac{\left\lfloor rn\right\rfloor}{n}R_{\mathcal{H}}\left(\phi\right)\\
		&= \frac{R_{\mathcal{H}^{\otimes\left\lfloor rn\right\rfloor}}\left(\phi^{\otimes\left\lfloor rn\right\rfloor}\right)}{n}\\
		&\geqq \frac{R_{\mathcal{H}^{\otimes n}}\left(\mathcal{M}_n\left(\phi^{\otimes\left\lfloor rn\right\rfloor}\right)\right)}{n},
		\end{aligned}
	\end{equation}
	for any free operation $\mathcal{M}_n$.
By the definition of $r'_\textup{conv}\left(\phi\to\psi\right)$ shown in~\eqref{eq:conversion_rate_cf}, there exist free operations $\left(\mathcal{M}_n \in \mathcal{O}\left(\mathcal{H}^{\otimes\left\lfloor rn\right\rfloor} \to \mathcal{H}^{\otimes n}\right)\right)$ such that for any $\epsilon$, it holds that
	\begin{equation}
		\left\|\mathcal{M}_n\left(\phi^{\otimes\left\lfloor rn\right\rfloor}\right) - \psi^{\otimes n}\right\|_1 < \epsilon, 
	\end{equation}
for an infinitely large subset of $\mathbb{N}$. 
Because of asymptotic continuity, there exists sufficiently large $n$ such that
	\begin{equation}\label{uniqueness_cost_3}
		\frac{R_{\mathcal{H}^{\otimes n}}\left(\mathcal{M}_n\left(\phi^{\otimes\left\lfloor rn\right\rfloor}\right)\right)}{n} \geqq\frac{R_{\mathcal{H}^{\otimes n}}\left(\psi^{\otimes n}\right)}{n} - \delta.
	\end{equation}
Therefore, by~\eqref{uniqueness_cost_0},~\eqref{uniqueness_cost_1},~\eqref{uniqueness_cost_2},~\eqref{uniqueness_cost_3}, and {weak additivity}, it holds that
	\begin{equation}
		\begin{aligned}
			R_\textup{C}\left(\psi\right) + \delta
			&\geqq \frac{R_{\mathcal{H}^{\otimes n}}\left(\psi^{\otimes n}\right)}{n} - \delta\\
			&= \frac{nR_{\mathcal{H}}\left(\psi\right)}{n} - \delta\\
			&= R_{\mathcal{H}}\left(\psi\right) -\delta.
		\end{aligned}
	\end{equation}
Since we can take arbitrarily small $\delta$, it holds that
	\begin{equation}
		R_\textup{C}\left(\psi\right) \geqq R_{\mathcal{H}}\left(\psi\right).
	\end{equation}
\end{proof}

\begin{remark}\label{remark_UI}
  If we can assume strong superadditivity of the resource measure $R$,
  the uniqueness inequality in Proposition~\ref{prop24} holds for $R$ satisfying the conventional normalization, {weak additivity}, {strong superadditivity}, and {lower semi-continuity} for each quantum system $\mathcal{H}$, instead of assuming asymptotic continuity.
  A function $R_\mathcal{H}: \mathcal{S}\left(\mathcal{H}\right) \to \mathbb{R}$ is \textit{lower semi-continuous} if for any fixed state $\psi\in \mathcal{S}\left(\mathcal{H}\right)$ and any sequence of states ${\left(\psi_{n} \in \mathcal{S}\left(\mathcal{H}\right)\right)}_n$ converging to $\psi$, \textit{i.e.}, $\lim_{n \to \infty} \left\|\psi-\psi_{n}\right\|_1 =0$, it holds that
  \begin{equation}\label{lower_semi_continuity}
    \liminf_{n\to\infty} R_\mathcal{H}\left(\psi_n\right)\geqq R_\mathcal{H}\left(\psi\right).
  \end{equation}
  {Strong} superadditivity and {lower semi-continuity} do not necessarily imply asymptotic continuity, and asymptotic continuity does not necessarily imply {strong superadditivity} and {lower semi-continuity} either; that is, these assumptions are independent.
  We prove the uniqueness inequality for $R$ satisfying the conventional normalization, {weak additivity}, {strong superadditivity}, and {lower semi-continuity} in Appendix~\ref{Appendix_B}.
  The combination of {strong superadditivity} and {lower semi-continuity} is first used in Ref.~\cite{Lami2019} to bound a resource cost in the context of a QRT of coherence, but our contribution is to generalize it to the full uniqueness inequality by completing proofs of both bounds, that is, the bounds for the distillable resource as well as the resource cost.
\end{remark}

Despite the general uniqueness inequality shown in Proposition~\ref{prop24}, we show a condition of QRTs where no resource measure satisfies the conventional normalization, asymptotic continuity, and {weak additivity} simultaneously. In these QRTs, Proposition~\ref{prop24} is not applicable.
Note that because of Proposition~\ref{prop15}, the condition~\eqref{eq:thm24_1} in the following theorem is not satisfied for QRTs with a unique maximally resourceful state, but may hold for QRTs with two or more different equivalence classes of maximally resourceful states; that is, 
	\begin{equation}
		\left|\mathcal{G}\left(\mathcal{H}\right)/\sim\right| \geqq 2. 
	\end{equation}
\begin{theorem}[Inconsistency of {Conventional Properties of Resource Measure}]\label{thm24}
Suppose that there exist maximally resourceful states $\phi_0, \phi_1 \in \mathcal{G}\left(\mathcal{H}\right)$ such that 
	\begin{equation}\label{eq:thm24_1}
		r_\textup{conv}\left(\phi_0 \to \phi_1\right) > 1. 
	\end{equation}
If $R_{\max}^{(\mathcal{H})} > 0$, then there exists no resource measure satisfying the conventional normalization, asymptotic continuity, and {weak additivity} simultaneously. 
\end{theorem}
\begin{proof}
The proof is by contradiction. Assume that there exists a resource measure $R$ that satisfies all of the conventional normalization, asymptotic continuity, and {weak additivity}. Let $\phi_0, \phi_1 \in \mathcal{G}\left(\mathcal{H}\right)$ be maximally resourceful states such that  
	\begin{equation}
		r\coloneqq r_\textup{conv}\left(\phi_0 \to \phi_1\right) > 1. 
	\end{equation}
By the definition of $r_\textup{conv}\left(\phi_0\to\phi_1\right)$ shown in~\eqref{eq:conversion_rate}, there exist free operations $\left(\mathcal{N}_n \in \mathcal{O}\left(\mathcal{H}^{\otimes n} \to \mathcal{H}^{\otimes\left\lceil rn\right\rceil}\right)\right)$ such that for any $\epsilon$, it holds that
	\begin{equation}
		\left\|\mathcal{N}_n\left({\phi_0}^{\otimes n}\right) - {\phi_1}^{\otimes\left\lceil rn\right\rceil}\right\|_1 < \epsilon, 
	\end{equation}
for an infinitely large subset of $\mathbb{N}$. 
Then, by asymptotic continuity, for a fixed $\epsilon > 0$, there exists sufficiently large $n$ such that
	\begin{equation}\label{eq:theorem21_1}
		\frac{R_{\mathcal{H}^{\otimes \left\lceil rn \right\rceil}}\left(\mathcal{N}_n\left(\phi_0^{\otimes n}\right)\right)}{n} \geqq \frac{R_{\mathcal{H}^{\otimes \left\lceil rn \right\rceil}}\left(\phi_1^{\otimes \left\lceil rn \right\rceil}\right)}{n} - \epsilon. 
	\end{equation}
By monotonicity, it holds that 
	\begin{equation}
		\frac{R_{\mathcal{H}^{\otimes n}}\left(\phi_0^{\otimes n}\right)}{n} 
		\geqq \frac{R_{\mathcal{H}^{\otimes \left\lceil rn \right\rceil}}\left(\phi_1^{\otimes \left\lceil rn \right\rceil}\right)}{n} - \epsilon.
	\end{equation}	
Then, from {weak additivity}, it follows that
\begin{equation}
	R_\mathcal{H}\left(\phi_0\right) \geqq \frac{\left\lceil rn \right\rceil}{n} R_\mathcal{H}\left(\phi_1\right) - \epsilon. 
\end{equation}
Therefore, due to the conventional normalization, it is necessary that for any $\epsilon$ and $n$ {such that \eqref{eq:theorem21_1} holds}, it holds that
\begin{equation}
	R_{\max}^{(\mathcal{H})} - \frac{\left\lceil rn \right\rceil}{n} R_{\max}^{(\mathcal{H})} \geqq -\epsilon. 
\end{equation}
Since $R_{\max}^{(\mathcal{H})} > 0$, we have
\begin{equation}
	1 - \frac{\left\lceil rn \right\rceil}{n} \geqq - \frac{\epsilon}{R_{\max}^{(\mathcal{H})}}. 
\end{equation}
Since $r>1$, this inequality does not hold for sufficiently small $\epsilon$ or sufficiently large $n$, which implies that there is no such $R$. 
\end{proof}

The following example shows that the QRT of magic has no measure with the conventional normalization, asymptotic continuity, and {weak additivity}.
\begin{example}[QRT without Measure with Conventional Normalization, Asymptotic Continuity, and Weak Additivity]\label{example6}
	The QRT of magic for qutrits~\cite{Veitch2014} shown in Example~\ref{example_QRT_noneq_max} does not have any conventionally normalized, asymptotically continuous and {weakly additive} measure. It has been proved that the asymptotic conversion rate from the Strange state to the Norrell state is larger than 1~\cite{Wang2018}. Therefore, from Theorem~\ref{thm24}, it follows that no measure can satisfy the conventional normalization, asymptotic continuity and {weak additivity} simultaneously in this QRT\@. 
\end{example}

\subsection{Consistency of Resource Measures}\label{Subsection5_c}
In this section, we introduce consistent resource measures in place of resource measures in the previous axiomatic approach in Secs.~\ref{Subsection5_a} and~\ref{Subsection5_b}. As Theorem~\ref{thm24} and Example~\ref{example6} suggest, the conventional normalization, asymptotic continuity and {weak additivity} do not necessarily hold simultaneously with monotonicity in general QRTs. 
On the other hand, a consistent resource measure is compatible with the state conversion rate. We prove that the uniqueness inequality~\eqref{eq:prop14_claim} also holds for a consistent resource measure that is appropriately normalized. 
Note that this normalization respects the state conversion rate, and hence can be different from {conventional normalization }given by~\eqref{conventional_normalization_zero} and~\eqref{conventional_normalization_max}.

First, we introduce a definition of a consistent resource measure. A consistent resource measure quantifies the amount of a resource without contradicting the state conversion rate.
\begin{definition}[Consistent Resource Measure]
	For quantum systems $\mathcal{H}$ and $\mathcal{H}'$, a resource measure $R$ is called a consistent resource measure if for any states $\psi \in \mathcal{S}\left(\mathcal{H}\right)$ and $\phi \in \mathcal{S}\left(\mathcal{H}'\right)$, it holds that 
	\begin{equation}
		R_\mathcal{H}\left(\psi\right)r_\textup{conv}\left(\phi \to \psi \right) \leqq R_{\mathcal{H}'}\left(\phi\right).
	\end{equation} 
\end{definition}

The following proposition shows that a consistent resource measure must be {weakly additive} for any non-catalytically replicable resource states. 
\begin{proposition}\label{prop_weak_add}
	Let $R$ be a consistent resource measure. Then, for any resource state $\psi \in \mathcal{S}\left(\mathcal{H}\right)\setminus\mathcal{F}\left(\mathcal{H}\right)$ that is not catalytically replicable and for any positive integer $n$, it holds that
	\begin{equation}
		R\left(\psi^{\otimes n}\right) = nR\left(\psi\right).
	\end{equation}
\end{proposition}

\begin{proof}
	By~\eqref{conversion_rate_ineq}, it holds that
	\begin{equation}
		\begin{aligned}
		&r_\textup{conv}\left(\psi\to\psi^{\otimes n}\right)r_\textup{conv}\left(\psi^{\otimes n}\to\psi\right) \\
		&\leqq r_\textup{conv}\left(\psi\to\psi\right) \\
		&= 1.
		\end{aligned}
	\end{equation}
Since the identity map is a free operation, we have
	\begin{align}
		r_\textup{conv}\left(\psi\to\psi^{\otimes n}\right) &\geqq {\frac{1}{n}},\\
		r_\textup{conv}\left(\psi^{\otimes n}\to\psi\right) &\geqq {n}.
	\end{align}
Thus, we have
	\begin{align}
		r_\textup{conv}\left(\psi\to\psi^{\otimes n}\right) &= \frac{1}{n},\\
		r_\textup{conv}\left(\psi^{\otimes n}\to\psi\right) &= n.
	\end{align}
By the definition of a consistent resource measure combined with the equations above, it holds that
	\begin{align}
		nR\left(\psi\right)&\leqq R\left(\psi^{\otimes n}\right),\\
		\frac{1}{n}R\left(\psi^{\otimes n}\right) &\leqq R\left(\psi\right).
	\end{align}
Therefore, it holds that
		\begin{equation}
		R\left(\psi^{\otimes n}\right) = nR\left(\psi\right).
	\end{equation}
\end{proof}

We prove that the uniqueness inequality holds for a consistent resource measure that satisfies normalizations in the following propositions.  

\begin{proposition}\label{prop22}
Let $R_\mathcal{H}$ be a consistent resource measure. Suppose that $0\leqq R_\mathcal{H}\left(\psi\right) \leqq R_{\max}^{(\mathcal{H})}$ for any state $\psi \in \mathcal{S}\left(\mathcal{H}\right)$. Then, $R_\mathcal{H}$ satisfies
  \begin{equation}\label{eq:prop22_claim}
         R_\mathcal{H}\left(\psi\right)\leqq R_\textup{C}\left(\psi\right).
  \end{equation}
\end{proposition}

\begin{proof}
	Let $\epsilon$ be an arbitrary positive number. Due to Theorem~\ref{theorem11}, we take a maximally resourceful state $\phi \in \mathcal{G}\left(\mathcal{H}\right)$ such that
	\begin{equation}
		R_\textup{C}\left(\psi\right) + \epsilon \geqq \frac{R_{\max}^{(\mathcal{H})}}{r_\textup{conv}\left(\phi\to\psi\right)}.
	\end{equation}
Then, it holds that
  \begin{align}
	R_\textup{C}\left(\psi\right) + \epsilon
	& \geqq \frac{R_{\max}^{(\mathcal{H})}}{r_\textup{conv}\left(\phi\to\psi\right)}\\
	& \geqq \frac{R_\mathcal{H}\left(\phi\right)}{r_\textup{conv}\left(\phi\to\psi\right)}\\
	& \geqq R_\mathcal{H}\left(\psi\right), 
  \end{align}
where the second inequality follows from the definition of consistent resource measures. 
As we can take an arbitrarily small $\epsilon$, $R_\mathcal{H}\left(\psi\right)\leqq R_\textup{C}\left(\psi\right)$ holds. 
\end{proof}

\begin{proposition}\label{prop23}
	Let $R_\mathcal{H}$ be a consistent resource measure. Suppose that there exists a maximally resourceful state $\phi \in \mathcal{G}\left(\mathcal{H}\right)$ such that $R_\mathcal{H}\left(\phi\right) = R_{\max}^{(\mathcal{H})}$. Then, $R_\mathcal{H}$ satisfies
\begin{equation}\label{eq:prop23_claim}
         R_\mathcal{H}\left(\psi\right)\geqq R_\textup{D}\left(\psi\right)
  \end{equation}
for any state $\psi \in \mathcal{S}\left(\mathcal{H}\right)$.
\end{proposition}

\begin{proof}
	By the definition of the distillable resource, it holds that
	\begin{equation}
		R_\textup{D}\left(\psi\right)  \leqq r_\textup{conv}\left(\psi\to{\phi}\right)R_{\max}^{(\mathcal{H})}.
	\end{equation}
Then, it holds that
  \begin{align}
	R_\textup{D}\left(\psi\right) 
	&\leqq  r_\textup{conv}\left(\psi\to{\phi}\right)R_{\max}^{(\mathcal{H})}\\
	&=  r_\textup{conv}\left(\psi\to{\phi}\right)R_\mathcal{H}\left({\phi}\right)\\
	&\leqq R_\mathcal{H}\left(\psi\right). 
  \end{align}
Therefore, $R_\textup{D}\left(\psi\right)\leqq R_\mathcal{H}\left(\psi\right)$ holds. 
\end{proof}

From Proposition~\ref{prop22} and Proposition~\ref{prop23}, we obtain the following corollary. 
	\begin{corollary}\label{coroll24}
		Let $R_\mathcal{H}$ be a consistent resource measure. Suppose that $R_\mathcal{H}$ satisfies the following assumptions:
	\begin{itemize}
		\item For any state $\psi \in \mathcal{S}\left(\mathcal{H}\right)$, $0\leqq R_\mathcal{H}\left(\psi\right) \leqq R_{\max}^{(\mathcal{H})}$;
		\item There exists $\phi \in \mathcal{G}\left(\mathcal{H}\right)$ such that $R_\mathcal{H}\left(\phi\right) = R_{\max}^{(\mathcal{H})}$.
	\end{itemize}
Then, it holds that
	\begin{equation}
		R_\textup{D}\left(\psi\right) \leqq R_\mathcal{H}\left(\psi\right)\leqq R_\textup{C}\left(\psi\right),
	\end{equation}
for any state $\psi \in \mathcal{S}\left(\mathcal{H}\right)$.
	\end{corollary}

\begin{remark}
The second condition of Corollary~\ref{coroll24} can be replaced by the existence of a state $\rho \in \mathcal{S}\left(\mathcal{H}\right)$ (which is not necessarily maximally resourceful) such that $R_\mathcal{H}\left(\rho\right) = R_{\max}^{(\mathcal{H})}$. 
Suppose a resource measure $R_\mathcal{H}$ satisfies the following conditions: 
	\begin{itemize}
		\item The measure $R_\mathcal{H}$ is normalized in such a way that there exists a positive real number $C$ such that $0\leqq R_\mathcal{H}\left(\psi\right) \leqq C$ for any state $\psi \in \mathcal{S}\left(\mathcal{H}\right)$;
		\item The upper-bound of $R_\mathcal{H}$ is achieved by some state; that is, there exists $\rho \in \mathcal{S}\left(\mathcal{H}\right)$ such that $R_\mathcal{H}\left(\rho\right) = C$. 
	\end{itemize}
Then, there exists a maximally resourceful state $\phi \in \mathcal{G}\left(\mathcal{H}\right)$ such that $R_\mathcal{H}\left(\phi\right) = C$ by monotonicity of the resource measure.
\end{remark}

To observe whether a consistent resource measure satisfies asymptotic continuity, take any two states $\phi, \psi \in \mathcal{S}\left(\mathcal{H}\right)$ such that $r_\textup{conv}\left(\phi\to\psi\right) \leqq 1$. 
Consider a consistent resource measure $R$ satisfying $0\leqq R_\mathcal{H}\left(\psi\right) \leqq R_{\max}^{(\mathcal{H})}$ for any $\psi \in \mathcal{S}(\mathcal{H})$.
By the definition of a consistent resource measure, it holds that
	\begin{equation}\label{eq:continuity_rate}
		\frac{R_\mathcal{H}\left(\psi\right) - R_\mathcal{H}\left(\phi\right)}{R_{\max}^{(\mathcal{H})}} \leqq 1 - r_\textup{conv}\left(\phi\to\psi\right). 
	\end{equation}
Assume that the quantum system $\mathcal{H}$ is finite-dimensional, and take $R_{\max}^{(\mathcal{H})} = \log_2(\dim\mathcal{H})$ as {stated} in~\eqref{typical_normalized_constant}. 
Inequality~\eqref{eq:continuity_rate} suggests that a consistent resource measure has asymptotic continuity if $1 - r_\textup{conv}\left(\phi_{n_i}\to\psi_{n_i}\right)$ converges to zero for any sequence of positive integers ${(n_i)}_{i\in\mathbb{N}}$ and any sequences of states $\left(\phi_{n_i} \in \mathcal{S}(\mathcal{H}^{\otimes n_i})\right)$ and $\left(\psi_{n_i} \in \mathcal{S}(\mathcal{H}^{\otimes n_i})\right)$ satisfying $\lim_{i\to\infty}\|\phi_{n_i} - \psi_{n_i}\|_1 = 0$. However, in general, $1 - r_\textup{conv}\left(\phi_{n_i}\to\psi_{n_i}\right)$ is not necessarily small even if $\|\phi_{n_i} - \psi_{n_i}\|_1$ is small, because the convertibility of two states under free operations is not related to the distance between the two states. Therefore, a consistent resource measure is not necessarily asymptotically continuous. 
More generally, {the difference in the amount of resources present in two states is not necessarily related to the distance between the states since the preorder and the distance are independent.} 
Thus, we do not assume asymptotic continuity in the definition of a consistent resource measure.

\subsection{Example of Consistent Resource Measures}\label{Subsection5_d}
In this section, we show an example of {a} consistent resource measure, which is known as the regularized relative entropy of resource and widely used in known QRTs such as bipartite entanglement~\cite{Vedral1997}, coherence~\cite{Baumgratz2014} and magic states~\cite{Veitch2014}. 
We give the definition of the relative entropy of resource $R_{\textup{R}}$ in our framework. 
	\begin{definition}[Relative Entropy of Resource]
		The relative entropy of resource $R_\textup{R}$ is defined as
	\begin{equation}
		R_\textup{R}\left(\psi\right) \coloneqq \inf_{\phi \in \mathcal{F}\left(\mathcal{H}\right)} D\left(\psi\|\phi\right), 
	\end{equation}
where $D\left(\cdot\|\cdot\right)$ is the quantum relative entropy defined as $D\left(\psi\|\phi\right) =\tr \psi \log_2 \psi - \tr \psi \log_2 \phi$.
	\end{definition}
The relative entropy of resource $R_\textup{R}$ is {subadditive} 
since the set of free operations is closed under tensor product. 
Therefore, by {subadditivity} of $R_\textup{R}$, the regularized relative entropy of resource defined as 
	\begin{equation}
		R^\infty_\textup{R}\left(\rho\right) \coloneqq \lim_{n\to\infty}\frac{R_\textup{R}\left(\rho^{\otimes n}\right)}{n}
	\end{equation}
 exists~\cite{Donald2002}.

We show that the regularized relative entropy of resource serves as a consistent resource measure for a finite-dimensional convex QRT in which $\mathcal{F}\left(\mathcal{H}\right)$ for each $\mathcal{H}$ contains at least one full-rank state. 
Consider a convex QRT that is defined for finite-dimensional systems and satisfies the axioms of QRTs given in Sec.~\ref{Section2}. 
It has been shown that if the set of free states $\mathcal{F}\left(\mathcal{H}\right)$ for each $\mathcal{H}$ contains at least one full-rank state, the relative entropy of resource is asymptotically continuous~\cite{Synak2006,Brandao2010}.
On the other hand, in Refs.~\cite{Horodecki2002, Chitambar2018}, 
it is shown that 
\begin{equation}
	r_{\textup{conv}}(\phi \to \psi) \leqq \frac{f^\infty(\phi)}{f^\infty(\psi)}
\end{equation}
holds for an asymptotically continuous resource measure $f$ and its regularization $f^\infty$.
Therefore, the regularized relative entropy of resource $R^\infty_\textup{R}$ is a consistent resource measure for a convex QRT that has a full-rank free state in each dimension.
The QRT of bipartite entanglement~\cite{Chitambar2014}, coherence~\cite{Winter2016} and magic~\cite{Veitch2014} are known as convex QRTs with full-rank free states. In these QRTs, the regularized relative entropy of resource works as a consistent resource measure. 

We remark that this proof of the existence of a consistent measure is not applicable to non-convex QRTs because the relative entropy of resource for a non-convex set of free states can be discontinuous~\cite{Weis2012}. 
However, as mentioned in the last paragraph of Sec.~\ref{Subsection5_c},
asymptotic continuity is a sufficient but not necessary condition for a resource measure to be consistent.
Therefore, there may be a consistent resource measure that is not asymptotically continuous.
Thus, there may be a consistent resource measure even in QRTs that are not convex, not finite-dimensional, 
or {do not contain} full-rank free states while further research is needed to explicitly construct consistent resource measures in these QRTs.

\section{Conclusion}\label{Section6}
We have formulated and investigated quantum state conversion and resource measures in a framework of general QRTs to figure out general properties of quantum resources. 
{Our framework is based on minimal assumptions,} and hence covers a broad range of QRTs including those with non-unique maximally resourceful states, non-convexity, and infinite dimension. In our general framework, the existence of maximally resourceful states is no longer trivial, but we proved that there always exists a maximally resourceful state in the general QRTs.

To clarify general properties of resource manipulation, we investigated one-shot and asymptotic state conversions, which are central tasks in QRTs. We discovered {the existence of catalytically replicable states, which are resources that are infinitely replicable by free operations}. In addition, we introduced the distillable resource and the resource cost in our framework without assuming uniqueness of maximally resourceful states. We showed that the distillable resource and the resource cost are weakly subadditive. Furthermore, we showed that the distillable resource is always smaller or equal to the resource cost if there is no catalytically replicable state. 

As for quantification of quantum resources, we proved that the conventional normalization, asymptotic continuity, and {weak additivity} are incompatible with each other in general QRTs with non-unique maximally resourceful states. Motivated by this incompatibility, we introduced a consistent resource measure, which is consistent with the asymptotic state conversion rate. Moreover, we proved {that} a normalized consistent resource measure is bounded by the distillable resource and the resource cost, generalizing the previous work on the uniqueness inequality in the entanglement theory to general QRTs.

Owing to the generality, our formulations and results broaden potential applications of QRTs in the following future research directions. Since we formulated a framework of QRTs applicable to non-convex QRTs where randomness can be regarded as a resource, it would be interesting to find further applications of non-convex QRTs, such as analyses of random-number generation~\cite{Miguel2017} and quantum $t$-design~\cite{Cleve2009}.
In addition, since our framework forms a basis of QRTs on infinite-dimensional quantum systems, our results provide a foundation for applying QRTs to quantum field theory.
Since we discovered a counter-intuitive phenomenon of catalytically replicable resources, it is interesting to find more situations where catalytically replicable states arise.
Furthermore, while we showed that the regularized relative entropy serves as a consistent resource measure in convex finite-dimensional QRTs that have full-rank free states, construction of a consistent resource measure for all the QRTs in our framework including non-convex or infinite-dimensional QRTs is still open.
Finally, extension of our framework to dynamic resource~\cite{Theurer2019,Gour2019a,Liu_YC2020,Gour_Wilde2018,Li2018,Liu_ZW2019b,Takagi2020} would also be an interesting future direction.

We established general and fruitful structures of QRTs disclosing universal properties of quantum resources.
Owing to the broad applicability of our formulations, our results open a way to quantitative understandings of complicated quantum-mechanical phenomena that are sometimes hard to analyze, through a unified approach using our general formulation of QRTs.

\begin{acknowledgments}
We acknowledge Debbie Leung, Yui Kuramochi, Toshihiko Sasaki, and Masato Koashi for the helpful advice and discussion. 
K.\ K.\ was supported by Mike and Ophelia Lazaridis, and research grants by NSERC\@.
H.\ Y.\ was supported by CREST (Japan  Science and Technology Agency) JPMJCR1671, Cross-ministerial Strategic Innovation Promotion   Program (SIP) (Council for Science, Technology and Innovation (CSTI)), and JSPS Overseas Research Fellowships.
\end{acknowledgments}

\appendix
\setcounter{section}{0}

\section{\label{Appendix_A}Equivalence of Compactness in Weak Operator Topology and Trace Norm Topology}
In this section, we prove the following lemma, which shows that compactness in the weak operator topology is equivalent to that in the trace norm topology on a set of density operators. We exploit this lemma in the proof of Theorem~\ref{prop5} in Sec.~\ref{Subsection2_c}.
\begin{lemma}\label{l1}
For any set of density operators $K \subset \mathcal{D}\left(\mathcal{H}\right)$, $K$ is {compact in terms of the weak operator topology} if and only if
$K$ is compact in the trace norm topology.
\end{lemma}
\begin{proof}
Since the trace norm topology is stronger than the weak operator topology,
the \lq\lq{}if\rq\rq{} part is obvious.
Assume that $K$ is {compact in terms of the weak operator topology} to show the \lq\lq{}only if\rq\rq{} part.
To show the compactness in the trace norm topology, take an arbitrary sequence
${(\psi_n)}_{n \in \mathbb{N}}$ in $K$.
According to the Eberlein-\v{S}mulian theorem (\textit{e.g.}~\cite{dunfordschwartzvol1}, Theorem~V.6.1),
in the weak operator topology,
the condition of the compactness coincides with that of the sequential compactness.
Therefore,
there exists a subsequence ${(\psi_{n(k)})}_{k \in \mathbb{N}}$
converging to some $\psi \in K  \subset \mathcal{D}\left(\mathcal{H}\right)$ {in terms of the weak opertor topology}.
Moreover, according to~\cite{Dell'Antonio1967} (Lemma~2),
a sequence in $\mathcal{D}\left(\mathcal{H}\right)$ convergent to a density operator in terms of the weak operator topology
is in fact convergent to that density operator in terms of the trace norm topology.
Thus, ${(\psi_{n(k)})}_{k \in \mathbb{N}}$ is convergent to $\psi$ in terms of the trace norm topology.
Therefore $K$ is sequentially compact, hence compact, in the trace norm topology.
\end{proof}

\section{\label{Appendix_B} Uniqueness Inequality Based on Strong Superadditivity and Lower Semi-Continuity}
In this section, we prove the uniqueness inequality under the assumptions of {strong superadditivity} and {lower semi-continuity}, instead of asymptotic continuity as mentioned in Remark~\ref{remark_UI}. 
In Ref.~\cite{Lami2019}, {strong superadditivity} and {lower semi-continuity} are used to bound a resource cost in the context of a QRT of coherence, and the extension for general resources was mentioned. 
Our proof further generalizes this statement to the full uniqueness inequality for general QRTs by completing proofs of both bounds, that is, the bounds for the distillable resource as well as the resource cost.

\begin{proposition}
Let $\mathcal{H}$ be a quantum system. Suppose that there is no catalytically replicable state; that is, $r_\textup{conv}\left(\phi\to\phi\right) = 1$ for any resource state $\phi \in \mathcal{S}\left(\mathcal{H}\right)\setminus\mathcal{F}\left(\mathcal{H}\right)$. If a resource measure $R_\mathcal{H}$ satisfies the conventional normalization, {weak additivity}, {strong superadditivity}, and {lower semi-continuity}, 
then for any state $\psi \in \mathcal{S}\left(\mathcal{H}\right)$, $R_\mathcal{H}$ satisfies
  \begin{equation}\label{eq:prop14_claim_appendix}
         R_\textup{D}\left(\psi\right)\leqq R_{\mathcal{H}}\left(\psi\right)\leqq R_\textup{C}\left(\psi\right).
  \end{equation}	
\end{proposition}

\begin{proof}
First, we prove $R_\textup{D}\left(\psi\right)\leqq R_{\mathcal{H}}\left(\psi\right)$ for any state $\psi \in \mathcal{S}\left(\mathcal{H}\right)$. 
Let $\delta$ be an arbitrary positive number. 
Due to Theorem~\ref{theorem11}, we take a maximally resourceful state $\phi \in \mathcal{G}\left(\mathcal{H}\right)$ such that 
	\begin{equation}\label{eq:uniqueness_distill_0}
		R_\textup{D}\left(\psi\right) \leqq r_\textup{conv}\left(\psi\to\phi\right)R_{\max}^{(\mathcal{H})}\leqq R_\textup{D}\left(\psi\right) + \delta.
	\end{equation}
Let $r\coloneqq r_\textup{conv}\left(\psi\to\phi\right)$.
For any positive integer $n$, it holds that
	\begin{equation}\label{eq:uniqueness_distill_1}
		rR_{\max}^{(\mathcal{H})} \leqq\frac{\left\lceil rn\right\rceil}{n}R_{\max}^{(\mathcal{H})}.
	\end{equation}
Then, by {conventional normalization} and {weak additivity}, it holds that
	\begin{equation}\label{eq:uniqueness_distill_2}
		\begin{aligned}
			\frac{\left\lceil rn\right\rceil}{n}R_{\max}^{(\mathcal{H})}
			&= \frac{\left\lceil rn\right\rceil}{n} R_{\mathcal{H}}\left(\phi\right).
		\end{aligned}
	\end{equation}
By the definition of $r_\textup{conv}\left(\psi\to\phi\right)$ shown in~\eqref{eq:conversion_rate}, there exist free operations $\left(\mathcal{N}_n \in \mathcal{O}\left(\mathcal{H}^{\otimes n} \to \mathcal{H}^{\otimes\left\lceil rn\right\rceil}\right)\right)$ such that for any $\epsilon$, it holds that
	\begin{equation}
	  \label{eq:uniqueness_distill_trace_norm}
		\left\|\mathcal{N}_n\left(\psi^{\otimes n}\right) - \phi^{\otimes\left\lceil rn\right\rceil}\right\|_1 < \epsilon, 
	\end{equation}
for an infinitely large subset of $\mathbb{N}$.
{Note that $\mathcal{N}_n\left(\psi^{\otimes n}\right)\in\mathcal{S}\left(\mathcal{H}^{\otimes\lceil rn\rceil}\right)$ is a state on a system composed of $\lceil rn\rceil$ subsystems. 
For $k\in\left\{1,\ldots,\lceil rn\rceil\right\}$, let $\tilde{\phi}_{\epsilon}^{(k)} \in \mathcal{S}\left(\mathcal{H}\right)$ denote a state obtained by tracing out $(\lceil rn\rceil-1)$ subsystems except the $k$th one for $\mathcal{N}_n\left(\psi^{\otimes n}\right)$.}
{Strong} superadditivity yields
\begin{equation}
  R_{\mathcal{H}^{\otimes\lceil rn\rceil}}\left(\mathcal{N}_n\left(\psi^{\otimes n}\right)\right)\geqq\sum_{k=1}^{\lceil rn\rceil}R_{\mathcal{H}}\left(\tilde{\phi}_{\epsilon}^{(k)}\right).
\end{equation}
Then defining
\begin{equation}
  \tilde{\phi}_{\epsilon}\coloneqq\argmin_{\tilde{\phi}\in\left\{\tilde{\phi}_{\epsilon}^{(k)}:k\in\left\{1,\ldots,\lceil rn\rceil\right\}\right\}}R_{\mathcal{H}}\left(\tilde{\phi}\right),
\end{equation}
we have
\begin{equation}\label{eq:uniqueness_distill_strong_additivity}
  R_{\mathcal{H}^{\otimes\lceil rn\rceil}}\left(\mathcal{N}_n\left(\psi^{\otimes n}\right)\right)\geqq\lceil rn\rceil R\left(\tilde{\phi}_{\epsilon}\right).
\end{equation}
From~\eqref{eq:uniqueness_distill_trace_norm}, monotonicity of the trace distance implies
\begin{equation}
  \left\|\tilde{\phi}_{\epsilon} - \phi\right\|_1 < \epsilon.
\end{equation}
Because of {lower semi-continuity}, we can take a sufficiently small $\epsilon>0$ such that
\begin{equation}\label{eq:uniqueness_distill_5}
  R_{\mathcal{H}}\left(\phi\right)\leqq R_{\mathcal{H}}\left(\phi_\epsilon\right)+\delta.
\end{equation}
Using~\eqref{eq:uniqueness_distill_strong_additivity}, we obtain
\begin{equation}\label{eq:uniqueness_distill_3}
  \frac{\left\lceil rn\right\rceil}{n} R_{\mathcal{H}}\left(\phi_\epsilon\right)\leqq \frac{R_{\mathcal{H}^{\otimes\lceil rn\rceil}}\left(\mathcal{N}_n\left(\psi^{\otimes n}\right)\right)}{n}.
\end{equation}
By monotonicity and {weak additivity}, it holds that
\begin{equation}\label{eq:uniqueness_distill_4}
  \begin{aligned}
    \frac{R_{\mathcal{H}^{\otimes\left\lceil rn\right\rceil}}\left(\mathcal{N}_n\left(\psi^{\otimes n}\right)\right)}{n}
		&\leqq \frac{R_{\mathcal{H}^{\otimes n}}\left(\psi^{\otimes n}\right)}{n}\\
		&=\frac{nR_{\mathcal{H}}\left(\psi\right)}{n}\\
		&= R_{\mathcal{H}}\left(\psi\right).
  \end{aligned}
\end{equation}
Therefore, by~\eqref{eq:uniqueness_distill_0},~\eqref{eq:uniqueness_distill_1},~\eqref{eq:uniqueness_distill_2},~\eqref{eq:uniqueness_distill_5},~\eqref{eq:uniqueness_distill_3}, and~\eqref{eq:uniqueness_distill_4}, it holds that
\begin{align}
  R_\textup{D}\left(\psi\right) &\leqq R_{\mathcal{H}}\left(\psi\right) + \frac{\lceil rn\rceil}{n}\delta\\
				&\leqq R_{\mathcal{H}}\left(\psi\right) + \left(\frac{R_\textup{D}\left(\psi\right)+\delta}{R_{\max}^{(\mathcal{H})}}+\frac{1}{n}\right)\delta\\
				&\leqq R_{\mathcal{H}}\left(\psi\right) + \left(1+\frac{\delta}{R_{\max}^{(\mathcal{H})}}+1\right)\delta,
\end{align}
where the last inequality follows from Proposition~\ref{proposition4_a_1} and $\frac{1}{n}\leqq 1$.
Since we can take arbitrarily small $\delta$, it holds that $R_\textup{D}\left(\psi\right)\leqq  R_{\mathcal{H}}\left(\psi\right)$.

Next, we prove $R_\textup{C}\left(\psi\right)\geqq  R_{\mathcal{H}}\left(\psi\right)$ for any state $\psi \in \mathcal{S}\left(\mathcal{H}\right)$.
Let $\delta$ be an arbitrary positive number. 
Due to Theorem~\ref{theorem11}, we take a maximally resourceful state $\phi \in \mathcal{G}\left(\mathcal{H}\right)$ such that 
	\begin{equation}\label{eq:uniqueness_cost_0}
		\begin{aligned}
		R_\textup{C}\left(\psi\right) + \delta
		&\geqq \frac{R_{\max}^{(\mathcal{H})}}{r_\textup{conv}\left(\phi\to\psi\right)}\\
		& = r'_\textup{conv}\left(\phi\to\psi\right)R_{\max}^{(\mathcal{H})}\\
		&\geqq R_\textup{C}\left(\psi\right).
		\end{aligned}
	\end{equation}
	Let $r\coloneqq r'_\textup{conv}\left(\phi\to\psi\right)$. For any positive integer $n$, it holds that
	\begin{equation}\label{eq:uniqueness_cost_1}
		rR_{\max}^{(\mathcal{H})} \geqq\frac{\left\lfloor rn\right\rfloor}{n}R_{\max}^{(\mathcal{H})}.  
	\end{equation}
By the conventional normalization, {weak additivity}, and monotonicity, it holds that
	\begin{equation}\label{eq:uniqueness_cost_2}
		\begin{aligned}
		\frac{\left\lfloor rn\right\rfloor}{n}R_{\max}^{(\mathcal{H})}
		&= \frac{\left\lfloor rn\right\rfloor}{n}R_{\mathcal{H}}\left(\phi\right)\\
		&= \frac{R_{\mathcal{H}^{\otimes\left\lfloor rn\right\rfloor}}\left(\phi^{\otimes\left\lfloor rn\right\rfloor}\right)}{n}\\
		&\geqq \frac{R_{\mathcal{H}^{\otimes n}}\left(\mathcal{M}_n\left(\phi^{\otimes\left\lfloor rn\right\rfloor}\right)\right)}{n},
		\end{aligned}
	\end{equation}
for any free operation $\mathcal{M}_n$.
By the definition of $r'_\textup{conv}\left(\phi\to\psi\right)$ shown in~\eqref{eq:conversion_rate_cf}, there exist free operations $\left(\mathcal{M}_n \in \mathcal{O}\left(\mathcal{H}^{\otimes\left\lfloor rn\right\rfloor} \to \mathcal{H}^{\otimes n}\right)\right)$ such that for any $\epsilon$, it holds that
	\begin{equation}\label{eq:uniqueness_cost_trace_norm}
		\left\|\mathcal{M}_n\left(\phi^{\otimes\left\lfloor rn\right\rfloor}\right) - \psi^{\otimes n}\right\|_1 < \epsilon, 
	\end{equation}
for an infinitely large subset of $\mathbb{N}$. 
{Note that $\mathcal{M}_n\left(\phi^{\otimes\left\lfloor rn\right\rfloor}\right)\in\mathcal{S}\left(\mathcal{H}^{\otimes n}\right)$ is a state on a system composed of $n$ subsystems. 
For each $k\in\left\{1,\ldots,n\right\}$, let $\tilde{\psi}_{\epsilon}^{(k)} \in \mathcal{S}\left(\mathcal{H}\right)$  denote a state obtained by tracing out $(n-1)$ subsystems except the $k$th one for $\mathcal{M}_n\left(\phi^{\otimes\left\lfloor rn\right\rfloor}\right)$.}
{Strong} superadditivity yields
\begin{equation}
  R_{\mathcal{H}^{\otimes n}}\left(\mathcal{M}_n\left(\phi^{\otimes\left\lfloor rn\right\rfloor}\right)\right)\geqq\sum_{k=1}^{n}R_{\mathcal{H}}\left(\tilde{\psi}_{\epsilon}^{(k)}\right).
\end{equation}
Then defining
\begin{equation}
  \tilde{\psi}_{\epsilon}\coloneqq\argmin_{\tilde{\psi}\in\left\{\tilde{\psi}_{\epsilon}^{(k)}:k\in\left\{1,\ldots,\lceil rn\rceil\right\}\right\}}R_{\mathcal{H}}\left(\tilde{\psi}\right),
\end{equation}
we have
\begin{equation}\label{eq:uniqueness_cost_strong_additivity}
  R_{\mathcal{H}^{\otimes n}}\left(\mathcal{M}_n\left(\phi^{\otimes\left\lfloor rn\right\rfloor}\right)\right)\geqq nR\left(\tilde{\psi}_{\epsilon}\right).
\end{equation}
Thus, we obtain
\begin{equation}\label{eq:uniqueness_cost_3}
  \frac{R_{\mathcal{H}^{\otimes n}}\left(\mathcal{M}_n\left(\phi^{\otimes\left\lfloor rn\right\rfloor}\right)\right)}{n}\geqq R\left(\tilde{\psi}_{\epsilon}\right).
\end{equation}
From~\eqref{eq:uniqueness_cost_trace_norm}, monotonicity of the trace distance implies
\begin{equation}
  \left\|\tilde{\psi}_{\epsilon} - \psi\right\|_1 < \epsilon.
\end{equation}
Because of {lower semi-continuity}, we can take a sufficiently small $\epsilon>0$ such that
\begin{equation}\label{eq:uniqueness_cost_5}
  R_{\mathcal{H}}\left(\psi_\epsilon\right)\geqq R_{\mathcal{H}}\left(\psi\right)-\delta.
\end{equation}
Therefore, by~\eqref{eq:uniqueness_cost_0},~\eqref{eq:uniqueness_cost_1},~\eqref{eq:uniqueness_cost_2},~\eqref{eq:uniqueness_cost_3}, and~\eqref{eq:uniqueness_cost_5}
it holds that
	\begin{equation}
		\begin{aligned}
			R_\textup{C}\left(\psi\right) + \delta
			&\geqq R_{\mathcal{H}}\left(\psi\right) -\delta.
		\end{aligned}
	\end{equation}
Since we can take arbitrarily small $\delta$, it holds that
	\begin{equation}
		R_\textup{C}\left(\psi\right) \geqq R_{\mathcal{H}}\left(\psi\right).
	\end{equation}
\end{proof}

\newpage
\bibliographystyle{unsrtnat}
\bibliography{GQRT}

\begin{thebibliography}{92}
\providecommand{\natexlab}[1]{#1}
\providecommand{\url}[1]{\texttt{#1}}
\expandafter\ifx\csname urlstyle\endcsname\relax
  \providecommand{\doi}[1]{doi: #1}\else
  \providecommand{\doi}{doi: \begingroup \urlstyle{rm}\Url}\fi

\bibitem[Chitambar and Gour(2019)]{Chitambar2018}
Eric Chitambar and Gilad Gour.
\newblock Quantum resource theories.
\newblock \emph{Rev. Mod. Phys.}, 91:\penalty0 025001, Apr 2019.
\newblock \doi{10.1103/RevModPhys.91.025001}.
\newblock URL \url{https://link.aps.org/doi/10.1103/RevModPhys.91.025001}.

\bibitem[Horodecki et~al.(2009)Horodecki, Horodecki, Horodecki, and
  Horodecki]{Horodecki2009}
Ryszard Horodecki, Pawe\l{} Horodecki, Micha\l{} Horodecki, and Karol
  Horodecki.
\newblock Quantum entanglement.
\newblock \emph{Rev. Mod. Phys.}, 81:\penalty0 865--942, Jun 2009.
\newblock \doi{10.1103/RevModPhys.81.865}.
\newblock URL \url{https://link.aps.org/doi/10.1103/RevModPhys.81.865}.

\bibitem[{Rains}(2001)]{Rain2001}
E.~M. {Rains}.
\newblock A semidefinite program for distillable entanglement.
\newblock \emph{IEEE Trans. Inf. Theory}, 47\penalty0 (7):\penalty0 2921--2933,
  Nov 2001.
\newblock ISSN 1557-9654.
\newblock \doi{10.1109/18.959270}.
\newblock URL \url{https://doi.org/10.1109/18.959270}.

\bibitem[Brand{\~a}o and {Datta}(2011)]{Brandao2011}
Fernando G. S.~L. Brand{\~a}o and N.~{Datta}.
\newblock One-shot rates for entanglement manipulation under non-entangling
  maps.
\newblock \emph{IEEE Trans. Inf. Theory}, 57\penalty0 (3):\penalty0 1754--1760,
  March 2011.
\newblock ISSN 1557-9654.
\newblock \doi{10.1109/TIT.2011.2104531}.
\newblock URL \url{https://doi.org/10.1109/TIT.2011.2104531}.

\bibitem[Streltsov et~al.(2017)Streltsov, Adesso, and Plenio]{Streltsov2017}
Alexander Streltsov, Gerardo Adesso, and Martin~B. Plenio.
\newblock Colloquium: Quantum coherence as a resource.
\newblock \emph{Rev. Mod. Phys.}, 89:\penalty0 041003, Oct 2017.
\newblock \doi{10.1103/RevModPhys.89.041003}.
\newblock URL \url{https://link.aps.org/doi/10.1103/RevModPhys.89.041003}.

\bibitem[Marvian and Spekkens(2016)]{Marvian2016}
Iman Marvian and Robert~W. Spekkens.
\newblock How to quantify coherence: Distinguishing speakable and unspeakable
  notions.
\newblock \emph{Phys. Rev. A}, 94:\penalty0 052324, Nov 2016.
\newblock \doi{10.1103/PhysRevA.94.052324}.
\newblock URL \url{https://link.aps.org/doi/10.1103/PhysRevA.94.052324}.

\bibitem[Aberg(2006)]{Aberg2006}
Johan Aberg.
\newblock Quantifying superposition.
\newblock \emph{arXiv preprint
  \href{https://arxiv.org/abs/quant-ph/0612146}{arXiv:quant-ph/0612146}}, 2006.
\newblock URL \url{https://arxiv.org/abs/quant-ph/0612146}.

\bibitem[Baumgratz et~al.(2014)Baumgratz, Cramer, and Plenio]{Baumgratz2014}
T.~Baumgratz, M.~Cramer, and M.~B. Plenio.
\newblock Quantifying coherence.
\newblock \emph{Phys. Rev. Lett.}, 113:\penalty0 140401, Sep 2014.
\newblock \doi{10.1103/PhysRevLett.113.140401}.
\newblock URL \url{https://link.aps.org/doi/10.1103/PhysRevLett.113.140401}.

\bibitem[Yadin et~al.(2016)Yadin, Ma, Girolami, Gu, and Vedral]{Yadin2016}
Benjamin Yadin, Jiajun Ma, Davide Girolami, Mile Gu, and Vlatko Vedral.
\newblock Quantum processes which do not use coherence.
\newblock \emph{Phys. Rev. X}, 6:\penalty0 041028, Nov 2016.
\newblock \doi{10.1103/PhysRevX.6.041028}.
\newblock URL \url{https://link.aps.org/doi/10.1103/PhysRevX.6.041028}.

\bibitem[Chitambar and Gour(2016{\natexlab{a}})]{Chitambar2016a}
Eric Chitambar and Gilad Gour.
\newblock Critical examination of incoherent operations and a physically
  consistent resource theory of quantum coherence.
\newblock \emph{Phys. Rev. Lett.}, 117:\penalty0 030401, Jul
  2016{\natexlab{a}}.
\newblock \doi{10.1103/PhysRevLett.117.030401}.
\newblock URL \url{https://link.aps.org/doi/10.1103/PhysRevLett.117.030401}.

\bibitem[Chitambar and Gour(2016{\natexlab{b}})]{Chitambar2016}
Eric Chitambar and Gilad Gour.
\newblock Comparison of incoherent operations and measures of coherence.
\newblock \emph{Phys. Rev. A}, 94:\penalty0 052336, Nov 2016{\natexlab{b}}.
\newblock \doi{10.1103/PhysRevA.94.052336}.
\newblock URL \url{https://link.aps.org/doi/10.1103/PhysRevA.94.052336}.

\bibitem[Janzing et~al.(2000)Janzing, Wocjan, Zeier, Geiss, and
  Beth]{Janzing2000}
D.~Janzing, P.~Wocjan, R.~Zeier, R.~Geiss, and Th. Beth.
\newblock Thermodynamic cost of reliability and low temperatures: Tightening
  {Landauer's} principle and the second law.
\newblock \emph{Int. J. Theor. Phys.}, 39:\penalty0 2717, 2000.
\newblock \doi{10.1023/A:1026422630734}.
\newblock URL \url{https://doi.org/10.1023/A:1026422630734}.

\bibitem[Brand\~ao et~al.(2013)Brand\~ao, Horodecki, Oppenheim, Renes, and
  Spekkens]{Horodecki2013}
Fernando G. S.~L. Brand\~ao, Micha\l{} Horodecki, Jonathan Oppenheim, Joseph~M.
  Renes, and Robert~W. Spekkens.
\newblock Resource theory of quantum states out of thermal equilibrium.
\newblock \emph{Phys. Rev. Lett.}, 111:\penalty0 250404, Dec 2013.
\newblock \doi{10.1103/PhysRevLett.111.250404}.
\newblock URL \url{https://link.aps.org/doi/10.1103/PhysRevLett.111.250404}.

\bibitem[Horodecki and Oppenheim(2013{\natexlab{a}})]{Horodecki2013a}
Michal Horodecki and Jonathan Oppenheim.
\newblock Fundamental limitations for quantum and nanoscale thermodynamics.
\newblock \emph{Nat. Commun.}, 4:\penalty0 2059, 2013{\natexlab{a}}.
\newblock \doi{10.1038/ncomms3059}.
\newblock URL \url{https://doi.org/10.1038/ncomms3059}.

\bibitem[Halpern(2017)]{Halpern2017}
Nicole~Yunger Halpern.
\newblock Toward physical realizations of thermodynamic resource theories.
\newblock In \emph{Information and Interaction}, pages 135--166. Springer,
  2017.
\newblock \doi{10.1007/978-3-319-43760-6_8}.
\newblock URL \url{https://doi.org/10.1007/978-3-319-43760-6_8}.

\bibitem[Yunger~Halpern and Limmer(2020)]{Halpern2020}
Nicole Yunger~Halpern and David~T. Limmer.
\newblock Fundamental limitations on photoisomerization from thermodynamic
  resource theories.
\newblock \emph{Phys. Rev. A}, 101:\penalty0 042116, Apr 2020.
\newblock \doi{10.1103/PhysRevA.101.042116}.
\newblock URL \url{https://link.aps.org/doi/10.1103/PhysRevA.101.042116}.

\bibitem[Veitch et~al.(2014)Veitch, Mousavian, Gottesman, and
  Emerson]{Veitch2014}
Victor Veitch, S~A~Hamed Mousavian, Daniel Gottesman, and Joseph Emerson.
\newblock The resource theory of stabilizer quantum computation.
\newblock \emph{New J. Phys.}, 16\penalty0 (1):\penalty0 013009, jan 2014.
\newblock \doi{10.1088/1367-2630/16/1/013009}.
\newblock URL \url{https://doi.org/10.1088%2F1367-2630%2F16%2F1%2F013009}.

\bibitem[Howard and Campbell(2017)]{Howard2017}
Mark Howard and Earl Campbell.
\newblock Application of a resource theory for magic states to fault-tolerant
  quantum computing.
\newblock \emph{Phys. Rev. Lett.}, 118:\penalty0 090501, Mar 2017.
\newblock \doi{10.1103/PhysRevLett.118.090501}.
\newblock URL \url{https://link.aps.org/doi/10.1103/PhysRevLett.118.090501}.

\bibitem[Gour and Spekkens(2008)]{Gour2008}
Gilad Gour and Robert~W Spekkens.
\newblock The resource theory of quantum reference frames: manipulations and
  monotones.
\newblock \emph{New J. Phys.}, 10\penalty0 (3):\penalty0 033023, 2008.
\newblock \doi{10.1088/1367-2630/10/3/033023}.
\newblock URL \url{https://doi.org/10.1088/1367-2630/10/3/033023}.

\bibitem[Gour et~al.(2009)Gour, Marvian, and Spekkens]{Gour2009}
Gilad Gour, Iman Marvian, and Robert~W. Spekkens.
\newblock Measuring the quality of a quantum reference frame: The relative
  entropy of frameness.
\newblock \emph{Phys. Rev. A}, 80:\penalty0 012307, Jul 2009.
\newblock \doi{10.1103/PhysRevA.80.012307}.
\newblock URL \url{https://link.aps.org/doi/10.1103/PhysRevA.80.012307}.

\bibitem[Marvian and Spekkens(2013)]{Marvian2013}
Iman Marvian and Robert~W Spekkens.
\newblock The theory of manipulations of pure state asymmetry: I. basic tools,
  equivalence classes and single copy transformations.
\newblock \emph{New J. Phys.}, 15\penalty0 (3):\penalty0 033001, mar 2013.
\newblock \doi{10.1088/1367-2630/15/3/033001}.
\newblock URL \url{https://doi.org/10.1088%2F1367-2630%2F15%2F3%2F033001}.

\bibitem[Marvian and Spekkens(2014{\natexlab{a}})]{Marvian2014}
Iman Marvian and Robert~W. Spekkens.
\newblock Asymmetry properties of pure quantum states.
\newblock \emph{Phys. Rev. A}, 90:\penalty0 014102, Jul 2014{\natexlab{a}}.
\newblock \doi{10.1103/PhysRevA.90.014102}.
\newblock URL \url{https://link.aps.org/doi/10.1103/PhysRevA.90.014102}.

\bibitem[Marvian and Spekkens(2014{\natexlab{b}})]{Marvian2014a}
Iman Marvian and Robert~W Spekkens.
\newblock Extending {Noether}'s theorem by quantifying the asymmetry of quantum
  states.
\newblock \emph{Nat. Commun.}, 5\penalty0 (1):\penalty0 1--8,
  2014{\natexlab{b}}.
\newblock \doi{10.1038/ncomms4821}.
\newblock URL \url{https://doi.org/10.1038/ncomms4821}.

\bibitem[Gour et~al.(2015)Gour, M\"{u}ller, Narasimhachar, Spekkens, and
  Halpern]{Gour2015}
Gilad Gour, Markus~P. M\"{u}ller, Varun Narasimhachar, Robert~W. Spekkens, and
  Nicole~Yunger Halpern.
\newblock The resource theory of informational nonequilibrium in
  thermodynamics.
\newblock \emph{Phys. Rep.}, 583:\penalty0 1 -- 58, 2015.
\newblock ISSN 0370-1573.
\newblock \doi{https://doi.org/10.1016/j.physrep.2015.04.003}.
\newblock URL
  \url{http://www.sciencedirect.com/science/article/pii/S037015731500229X}.

\bibitem[Takagi and Zhuang(2018)]{Takagi2018}
Ryuji Takagi and Quntao Zhuang.
\newblock Convex resource theory of non-{Gaussianity}.
\newblock \emph{Phys. Rev. A}, 97:\penalty0 062337, Jun 2018.
\newblock \doi{10.1103/PhysRevA.97.062337}.
\newblock URL \url{https://link.aps.org/doi/10.1103/PhysRevA.97.062337}.

\bibitem[Lami et~al.(2018)Lami, Regula, Wang, Nichols, Winter, and
  Adesso]{Lami2018}
Ludovico Lami, Bartosz Regula, Xin Wang, Rosanna Nichols, Andreas Winter, and
  Gerardo Adesso.
\newblock {Gaussian} quantum resource theories.
\newblock \emph{Phys. Rev. A}, 98:\penalty0 022335, Aug 2018.
\newblock \doi{10.1103/PhysRevA.98.022335}.
\newblock URL \url{https://link.aps.org/doi/10.1103/PhysRevA.98.022335}.

\bibitem[Albarelli et~al.(2018)Albarelli, Genoni, Paris, and
  Ferraro]{Albarelli2018}
Francesco Albarelli, Marco~G. Genoni, Matteo G.~A. Paris, and Alessandro
  Ferraro.
\newblock Resource theory of quantum non-{Gaussianity} and wigner negativity.
\newblock \emph{Phys. Rev. A}, 98:\penalty0 052350, Nov 2018.
\newblock \doi{10.1103/PhysRevA.98.052350}.
\newblock URL \url{https://link.aps.org/doi/10.1103/PhysRevA.98.052350}.

\bibitem[Yamasaki et~al.(2020)Yamasaki, Matsuura, and Koashi]{Yamasaki2019}
Hayata Yamasaki, Takaya Matsuura, and Masato Koashi.
\newblock Cost-reduced all-gaussian universality with the
  gottesman-kitaev-preskill code: Resource-theoretic approach to cost analysis.
\newblock \emph{Phys. Rev. Research}, 2:\penalty0 023270, Jun 2020.
\newblock \doi{10.1103/PhysRevResearch.2.023270}.
\newblock URL \url{https://link.aps.org/doi/10.1103/PhysRevResearch.2.023270}.

\bibitem[Wakakuwa(2017)]{Wakakuwa2017}
Eyuri Wakakuwa.
\newblock Operational resource theory of non-markovianity.
\newblock \emph{arXiv preprint
  \href{https://arxiv.org/abs/1709.07248}{arXiv:1709.07248 [quant-ph]}}, 2017.
\newblock URL \url{https://arxiv.org/abs/1709.07248}.

\bibitem[Wakakuwa(2019)]{Wakakuwa2019}
Eyuri Wakakuwa.
\newblock Communication cost for non-markovianity of tripartite quantum states:
  A resource theoretic approach.
\newblock \emph{arXiv preprint
  \href{https://arxiv.org/abs/1904.08852}{arXiv:1904.08852 [quant-ph]}}, 2019.
\newblock URL \url{https://arxiv.org/abs/1904.08852}.

\bibitem[Horodecki and Oppenheim(2013{\natexlab{b}})]{Horodecki2013b}
Michal Horodecki and Jonathan Oppenheim.
\newblock {(Quantumness in the context of) resource theories}.
\newblock \emph{Int. J. Mod. Phys. B}, 27\penalty0 (1-3):\penalty0 1345019,
  2013{\natexlab{b}}.
\newblock \doi{10.1142/S0217979213450197}.
\newblock URL \url{https://doi.org/10.1142/S0217979213450197}.

\bibitem[Brand{\~a}o and Gour(2015)]{Brandao2015}
Fernando G. S.~L. Brand{\~a}o and Gilad Gour.
\newblock Reversible framework for quantum resource theories.
\newblock \emph{Phys. Rev. Lett.}, 115:\penalty0 070503, Aug 2015.
\newblock \doi{10.1103/PhysRevLett.115.070503}.
\newblock URL \url{https://link.aps.org/doi/10.1103/PhysRevLett.115.070503}.

\bibitem[Korzekwa et~al.(2019)Korzekwa, Chubb, and Tomamichel]{Korzekwa2019}
Kamil Korzekwa, Christopher~T. Chubb, and Marco Tomamichel.
\newblock Avoiding irreversibility: Engineering resonant conversions of quantum
  resources.
\newblock \emph{Phys. Rev. Lett.}, 122:\penalty0 110403, Mar 2019.
\newblock \doi{10.1103/PhysRevLett.122.110403}.
\newblock URL \url{https://link.aps.org/doi/10.1103/PhysRevLett.122.110403}.

\bibitem[Liu et~al.(2019{\natexlab{a}})Liu, Yu, and Tong]{Liu_CL2019}
C.~L. Liu, Xiao-Dong Yu, and D.~M. Tong.
\newblock Flag additivity in quantum resource theories.
\newblock \emph{Phys. Rev. A}, 99:\penalty0 042322, Apr 2019{\natexlab{a}}.
\newblock \doi{10.1103/PhysRevA.99.042322}.
\newblock URL \url{https://link.aps.org/doi/10.1103/PhysRevA.99.042322}.

\bibitem[Liu et~al.(2019{\natexlab{b}})Liu, Bu, and Takagi]{Liu_ZW2019a}
Zi-Wen Liu, Kaifeng Bu, and Ryuji Takagi.
\newblock One-shot operational quantum resource theory.
\newblock \emph{Phys. Rev. Lett.}, 123:\penalty0 020401, Jul
  2019{\natexlab{b}}.
\newblock \doi{10.1103/PhysRevLett.123.020401}.
\newblock URL \url{https://link.aps.org/doi/10.1103/PhysRevLett.123.020401}.

\bibitem[Takagi et~al.(2019)Takagi, Regula, Bu, Liu, and Adesso]{Takagi2019b}
Ryuji Takagi, Bartosz Regula, Kaifeng Bu, Zi-Wen Liu, and Gerardo Adesso.
\newblock Operational advantage of quantum resources in subchannel
  discrimination.
\newblock \emph{Phys. Rev. Lett.}, 122:\penalty0 140402, Apr 2019.
\newblock \doi{10.1103/PhysRevLett.122.140402}.
\newblock URL \url{https://link.aps.org/doi/10.1103/PhysRevLett.122.140402}.

\bibitem[Takagi and Regula(2019)]{Takagi2019a}
Ryuji Takagi and Bartosz Regula.
\newblock General resource theories in quantum mechanics and beyond:
  Operational characterization via discrimination tasks.
\newblock \emph{Phys. Rev. X}, 9:\penalty0 031053, Sep 2019.
\newblock \doi{10.1103/PhysRevX.9.031053}.
\newblock URL \url{https://link.aps.org/doi/10.1103/PhysRevX.9.031053}.

\bibitem[Fang and Liu(2020)]{Fang2020}
Kun Fang and Zi-Wen Liu.
\newblock No-go theorems for quantum resource purification.
\newblock \emph{Phys. Rev. Lett.}, 125:\penalty0 060405, Aug 2020.
\newblock \doi{10.1103/PhysRevLett.125.060405}.
\newblock URL \url{https://link.aps.org/doi/10.1103/PhysRevLett.125.060405}.

\bibitem[Regula et~al.(2020)Regula, Bu, Takagi, and Liu]{Regula2020}
Bartosz Regula, Kaifeng Bu, Ryuji Takagi, and Zi-Wen Liu.
\newblock Benchmarking one-shot distillation in general quantum resource
  theories.
\newblock \emph{Phys. Rev. A}, 101:\penalty0 062315, Jun 2020.
\newblock \doi{10.1103/PhysRevA.101.062315}.
\newblock URL \url{https://link.aps.org/doi/10.1103/PhysRevA.101.062315}.

\bibitem[Vijayan et~al.(2019)Vijayan, Chitambar, and Hsieh]{Vijayan2019}
Madhav~Krishnan Vijayan, Eric Chitambar, and Min-Hsiu Hsieh.
\newblock One-shot distillation in a general resource theory.
\newblock \emph{arXiv preprint
  \href{https://arxiv.org/abs/1906.04959}{arXiv:1906.04959 [quant-ph]}}, 2019.
\newblock URL \url{https://arxiv.org/abs/1906.04959}.

\bibitem[Eisert and Wolf()]{Eisert2005}
J.~Eisert and M.~M. Wolf.
\newblock \emph{{Gaussian} Quantum Channels}, pages 23--42.
\newblock \doi{10.1142/9781860948169_0002}.
\newblock URL
  \url{https://www.worldscientific.com/doi/abs/10.1142/9781860948169_0002}.

\bibitem[Holevo(2007)]{Holevo2007}
A.~S. Holevo.
\newblock One-mode quantum {Gaussian} channels: Structure and quantum capacity.
\newblock \emph{Probl. Inf. Transm.}, 43, Mar 2007.
\newblock \doi{10.1134/S0032946007010012}.
\newblock URL \url{https://doi.org/10.1134/S0032946007010012}.

\bibitem[Bera et~al.(2017)Bera, Das, Sadhukhan, Roy, Sen(De), and
  Sen]{Bera2017}
Anindita Bera, Tamoghna Das, Debasis Sadhukhan, Sudipto~Singha Roy, Aditi
  Sen(De), and Ujjwal Sen.
\newblock Quantum discord and its allies: a review of recent progress.
\newblock \emph{Rep. Prog. Phys.}, 81\penalty0 (2):\penalty0 024001, dec 2017.
\newblock \doi{10.1088/1361-6633/aa872f}.
\newblock URL \url{https://doi.org/10.1088%2F1361-6633%2Faa872f}.

\bibitem[Gudder(2008)]{Gudder2008}
Stanley Gudder.
\newblock Quantum markov chains.
\newblock \emph{J. Math. Phys.}, 49\penalty0 (7):\penalty0 072105, 2008.
\newblock \doi{10.1063/1.2953952}.
\newblock URL \url{https://doi.org/10.1063/1.2953952}.

\bibitem[Jonathan and Plenio(1999)]{Jonathan1999}
Daniel Jonathan and Martin~B. Plenio.
\newblock Entanglement-assisted local manipulation of pure quantum states.
\newblock \emph{Phys. Rev. Lett.}, 83:\penalty0 3566--3569, Oct 1999.
\newblock \doi{10.1103/PhysRevLett.83.3566}.
\newblock URL \url{https://link.aps.org/doi/10.1103/PhysRevLett.83.3566}.

\bibitem[Bennett et~al.(1996{\natexlab{a}})Bennett, DiVincenzo, Smolin, and
  Wootters]{Bennett1996}
Charles~H. Bennett, David~P. DiVincenzo, John~A. Smolin, and William~K.
  Wootters.
\newblock Mixed-state entanglement and quantum error correction.
\newblock \emph{Phys. Rev. A}, 54:\penalty0 3824--3851, Nov 1996{\natexlab{a}}.
\newblock \doi{10.1103/PhysRevA.54.3824}.
\newblock URL \url{https://link.aps.org/doi/10.1103/PhysRevA.54.3824}.

\bibitem[Hayden et~al.(2001)Hayden, Horodecki, and Terhal]{Hayden2001}
Patrick~M Hayden, Michal Horodecki, and Barbara~M Terhal.
\newblock The asymptotic entanglement cost of preparing a quantum state.
\newblock \emph{J. Phys. A: Math. Gen.}, 34\penalty0 (35):\penalty0 6891--6898,
  aug 2001.
\newblock \doi{10.1088/0305-4470/34/35/314}.
\newblock URL \url{https://doi.org/10.1088%2F0305-4470%2F34%2F35%2F314}.

\bibitem[Winter and Yang(2016)]{Winter2016}
Andreas Winter and Dong Yang.
\newblock Operational resource theory of coherence.
\newblock \emph{Phys. Rev. Lett.}, 116:\penalty0 120404, Mar 2016.
\newblock \doi{10.1103/PhysRevLett.116.120404}.
\newblock URL \url{https://link.aps.org/doi/10.1103/PhysRevLett.116.120404}.

\bibitem[Donald et~al.(2002)Donald, Horodecki, and Rudolph]{Donald2002}
Matthew~J. Donald, Michał Horodecki, and Oliver Rudolph.
\newblock The uniqueness theorem for entanglement measures.
\newblock \emph{J. Math. Phys.}, 43\penalty0 (9):\penalty0 4252--4272, 2002.
\newblock \doi{10.1063/1.1495917}.
\newblock URL \url{https://doi.org/10.1063/1.1495917}.

\bibitem[Reed and Simon(1980)]{ReedandSimon1981}
Michael Reed and Barry Simon.
\newblock \emph{Methods of modern mathematical physics. vol. 1. Functional
  analysis}.
\newblock Academic New York, 1980.

\bibitem[Takesaki(2012)]{Takesaki2002}
Masamichi Takesaki.
\newblock \emph{Theory of Operator Algebras I}.
\newblock Springer, 2012.
\newblock ISBN 0387903917;9780387903910;.
\newblock \doi{10.1007/978-1-4612-6188-9}.
\newblock URL \url{https://doi.org/10.1007/978-1-4612-6188-9}.

\bibitem[Kitaev(1997)]{Kitaev1997}
A~Yu Kitaev.
\newblock Quantum computations: algorithms and error correction.
\newblock \emph{Russ. Math. Surv.}, 52\penalty0 (6):\penalty0 1191--1249, dec
  1997.
\newblock \doi{10.1070/rm1997v052n06abeh002155}.
\newblock URL \url{https://doi.org/10.1070%2Frm1997v052n06abeh002155}.

\bibitem[Brand{\~a}o and Plenio(2008)]{Brandao2008}
Fernando G. S.~L. Brand{\~a}o and Martin Plenio.
\newblock Entanglement theory and the second law of thermodynamics.
\newblock \emph{Nat. Phys.}, 4:\penalty0 873–877, 2008.
\newblock \doi{10.1038/nphys1100}.
\newblock URL \url{https://doi.org/10.1038/nphys1100}.

\bibitem[Contreras-Tejada et~al.(2019)Contreras-Tejada, Palazuelos, and
  de~Vicente]{Contreras-Tejada2019}
Patricia Contreras-Tejada, Carlos Palazuelos, and Julio~I. de~Vicente.
\newblock Resource theory of entanglement with a unique multipartite maximally
  entangled state.
\newblock \emph{Phys. Rev. Lett.}, 122:\penalty0 120503, Mar 2019.
\newblock \doi{10.1103/PhysRevLett.122.120503}.
\newblock URL \url{https://link.aps.org/doi/10.1103/PhysRevLett.122.120503}.

\bibitem[Chitambar et~al.(2014)Chitambar, Leung, Man{\v{c}}inska, Ozols, and
  Winter]{Chitambar2014}
Eric Chitambar, Debbie Leung, Laura Man{\v{c}}inska, Maris Ozols, and Andreas
  Winter.
\newblock Everything you always wanted to know about locc (but were afraid to
  ask).
\newblock \emph{Commun. Math. Phys.}, 328\penalty0 (1):\penalty0 303--326, May
  2014.
\newblock ISSN 1432-0916.
\newblock \doi{10.1007/s00220-014-1953-9}.
\newblock URL \url{https://doi.org/10.1007/s00220-014-1953-9}.

\bibitem[Nielsen and Chuang(2010)]{nielsen_chuang_2010}
Michael~A. Nielsen and Isaac~L. Chuang.
\newblock \emph{Quantum Computation and Quantum Information: 10th Anniversary
  Edition}.
\newblock Cambridge University Press, 2010.
\newblock \doi{10.1017/CBO9780511976667}.
\newblock URL \url{https://doi.org/10.1017/CBO9780511976667}.

\bibitem[Groisman et~al.(2005)Groisman, Popescu, and Winter]{Groisman2005}
Berry Groisman, Sandu Popescu, and Andreas Winter.
\newblock Quantum, classical, and total amount of correlations in a quantum
  state.
\newblock \emph{Phys. Rev. A}, 72:\penalty0 032317, Sep 2005.
\newblock \doi{10.1103/PhysRevA.72.032317}.
\newblock URL \url{https://link.aps.org/doi/10.1103/PhysRevA.72.032317}.

\bibitem[Anshu et~al.(2018)Anshu, Hsieh, and Jain]{Anshu2018}
Anurag Anshu, Min-Hsiu Hsieh, and Rahul Jain.
\newblock Quantifying resources in general resource theory with catalysts.
\newblock \emph{Phys. Rev. Lett.}, 121:\penalty0 190504, Nov 2018.
\newblock \doi{10.1103/PhysRevLett.121.190504}.
\newblock URL \url{https://link.aps.org/doi/10.1103/PhysRevLett.121.190504}.

\bibitem[Pirandola et~al.(2019)Pirandola, Andersen, Banchi, Berta, Bunandar,
  Colbeck, Englund, Gehring, Lupo, Ottaviani, Pereira, Razavi, Shaari,
  Tomamichel, Usenko, Vallone, Villoresi, and Wallden]{Pir2019}
S.~Pirandola, U.~L. Andersen, L.~Banchi, M.~Berta, D.~Bunandar, R.~Colbeck,
  D.~Englund, T.~Gehring, C.~Lupo, C.~Ottaviani, J.~Pereira, M.~Razavi, J.~S.
  Shaari, M.~Tomamichel, V.~C. Usenko, G.~Vallone, P.~Villoresi, and
  P.~Wallden.
\newblock Advances in quantum cryptography.
\newblock \emph{arXiv preprint
  \href{https://arxiv.org/abs/1906.01645}{arXiv:1906.01645 [quant-ph]}}, 2019.
\newblock URL \url{https://arxiv.org/abs/1906.01645}.

\bibitem[Preskill(2018)]{Preskill2018}
John Preskill.
\newblock Quantum computing in the {NISQ} era and beyond.
\newblock \emph{Quantum}, 2:\penalty0 79, Aug 2018.
\newblock ISSN 2521-327X.
\newblock \doi{10.22331/q-2018-08-06-79}.
\newblock URL \url{http://dx.doi.org/10.22331/q-2018-08-06-79}.

\bibitem[Bravyi and Haah(2012)]{Bravyi2012}
Sergey Bravyi and Jeongwan Haah.
\newblock Magic-state distillation with low overhead.
\newblock \emph{Phys. Rev. A}, 86:\penalty0 052329, Nov 2012.
\newblock \doi{10.1103/PhysRevA.86.052329}.
\newblock URL \url{https://link.aps.org/doi/10.1103/PhysRevA.86.052329}.

\bibitem[Berger et~al.(2003)Berger, Gierz, Hofmann, Keimel, Lawson, Mislove,
  and Scott]{Berger2003}
Ulrich Berger, G.~Gierz, K.~H. Hofmann, K.~Keimel, J.~D. Lawson, M.~W. Mislove,
  and D.~S. Scott.
\newblock \emph{{Continuous Lattices and Domains}}, volume~93.
\newblock Cambridge University Press, 2003.
\newblock \doi{10.1017/CBO9780511542725}.
\newblock URL \url{https://doi.org/10.1017/CBO9780511542725}.

\bibitem[Gottesman et~al.(2001)Gottesman, Kitaev, and Preskill]{Gottesman2001}
Daniel Gottesman, Alexei Kitaev, and John Preskill.
\newblock Encoding a qubit in an oscillator.
\newblock \emph{Phys. Rev. A}, 64:\penalty0 012310, Jun 2001.
\newblock \doi{10.1103/PhysRevA.64.012310}.
\newblock URL \url{https://link.aps.org/doi/10.1103/PhysRevA.64.012310}.

\bibitem[Pantaleoni et~al.(2020)Pantaleoni, Baragiola, and
  Menicucci]{Pantaleoni2019}
Giacomo Pantaleoni, Ben~Q. Baragiola, and Nicolas~C. Menicucci.
\newblock Modular bosonic subsystem codes.
\newblock \emph{Phys. Rev. Lett.}, 125:\penalty0 040501, Jul 2020.
\newblock \doi{10.1103/PhysRevLett.125.040501}.
\newblock URL \url{https://link.aps.org/doi/10.1103/PhysRevLett.125.040501}.

\bibitem[Weedbrook et~al.(2012)Weedbrook, Pirandola, Garc\'{\i}a-Patr\'on,
  Cerf, Ralph, Shapiro, and Lloyd]{Weedbrook2012}
Christian Weedbrook, Stefano Pirandola, Ra\'ul Garc\'{\i}a-Patr\'on, Nicolas~J.
  Cerf, Timothy~C. Ralph, Jeffrey~H. Shapiro, and Seth Lloyd.
\newblock {Gaussian} quantum information.
\newblock \emph{Rev. Mod. Phys.}, 84:\penalty0 621--669, May 2012.
\newblock \doi{10.1103/RevModPhys.84.621}.
\newblock URL \url{https://link.aps.org/doi/10.1103/RevModPhys.84.621}.

\bibitem[Horodecki et~al.(2003)Horodecki, Sen(De), and Sen]{Horodecki2003}
Micha\l{} Horodecki, Aditi Sen(De), and Ujjwal Sen.
\newblock Rates of asymptotic entanglement transformations for bipartite mixed
  states: Maximally entangled states are not special.
\newblock \emph{Phys. Rev. A}, 67:\penalty0 062314, Jun 2003.
\newblock \doi{10.1103/PhysRevA.67.062314}.
\newblock URL \url{https://link.aps.org/doi/10.1103/PhysRevA.67.062314}.

\bibitem[Horodecki et~al.(1998)Horodecki, Horodecki, and
  Horodecki]{Horodecki1998}
Micha\l{} Horodecki, Pawe\l{} Horodecki, and Ryszard Horodecki.
\newblock Mixed-state entanglement and distillation: Is there a ``bound''
  entanglement in nature?
\newblock \emph{Phys. Rev. Lett.}, 80:\penalty0 5239--5242, Jun 1998.
\newblock \doi{10.1103/PhysRevLett.80.5239}.
\newblock URL \url{https://link.aps.org/doi/10.1103/PhysRevLett.80.5239}.

\bibitem[Bennett et~al.(1996{\natexlab{b}})Bennett, Bernstein, Popescu, and
  Schumacher]{Bernstein1996}
Charles~H. Bennett, Herbert~J. Bernstein, Sandu Popescu, and Benjamin
  Schumacher.
\newblock Concentrating partial entanglement by local operations.
\newblock \emph{Phys. Rev. A}, 53:\penalty0 2046--2052, Apr 1996{\natexlab{b}}.
\newblock \doi{10.1103/PhysRevA.53.2046}.
\newblock URL \url{https://link.aps.org/doi/10.1103/PhysRevA.53.2046}.

\bibitem[Wilde(2017)]{Wilde2017}
Mark~M. Wilde.
\newblock \emph{Entanglement Manipulation}, page 517–536.
\newblock Cambridge University Press, 2 edition, 2017.
\newblock \doi{10.1017/9781316809976.022}.
\newblock URL \url{https://doi.org/10.1017/9781316809976.022}.

\bibitem[Wootters and Zurek(1982)]{Wootters1982}
W.~K. Wootters and W.~H. Zurek.
\newblock A single quantum cannot be cloned.
\newblock \emph{Nature}, 299:\penalty0 802--803, Oct 1982.
\newblock \doi{10.1038/299802a0}.
\newblock URL \url{https://doi.org/10.1038/299802a0}.

\bibitem[Wang et~al.(2020)Wang, Wilde, and Su]{Wang2018}
Xin Wang, Mark~M. Wilde, and Yuan Su.
\newblock Efficiently computable bounds for magic state distillation.
\newblock \emph{Phys. Rev. Lett.}, 124:\penalty0 090505, Mar 2020.
\newblock \doi{10.1103/PhysRevLett.124.090505}.
\newblock URL \url{https://link.aps.org/doi/10.1103/PhysRevLett.124.090505}.

\bibitem[Plbnio and Virmani(2007)]{Plenio2005}
Martin~B. Plbnio and Shashank Virmani.
\newblock An introduction to entanglement measures.
\newblock \emph{Quantum Info. Comput.}, 7\penalty0 (1):\penalty0 1–51, jan
  2007.
\newblock ISSN 1533-7146.
\newblock \doi{10.5555/2011706.2011707}.
\newblock URL \url{https://dl.acm.org/doi/10.5555/2011706.2011707}.

\bibitem[Bromley et~al.(2018)Bromley, Cianciaruso, Vourekas, Regula, and
  Adesso]{Bromley2018}
Thomas~R Bromley, Marco Cianciaruso, Sofoklis Vourekas, Bartosz Regula, and
  Gerardo Adesso.
\newblock Accessible bounds for general quantum resources.
\newblock \emph{J. Phys. A: Math. Theor.}, 51\penalty0 (32):\penalty0 325303,
  jul 2018.
\newblock \doi{10.1088/1751-8121/aacb4a}.
\newblock URL \url{https://doi.org/10.1088%2F1751-8121%2Faacb4a}.

\bibitem[Liu et~al.(2017)Liu, Hu, and Lloyd]{Liu2017}
Zi-Wen Liu, Xueyuan Hu, and Seth Lloyd.
\newblock Resource destroying maps.
\newblock \emph{Phys. Rev. Lett.}, 118:\penalty0 060502, Feb 2017.
\newblock \doi{10.1103/PhysRevLett.118.060502}.
\newblock URL \url{https://link.aps.org/doi/10.1103/PhysRevLett.118.060502}.

\bibitem[Regula(2017)]{Regula2017}
Bartosz Regula.
\newblock Convex geometry of quantum resource quantification.
\newblock \emph{J. Phys. A: Math. Theor.}, 51\penalty0 (4):\penalty0 045303,
  dec 2017.
\newblock \doi{10.1088/1751-8121/aa9100}.
\newblock URL \url{https://doi.org/10.1088%2F1751-8121%2Faa9100}.

\bibitem[Synak-Radtke and Horodecki(2006)]{Synak2006}
Barbara Synak-Radtke and Micha{\l} Horodecki.
\newblock On asymptotic continuity of functions of quantum states.
\newblock \emph{J. Phys. A: Math. Gen.}, 39\penalty0 (26):\penalty0 L423--L437,
  jun 2006.
\newblock \doi{10.1088/0305-4470/39/26/l02}.
\newblock URL \url{https://doi.org/10.1088%2F0305-4470%2F39%2F26%2Fl02}.

\bibitem[{Lami}(2019)]{Lami2019}
L.~{Lami}.
\newblock Completing the grand tour of asymptotic quantum coherence
  manipulation.
\newblock \emph{IEEE Trans. Inf. Theory}, pages 1--1, 2019.
\newblock ISSN 1557-9654.
\newblock \doi{10.1109/TIT.2019.2945798}.
\newblock URL \url{https://doi.org/10.1109/TIT.2019.2945798}.

\bibitem[Vedral et~al.(1997)Vedral, Plenio, Rippin, and Knight]{Vedral1997}
V.~Vedral, M.~B. Plenio, M.~A. Rippin, and P.~L. Knight.
\newblock Quantifying entanglement.
\newblock \emph{Phys. Rev. Lett.}, 78:\penalty0 2275--2279, Mar 1997.
\newblock \doi{10.1103/PhysRevLett.78.2275}.
\newblock URL \url{https://link.aps.org/doi/10.1103/PhysRevLett.78.2275}.

\bibitem[Brand{\~a}o and Plenio(2010)]{Brandao2010}
Fernando G. S.~L. Brand{\~a}o and Martin~B. Plenio.
\newblock A generalization of quantum {Stein}'s lemma.
\newblock \emph{Commun. Math. Phys.}, 295\penalty0 (3):\penalty0 791--828, May
  2010.
\newblock ISSN 1432-0916.
\newblock \doi{10.1007/s00220-010-1005-z}.
\newblock URL \url{https://doi.org/10.1007/s00220-010-1005-z}.

\bibitem[Horodecki et~al.(2002)Horodecki, Oppenheim, and
  Horodecki]{Horodecki2002}
Micha\l{} Horodecki, Jonathan Oppenheim, and Ryszard Horodecki.
\newblock Are the laws of entanglement theory thermodynamical?
\newblock \emph{Phys. Rev. Lett.}, 89:\penalty0 240403, Nov 2002.
\newblock \doi{10.1103/PhysRevLett.89.240403}.
\newblock URL \url{https://link.aps.org/doi/10.1103/PhysRevLett.89.240403}.

\bibitem[Weis and Knauf(2012)]{Weis2012}
Stephan Weis and Andreas Knauf.
\newblock Entropy distance: New quantum phenomena.
\newblock \emph{J. Math. Phys.}, 53\penalty0 (10):\penalty0 102206, Oct 2012.
\newblock ISSN 1089-7658.
\newblock \doi{10.1063/1.4757652}.
\newblock URL \url{http://dx.doi.org/10.1063/1.4757652}.

\bibitem[Herrero-Collantes and Garcia-Escartin(2017)]{Miguel2017}
Miguel Herrero-Collantes and Juan~Carlos Garcia-Escartin.
\newblock Quantum random number generators.
\newblock \emph{Rev. Mod. Phys.}, 89:\penalty0 015004, Feb 2017.
\newblock \doi{10.1103/RevModPhys.89.015004}.
\newblock URL \url{https://link.aps.org/doi/10.1103/RevModPhys.89.015004}.

\bibitem[Dankert et~al.(2009)Dankert, Cleve, Emerson, and Livine]{Cleve2009}
Christoph Dankert, Richard Cleve, Joseph Emerson, and Etera Livine.
\newblock Exact and approximate unitary 2-designs and their application to
  fidelity estimation.
\newblock \emph{Phys. Rev. A}, 80:\penalty0 012304, Jul 2009.
\newblock \doi{10.1103/PhysRevA.80.012304}.
\newblock URL \url{https://link.aps.org/doi/10.1103/PhysRevA.80.012304}.

\bibitem[Theurer et~al.(2019)Theurer, Egloff, Zhang, and Plenio]{Theurer2019}
Thomas Theurer, Dario Egloff, Lijian Zhang, and Martin~B. Plenio.
\newblock Quantifying operations with an application to coherence.
\newblock \emph{Phys. Rev. Lett.}, 122:\penalty0 190405, May 2019.
\newblock \doi{10.1103/PhysRevLett.122.190405}.
\newblock URL \url{https://link.aps.org/doi/10.1103/PhysRevLett.122.190405}.

\bibitem[Gour and Winter(2019)]{Gour2019a}
Gilad Gour and Andreas Winter.
\newblock How to quantify a dynamical quantum resource.
\newblock \emph{Phys. Rev. Lett.}, 123:\penalty0 150401, Oct 2019.
\newblock \doi{10.1103/PhysRevLett.123.150401}.
\newblock URL \url{https://link.aps.org/doi/10.1103/PhysRevLett.123.150401}.

\bibitem[Liu and Yuan(2020)]{Liu_YC2020}
Yunchao Liu and Xiao Yuan.
\newblock Operational resource theory of quantum channels.
\newblock \emph{Phys. Rev. Research}, 2:\penalty0 012035, Feb 2020.
\newblock \doi{10.1103/PhysRevResearch.2.012035}.
\newblock URL \url{https://link.aps.org/doi/10.1103/PhysRevResearch.2.012035}.

\bibitem[Gour and Wilde(2018)]{Gour_Wilde2018}
Gilad Gour and Mark~M. Wilde.
\newblock Entropy of a quantum channel.
\newblock \emph{arXiv preprint
  \href{https://arxiv.org/abs/1808.06980}{arXiv:1808.06980 [quant-ph]}}, 2018.
\newblock URL \url{https://arxiv.org/abs/1808.06980}.

\bibitem[Li et~al.(2020)Li, Bu, and Liu]{Li2018}
Lu~Li, Kaifeng Bu, and Zi-Wen Liu.
\newblock Quantifying the resource content of quantum channels: An operational
  approach.
\newblock \emph{Phys. Rev. A}, 101:\penalty0 022335, Feb 2020.
\newblock \doi{10.1103/PhysRevA.101.022335}.
\newblock URL \url{https://link.aps.org/doi/10.1103/PhysRevA.101.022335}.

\bibitem[Liu and Winter(2019)]{Liu_ZW2019b}
Zi-Wen Liu and Andreas Winter.
\newblock Resource theories of quantum channels and the universal role of
  resource erasure.
\newblock \emph{arXiv preprint
  \href{https://arxiv.org/abs/1904.04201}{arXiv:1904.04201 [quant-ph]}}, 2019.
\newblock URL \url{https://arxiv.org/abs/1904.04201}.

\bibitem[Takagi et~al.(2020)Takagi, Wang, and Hayashi]{Takagi2020}
Ryuji Takagi, Kun Wang, and Masahito Hayashi.
\newblock Application of the resource theory of channels to communication
  scenarios.
\newblock \emph{Phys. Rev. Lett.}, 124:\penalty0 120502, Mar 2020.
\newblock \doi{10.1103/PhysRevLett.124.120502}.
\newblock URL \url{https://link.aps.org/doi/10.1103/PhysRevLett.124.120502}.

\bibitem[Dunford and Schwartz(1958)]{dunfordschwartzvol1}
Nelson Dunford and Jacob~T Schwartz.
\newblock \emph{Linear operators. Part I, General theory}.
\newblock New York : Interscience., 1958.

\bibitem[Dell'Antonio(1967)]{Dell'Antonio1967}
G.~F. Dell'Antonio.
\newblock On the limits of sequences of normal states.
\newblock \emph{Commun. Pure Appl. Math.}, 20\penalty0 (2):\penalty0 413--429,
  1967.
\newblock \doi{10.1002/cpa.3160200209}.
\newblock URL \url{https://doi.org/10.1002/cpa.3160200209}.

\end{thebibliography}

\end{document}